\newif\iflong
\longtrue

\newif\iffuture
\futurefalse
\documentclass[acmsmall,screen]{acmart}

\setcopyright{rightsretained}
\acmPrice{}
\acmDOI{10.1145/3498676}
\acmYear{2022}
\copyrightyear{2022}
\acmSubmissionID{popl22main-p65-p}
\acmJournal{PACMPL}
\acmVolume{6}
\acmNumber{POPL}
\acmArticle{15}
\acmMonth{1}

\usepackage{listings}
\lstdefinelanguage{program}{%
  keywords={%
    let,pass,function,%
    var,const,bool,int,void,atomic,%
    while,do,if,then,else,assume,assert,call,return,rule,forall,with,new,choose,skip,%
    task,async,yield,for,wait,%
    type,relation,init, action, safety, invariant, axiom, input,repeat
  },
  morecomment=[l]{//},
  morecomment=[s]{/*}{*/},
  morecomment=[n]{(*}{*)},
  mathescape=true,
  escapeinside=`',
}
\usepackage{booktabs}
\usepackage{mathpartir}
\usepackage{multirow}
\usepackage{amsmath}
\usepackage{amsthm}
\usepackage{stmaryrd}
\usepackage{array}
\usepackage{subcaption}
\usepackage[T1]{fontenc}
\usepackage[nodayofweek]{datetime}
\usepackage{graphicx}
\usepackage{algorithm}
\usepackage{wrapfig}
\usepackage[noend]{algpseudocode}
\usepackage[flushleft]{threeparttable}
\usepackage{stmaryrd}
\usepackage{multicol}
\usepackage{pifont}
\usepackage[xcolor]{changebar}
\usepackage{caption}
\cbcolor{magenta}
\usepackage{tcolorbox}
\tcbuselibrary{theorems}
\usepackage{centernot}
\usepackage{galois}
\usepackage{musicography}
\usepackage{scalerel}

\usepackage{xcolor}
\usepackage{extarrows}
\usepackage{mathtools}
\usepackage{paralist}

\usepackage{tikz}
\usetikzlibrary{positioning,fit,calc,arrows.meta,shapes.geometric,decorations.pathmorphing}
\makeatletter
\newcommand\currentcoordinate{\the\tikz@lastxsaved,\the\tikz@lastysaved}
\makeatother

\usepackage[noabbrev,capitalise]{cleveref} %

\makeatletter
\newcounter{algorithmicH}%
\let\oldalgorithmic\algorithmic
\renewcommand{\algorithmic}{%
  \stepcounter{algorithmicH}%
  \oldalgorithmic}%
\renewcommand{\theHALG@line}{ALG@line.\thealgorithmicH.\arabic{ALG@line}}
\makeatother

\newif\ifnitpick
\nitpickfalse

\newif\ifproofs
\proofstrue

\iflong
\newcommand{\refappendix}[1]{\Cref{#1}}
\else
\newcommand{\refappendix}[1]{the extended version~\cite{extendedVersion}}
\fi

\iflong
\newcommand{\toolong}[1]{#1}
\else
\newcommand{\toolong}[1]{}
\fi

\crefformat{section}{\S#2#1#3} %
\crefformat{subsection}{\S#2#1#3}
\crefformat{subsubsection}{\S#2#1#3}

\Crefname{conjecture}{Conjecture}{Conjectures}
\Crefname{proposition}{Proposition}{Propositions}
\Crefname{lemma}{Lemma}{Lemmas}
\Crefname{corollary}{Corollary}{Corollaries}
\Crefname{example}{Example}{Examples}
\Crefname{definition}{Def.}{Defs.}
\Crefname{algorithm}{Alg.}{Alg.}
\Crefname{theorem}{Thm.}{Thms.}
\Crefname{figure}{Fig.}{Fig.}
\crefname{line}{line}{lines}

\ifproofs
\else
\excludecomment{proof}
\fi

\newtheorem{theo}{Theo}[section] %

\newtheorem{remark}[theo]{Remark}
\newcommand{\para}[1]{\vspace{2pt}\noindent\textbf{\textit{#1.}}}

\newcommand{\dom}[1]{dom({#1})}
\newcommand{\ov}{\overline}

\newcommand{\card}[1]{{\left\vert{#1}\right\vert}} %

\renewcommand{\implies}{\Longrightarrow}

\newcommand{\notimplies}{\centernot\implies}

\newcommand{\true}{{\textit{true}}}
\newcommand{\false}{{\textit{false}}}
\newcommand{\vocabulary}{\Sigma}
\newcommand{\voc}{\vocabulary}

\newcommand{\Init}{{\textit{Init}}}
\newcommand{\Bad}{\textit{Bad}}

\newcommand{\States}{{\mbox{States}}[\voc]}

\newcommand{\tr}{\delta}

\newcommand{\Frame}{\mathcal{F}}
\newcommand{\Framepdr}{\mathcal{F}^{\text{pdr}}}
\newcommand{\Frameai}{\mathcal{F}^{\text{ai}}}

\newcommand{\psigma}[1]{\Sigma_{#1}^{P}}

\renewcommand{\vec}{\ov}

\newcommand{\set}[1]{\{{#1}\}}

\newcommand{\eqdef}{\stackrel{\rm def}{=}}

\newcommand{\prop}{x}

\newcommand{\reflextr}[1]{\underline{#1}}
\newcommand{\postimage}[2]{{\reflextr{#1}}({#2})}
\newcommand{\bmcunroll}[2]{{\reflextr{#1}}^{{#2}}}
\newcommand{\bmc}[3]{\bmcunroll{#1}{#3}({#2})}

\newcommand{\bkwrch}[1]{\mathcal{B}_{#1}}

\newcommand{\sleq}{\sqsubseteq}

\algdef{SE}[DOWHILE]{Do}{doWhile}{\algorithmicdo}[1]{\algorithmicwhile\ #1}%

\newcommand{\dnfsize}[1]{\card{#1}_{\rm dnf}}

\newcommand{\cubemon}[2]{\textit{cube}_{{#2}}({#1})}
\newcommand{\moncube}[2]{\cubemon{#1}{#2}}
\newcommand{\monox}[2]{\mathcal{M}_{#2}({#1})}

\newcommand{\boundarypos}[1]{\partial^{+}({#1})}

\newcommand{\bigO}{O}

\newcommand{\mspan}[1]{{\rm MSpan}({#1})}
\newcommand{\bkwspan}[1]{\mspan{\bkwrch{#1}}}
\newcommand{\mhull}[2]{{\rm MHull}_{#2}({#1})}

\newcommand{\abs}[1]{{#1}^{\sharp}}
\newcommand{\absr}[1]{{#1}^{\musDoubleSharp}}
\newcommand{\malpha}[1]{\alpha_{#1}}
\newcommand{\madom}[1]{\mathbb{M}[{#1}]}

\newcommand{\eepdr}{\Lambda\mbox{\rm-PDR}}

\newcommand{\bkcube}{b}

\newcommand{\reflect}[1]{\textit{Ref}({#1})}

\newcommand{\join}{\mathbin{\sqcup}}

\newcommand{\bigjoin}{\bigsqcup}

\newcommand{\litabs}[1]{\textcolor{blue}{#1}}
\newcommand{\cubdom}[1]{\dom{#1}}

\makeatletter
\newcommand{\dotsym}[1]{{\vphantom{#1}\mathpalette\d@tsym{#1}\relax}}
\newcommand{\d@tsym}[1]{%
  \ooalign{\hidewidth$\m@th\cdot$\hidewidth\cr$\m@th #1$\cr}%
}
\makeatother
\makeatletter
\newcommand{\dotdelta}{{\vphantom{\delta}\mathpalette\d@td@lta\relax}}
\newcommand{\d@td@lta}[2]{%
  \ooalign{\hidewidth$\m@th#1\cdot$\hidewidth\cr$\m@th#1\delta$\cr}%
}
\makeatother

\newcommand\mathbox[1]{\mathord{\ThisStyle{%
  \fboxsep0\LMpt\relax\kern1\LMpt\fbox{$\SavedStyle#1$}\kern1\LMpt}}}
\newcommand{\cubejoin}[1]{\mathbox{#1}}

\newcommand{\restrict}[2]{{#1}\big|_{#2}}

\begin{document}

\newif\ifcomments
\commentsfalse
\nochangebars
\definecolor{dg}{cmyk}{0.60,0,0.88,0.27}

\newcommand{\sharonnew}[1]{\sharon{#1}}
\newcommand{\yotamnew}[1]{\yotamsmall{#1}}
\newcommand{\yotamforlater}[1]{}

\ifcomments
\newcommand{\artem}[1]{{\footnotesize\color{olive}[{\bf Artem}: #1]}}
\newcommand{\yotamsmall}[1]{{\footnotesize\color{magenta}[{\bf Yotam}: #1]}}

\newcommand{\sharon}[1]{{\textcolor{purple}{SS: {\em #1}}}}
\newcommand{\mooly}[1]{{\textcolor{cyan}{MS: {\em #1}}}}
\newcommand{\yotam}[1]{{\textcolor{magenta}{{\bf #1}}}}
\newcommand{\jrw}[1]{{\textcolor{green}{JRW: {\em #1}}}}
\newcommand{\TODO}[1]{{\textcolor{red}{TODO: {\em #1}}}}

\else
\newcommand{\sharon}[1]{}
\newcommand{\adam}[1]{}
\newcommand{\mooly}[1]{}
\newcommand{\neil}[1]{}
\newcommand{\jrw}[1]{}
\newcommand{\yotam}[1]{}
\newcommand{\TODO}[1]{}
\newcommand{\artem}[1]{}
\newcommand{\yotamsmall}[1]{}

\fi

\newcommand{\commentout}[1]{}
\newcommand{\OMIT}[1]{}  
\title{Property-Directed Reachability as Abstract Interpretation in the Monotone Theory}

\author{Yotam M. Y. Feldman}
\affiliation{
  \institution{Tel Aviv University}
  \country{Israel}
}
\email{yotam.feldman@gmail.com}

\author{Mooly Sagiv}
\affiliation{
  \institution{Tel Aviv University}
  \country{Israel}
}
\email{msagiv@acm.org}

\author{Sharon Shoham}
\affiliation{
  \institution{Tel Aviv University}
  \country{Israel}
}
\email{sharon.shoham@gmail.com}

\author{James R. Wilcox}
\affiliation{
  \institution{Certora}
  \country{USA}
}
\email{james@certora.com}

\begin{abstract}
Inferring inductive invariants is one of the main challenges of formal verification.
The theory of abstract interpretation provides a rich framework to devise invariant inference algorithms.
One of the latest breakthroughs in invariant inference is property-directed reachability (PDR), but
the research community views PDR and abstract interpretation as mostly unrelated techniques.

This paper shows that, surprisingly, propositional PDR can be formulated as an abstract interpretation algorithm in a logical domain.
More precisely, we define a version of PDR, called $\Lambda$-PDR, in which \emph{all} generalizations of  counterexamples are used to strengthen a frame.
In this way, there is no need to refine frames after their creation,
because all the possible supporting facts are included in advance.
We analyze this algorithm using notions from Bshouty's monotone theory, originally developed in the context of exact learning.
We show that there is an inherent overapproximation between the algorithm's frames that is related to the monotone theory.
We then define a new abstract domain in which the best abstract transformer performs this overapproximation, and show that it captures the invariant inference process, i.e., $\Lambda$-PDR corresponds to Kleene iterations with the best transformer in this abstract domain.
We provide some sufficient conditions for when this process converges in a small number of iterations, with sometimes an exponential gap from the number of iterations required for naive exact forward reachability.
These results provide a firm theoretical foundation for the benefits of how PDR tackles forward reachability.
\end{abstract}  


\maketitle

\section{Introduction}
The formal methods community has studied many approaches for automatic verification that are diverse and even seemingly disparate.
Two of the main approaches for verifying safety by the automatic inference of inductive invariants are abstract interpretation~\cite{DBLP:conf/popl/CousotC77} and model checking~\cite{DBLP:conf/lop/ClarkeE81,DBLP:conf/programm/QueilleS82}.
State-of-the-art model checking algorithms are based on SAT solving~\cite[e.g.][]{DBLP:conf/cav/GrafS97,DBLP:conf/cav/McMillan03,DBLP:conf/fmcad/GurfinkelI17,DBLP:conf/sigsoft/GurfinkelSM16,DBLP:conf/fmcad/SheeranSS00,DBLP:conf/cav/AlbarghouthiLGC12}, with many of the best tools implementing variants of the famous IC3/PDR algorithm~\cite{ic3,pdr}, which combines several heuristic ideas to achieve good performance in practice.
While PDR is practical and widely deployed, very little is known about it theoretically.
In particular, a previous investigation of PDR using the theory of abstract interpretation~\cite{DBLP:conf/vmcai/RinetzkyS16} had to employ abstractions that are both far from the usual practice of abstract interpreters, and are also too rich in that they can accommodate many algorithms that are unrelated to PDR (see~\Cref{sec:related}).
\begin{changebar}
As a result, there is currently no conceptual framework that explains how and when PDR is able to overapproximate and avoid the enumeration of reachable states, which is a key challenge to every invariant inference algorithm.
\end{changebar}

\begin{changebar}
In this paper, we set out to investigate the principles of how PDR achieves overapproximation.
\end{changebar}
To this end, we continue a line of work that applies learning theory to invariant inference~\cite[e.g.][]{ICELearning,DBLP:journals/jar/NeiderMSGP20,DBLP:journals/pacmpl/EzudheenND0M18,DBLP:conf/icse/JhaGST10,DBLP:conf/popl/0001NMR16,DBLP:conf/sas/0001GHAN13,DBLP:conf/cav/SharmaNA12,DBLP:conf/esop/0001GHALN13,DBLP:journals/fmsd/SharmaA16,DBLP:conf/pldi/KoenigPIA20,DBLP:journals/acta/JhaS17,DBLP:journals/pacmpl/FeldmanISS20,DBLP:journals/pacmpl/FeldmanSSW21} with a surprising result: the monotone theory from exact learning~\cite{DBLP:journals/iandc/Bshouty95} enables viewing PDR as classical abstract interpretation (in a new domain). This draws a deep connection between these techniques,
\begin{changebar}
and identifies a form of abstraction performed by PDR that distinguishes it both from explicit enumeration and from other algorithmic approaches. %
\end{changebar}

We focus on the fundamental setting of propositional systems,
\begin{changebar}
which also applies to infinite-state systems using predicate abstraction~\cite{DBLP:conf/popl/FlanaganQ02,DBLP:conf/cav/GrafS97}.
\end{changebar}
PDR constructs a sequence of formulas, called \emph{frames}, by blocking counterexamples (states that can reach bad states). 
Given a counterexample, the algorithm conjoins to the frame a generalization clause that blocks the counterexample and also additional states, but not states reachable in one step from the previous frame (we explain PDR in detail in~\Cref{sec:overview-pdr-alg}).
Theoretically analyzing the behavior of the algorithm is complicated by its highly nondeterministic nature---it depends on the choices of counterexamples and generalization clauses (many of them affected by idiosyncrasies of the underlying SAT solver), and different choices may lead PDR down radically different paths.
To ameliorate this, 
we present an algorithm, called \emph{$\Lambda$-PDR}, which resolves this nondeterminism by using all possible answers to these queries, blocking all counterexamples with all admissible generalizations. The resulting frames are tighter than those of PDR, as they include all lemmas that PDR could learn in any execution. This provides a theoretical handhold to study PDR.
$\Lambda$-PDR uncovers a key aspect of the generalization performed by standard PDR. The frames are usually viewed as a sequence of overapproximations that prove bounded safety with an increasing bound. While correct, this does not capture the full essence of generalization in PDR. In particular, naive exact forward reachability also computes such a sequence, albeit a trivial one.
We show that in $\Lambda$-PDR---and hence, in PDR---there is an inherent \emph{abstraction} that includes additional states in each frame beyond exact forward reachability. Applied successively, we show that this is a form of abstract interpretation,
which can lead to an exponential gap between the number of frames in $\Lambda$-PDR and the number of steps required for exact forward reachability.
We prove several results:
\begin{enumerate}

	\item We show that the relation between successive frames in $\Lambda$-PDR is characterized by an operation from Bshouty's monotone theory~\cite{DBLP:journals/iandc/Bshouty95}. The idea is that taking all the generalizations that block a state $b$ amounts to computing the least $b$-monotone overapproximation of the post-image of the previous frame (\Cref{sec:monotone}).

	\item We introduce a new abstract domain, of the formulas for which backward reachable states form a monotone basis. %
	We show that $\Lambda$-PDR can be viewed as computing Kleene iterations with the best abstract transformer in this domain.
	Standard PDR also operates in the same domain, and its frames overapproximate the Kleene iterations that $\Lambda$-PDR performs.
	This is the first time that the theory of state abstraction is able to explain property-directed generalization (\Cref{sec:ai}).
	\item We show exponential gaps between the number of frames in $\Lambda$-PDR and the number of iterations of exact forward reachability, as well as the unrolling depth in a dual interpolation-based algorithm (\Cref{sec:itp-friends}).

	\item We prove an upper bound on the number of frames in $\Lambda$-PDR in terms of the DNF size of certain ``monotonizations'' of the transition relation. Although not always tight, this result sheds light on the benefit of the abstraction in certain cases. The proof brings together results from the monotone theory, abstract interpretation, and diameter bounds for transitions systems. This is done by  constructing a (hyper)transition system where the states reachable in $i$ steps correspond to the $i$th Kleene iteration, and bounding the system's diameter (\Cref{sec:abstract-diameter-all}--\Cref{sec:hyper-all}).

	\item We show that in some cases the abstraction of $\Lambda$-PDR is overly precise, whereas the looser frames of standard PDR converge in fewer and smaller frames (\Cref{sec:vs-pdr}).
\end{enumerate}

\section{Overview}
\label{sec:overview}

\subsection{PDR, the Frames}
\label{sec:overview-frame-props}
\newcounter{overview-frame-props}
How does property-directed reachability find inductive invariants?
Given a set of initial states $\Init$, a transition relation $\tr$ describing one step of the system, and a set of bad states $\Bad$, the goal is to find an \emph{inductive invariant}: a formula $I$ such that $\Init \implies I$, $I \implies \neg\Bad$, and $\tr(I) \implies I$, where the post-image $\tr(X)$ is the set of states that $\tr$ reaches in one step from $X$.\footnote{
	We use a formula and the set of states that satisfy it interchangeably. Unless otherwise stated, the formula to represent a given set of states is chosen arbitrarily.
}
Such an $I$ proves \emph{safety}, that no execution of the system can reach a bad state.

The central data structure that PDR uses to find inductive invariants is the
\emph{frames}.
These are a sequence of formulas $\Frame_0,\Frame_1,\ldots,\Frame_N$ that satisfies the following properties, for all $0 \leq i \leq N-1$:

		\iflong\begin{enumerate}\else\begin{inparaenum}\fi
\item \label{it:frames-start} \label{it:frames-init} $\Frame_0 = \Init$;
\item \label{it:frames-monotone} $\Frame_i \implies \Frame_{i+1}$;
\item \label{it:frames-onestep-overapprox} $\tr(\Frame_i) \implies \Frame_{i+1}$;
\item \label{it:frames-end} \label{it:frames-safety} $\Frame_i \implies \neg \Bad$.
			\setcounter{overview-frame-props}{\value{enumi}}
		\iflong\end{enumerate}\else\end{inparaenum}\fi

\noindent
In words, the frames start with the set of initial states, grow monotonically, always include the states reachable in one step of $\tr$ from the previous frame, and are strong enough to prove safety (except possibly the last frame which is ``under construction'').

These properties ensure that each $\Frame_i$ overapproximates the set of states reachable in at most $i$ steps, and yet excludes the bad states; this constitutes a proof of \emph{bounded safety}. However, the ultimate goal of PDR is \emph{un}bounded safety, and it is not clear why frames would avoid ``overfitting'' to bounded executions, and rather converge to a true inductive invariant. In informal discussions, this is sometimes phrased as the criticism that the algorithm merely ``happens to find'' a bounded safety proof that generalizes to the unbounded case.
Indeed, properties~\ref{it:frames-start}--\ref{it:frames-end} of frames do not reflect any bias away from bounded proofs, as they are also satisfied by the exact forward reachability algorithm, $\Frame_0 = \Init, \Frame_{i+1}=\postimage{\tr}{\Frame_i}$, where $\postimage{\tr}{X} = X \cup \tr(X)$ denotes the reflexive closure of the post-image. Exact forward reachability might require many frames to converge to an unbounded proof if some states are reachable only by very long paths.

\begin{figure*}[t]
  \centering
\begin{minipage}{\textwidth}
  \begin{minipage}{0.43\textwidth}
  \begin{lstlisting}[numbersep=5pt, escapeinside={(*}{*)}, xleftmargin=3.0ex]
init $\vec{x} = (x_n,x_{n-1},\ldots,x_0) = 0\ldots0$, 
     $\vec{y} = (y_n,y_{n-1},\ldots,y_0) = 0\ldots0$, 
     $z=0$
repeat: increase_x() | increase_y()
assert $\neg\left(\vec{x} = 10\ldots0 \land \vec{y} = 11\ldots1 \land z=1\right)$
\end{lstlisting}
\end{minipage}
  \begin{minipage}{0.26\textwidth}
  \begin{lstlisting}[numbersep=5pt, escapeinside={(*}{*)}, xleftmargin=3.0ex]
increase_x():
   if $z = 0$:
      $\vec{x} = \vec{x}+1 \pmod{2^{n+1}}$
      if $\vec{x}=10\ldots0$:
         $\vec{x} = \vec{x}+1 \pmod{2^{n+1}}$
\end{lstlisting}
\end{minipage}
\begin{minipage}{0.25\textwidth}
  \begin{lstlisting}[numbersep=5pt, escapeinside={(*}{*)}, xleftmargin=3.0ex]
  increase_y():
   if $z = 1$:
      $\vec{y} = \vec{y}+1 \pmod{2^{n+1}}$
  $$
  $$
\end{lstlisting}
\end{minipage}
\vspace{-0.4cm}
\captionof{figure}{
\iflong
\footnotesize 
\fi
Skip-counter: running example of propositional transition system over the variables $\vec{x}=x_n,\ldots,x_0$. Either \texttt{increase\_x()} or \texttt{increase\_y()} is executed in each step according to whether $z=0$ or $z=1$, incrementing $\vec{x}$ or $\vec{y}$ resp., but skipping over the value $\vec{x}=10\ldots0$.}
  \label{fig:skip-counter}
  \end{minipage}
\iflong
\else
\vspace{-0.5cm}
\fi
\end{figure*}  Consider, for example, the simple family of propositional systems in~\Cref{fig:skip-counter}, parametrized by $n$. A bit $z$ chooses between incrementing a counter $\vec{x}$ or a counter $\vec{y}$, represented in binary by $x_n,x_{n-1},\ldots,x_0$ and $y_n,y_{n-1},\ldots,y_0$ respectively. The safety property to prove is that it is impossible for $\vec{x}$ to have the value $10\ldots 0$ while $\vec{y}$ is $11\ldots 1$. This property is not inductive as is (for instance, the state $\vec{x}=10\ldots 0, \vec{y}=11\ldots 10, z=1$ satisfies the safety property but reaches a bad state in one step); %
one inductive invariant for this system is
\begin{equation}
\label{eq:skip-counter-invariant}
	\vec{x} \neq 10\ldots0 \land \vec{y}=00\ldots0 \land z=0,
\end{equation}
which implies safety and is closed under a step of the system.

In these systems, exact forward reachability requires $\Omega\left(2^n\right)$ iterations before it converges to an inductive invariant. This is because some states, such as $\vec{x}=10\ldots 01, \vec{y}=00 \ldots 00,z=0$, require an exponential number of steps to reach---the system has an exponential \emph{reachability diameter}---so exact forward reachability discovers all reachable states and converges only after that many iterations.

Clearly, invariant inference algorithms must perform some sort of \emph{overapproximation}, or \emph{abstraction}, to overcome this slow convergence. %
This raises two important questions:
\begin{enumerate}
	\item \label{it:q-abstraction} What characterizes the abstraction that PDR performs? %
	\item \label{it:q-convergence} How does this abstraction achieve faster convergence than exact forward reachability? %
\end{enumerate}
The commonly-stated properties of frames do not provide an answer; to address these questions we must dive more deeply into how PDR works.

\subsection{PDR, the Algorithm}
\label{sec:overview-pdr-alg}
\vspace{-0.45cm}
\begin{algorithm}[H]
\caption{PDR}
\label{alg:pdr}
\vspace{-0.5cm}
\begin{multicols}{2}
\begin{algorithmic}[1]
\begin{footnotesize}
\Procedure{PDR}{$\Init$, $\tr$, $\Bad$}
  \State $\Framepdr_0 \gets \Init$
	\State $N \gets 0$
  \While{$\forall 1 \leq i \leq N. \ \Framepdr_{i} \not \implies \Framepdr_{i-1}$} $\label{ln:pdr-find-inductive-frame}$
		\State $\Framepdr_{N+1} \gets \true$ $\label{ln:pdr-init-frame}$
		\While{$\Framepdr_{N+1} \notimplies \neg\Bad$}
			\For{$\sigma_b \in \Framepdr_{N+1} \land \Bad$} $\label{ln:pdr-sample-bad}$
        		\State \Call{block}{$\sigma_b$, $N+1$} $\label{ln:pdr-block-bad}$
			\EndFor
		\EndWhile
		\State $N \gets N+1$
	\EndWhile
  \State \Return $\Framepdr_i$ such that $\Framepdr_{i} \implies \Framepdr_{i-1}$
\EndProcedure

\Procedure{block}{$\sigma_b$, $i$}
	\If{$i=0$}
		\State \textbf{unsafe}
	\EndIf
	\While{$\postimage{\tr}{\Framepdr_{i-1}} \notimplies \neg \sigma_b$} $\label{ln:pdr-back-refine}$
		\State take $\sigma$ s.t.\ $\sigma \models \Framepdr_i, (\sigma,\sigma_b) \models \reflextr{\tr}$ $\label{ln:pdr-back-sample-prestate}$
		\State \Call{block}{$\sigma$, $i-1$}
	\EndWhile
	\State {%
		take $c$ minimal s.t.\ $c \subseteq \neg \sigma_b$ and $\tr({\Framepdr_{i-1}}) \implies c$ \\
		\qquad \qquad \qquad \qquad \qquad \qquad \quad and $\Init \implies c$\strut} $\label{ln:pdr-minimal-clause}$
	\For{$1 \leq j \leq i$}
		\State $\Framepdr_j \gets \Framepdr_j \land c$ $\label{ln:pdr-strengthen}$
	\EndFor
\EndProcedure
\end{footnotesize}
\end{algorithmic}
\end{multicols}
\vspace{-0.35cm}
\end{algorithm}
 \vspace{-0.4cm}
\Cref{alg:pdr} presents a simple version of the basic PDR algorithm. The sequence of frames it manipulates are denoted $\Framepdr_0,\ldots,\Framepdr_N$, to distinguish between PDR's frames and frames of other algorithms in the paper. Initially, $\Framepdr_0$ is initialized to the set of initial states (thereby satisfying property~\ref{it:frames-init}).
The outer loop terminates once one of the frames is inductive (\cref{ln:pdr-find-inductive-frame}), which is when $\Framepdr_{i} \implies \Framepdr_{i-1}$ (because then, from the other properties of frames, $\tr(\Framepdr_{i-1}) \implies \Framepdr_i \implies \Framepdr_{i-1}$).
Otherwise, it initializes a new frontier frame to $\true$ (\cref{ln:pdr-init-frame}), and samples bad states (\cref{ln:pdr-sample-bad}) to block (exclude from the frame) until the frontier frame is strong enough to exclude all bad states (satisfying property~\ref{it:frames-safety}).

In order to satisfy property~\ref{it:frames-onestep-overapprox}, before a state $\sigma_b$ is blocked, the previous frame must be refined it excludes all the pre-states of $\sigma_b$ (\cref{ln:pdr-back-refine}); this is performed by sampling pre-states and blocking them recursively (\cref{ln:pdr-back-sample-prestate}).
Once all the pre-states are blocked in the previous frame, $\sigma_b$ can be excluded from the current frame.
However, at this point, %
PDR \emph{generalizes} and blocks a \emph{set} of states; this is done by finding \emph{clause} $c$---also called a \emph{lemma}---that excludes $\sigma_b$ and still does not exclude any state that is reachable in one step or less from the previous frame (preserving property~\ref{it:frames-onestep-overapprox}).
This is done (in~\cref{ln:pdr-minimal-clause}) by starting from all literals (variables or their negations) that are falsified in $\sigma_b$, and choosing a subset whose disjunction (which is a clause) satisfies the desired properties. PDR chooses a \emph{minimal subset} in order to exclude as many states as possible. (In practice, this involves a linear number of SAT calls.)\footnote{
	Practical implementations also attempt to push existing lemmas to other frames whenever possible; we omit this for simplicity. (We discuss inductive generalization below, in~\Cref{sec:overview-inductive-generalization}.)
}
The clause is conjoined (in~\cref{ln:pdr-strengthen}) to the frame as well as the preceding ones (thereby satisfying property~\ref{it:frames-monotone}, relying on $\Init \implies c$).

The above is an operational description of how the frames are generated to be overapproximations, but does not lay bare the principles of why they are computed in this way, and how to characterize the abstraction that frames perform.

To study this, we introduce $\Lambda$-PDR\footnote{In homage to Bshouty's $\Lambda$-algorithm~\cite{DBLP:journals/iandc/Bshouty95}, not the \textsc{SARS-CoV-2} variant.}, an alteration of PDR that is simpler for analysis.
This algorithm is a theoretical device to study the abstraction in PDR's frames: each frame of $\Lambda$-PDR is tighter than the corresponding frame of PDR, and thus \textbf{\textit{the overapproximation that $\Lambda$-PDR's frames perform is also performed in usual PDR}}.
We characterize the abstraction that $\Lambda$-PDR performs, and show how it can converge more rapidly than exact forward reachability, which sheds light on the abstraction in PDR.
\subsection{$\Lambda$-PDR}
\label{sec:overview-eepdr}

\subsubsection{The Algorithm}
\begin{wrapfigure}{R}{0.35\textwidth}
\vspace{-0.75cm}
\begin{minipage}{0.35\textwidth}
\begin{algorithm}[H]
\caption{$\Lambda$-PDR}
\label{alg:eepdr}
\begin{algorithmic}[1]
\begin{footnotesize}
\Procedure{$\Lambda$-PDR}{$\Init$, $\tr$, $\Bad$, $k$}
	\State $\Frame_{-1} \gets \false$
	\State $\Frame_0 \gets \Init$ $\label{ln:eepdr-frame0}$
	\State $\bkwrch{k} = \set{\sigma \ | \ \bmc{\tr}{\sigma}{k} \cap \Bad \neq \emptyset}$ $\label{ln:eepdr-bk}$
	\If{$\Init \cap \bkwrch{k} \neq \emptyset$} 
		\textbf{unsafe} $\label{ln:eepdr-unsafe}$
	\EndIf
	\State $i \gets 0$
        \While{$\Frame_{i} \not \implies \Frame_{i-1}$} $\label{ln:eepdr-while-not-inductive}$
		\If{$\tr(\Frame_{i}) \cap \bkwrch{k} \neq \emptyset$} $\label{ln:eepdr-restart}$
			\State \textbf{restart} with larger $k$ \label{ln:eepdr-increase-k}
		\EndIf
		\State $\Frame_{i+1} \gets \true$
		\For{$\sigma_b \in \bkwrch{k}$} $\label{ln:eepdr-bmc}$
			\For{$c \subseteq \neg \sigma_b$} $\label{ln:eepdr-for-clause}$
				\If{$\postimage{\tr}{\Frame_i} \implies c$} $\label{ln:eepdr-lemma-check}$
					\State $\Frame_{i+1} \gets \Frame_{i+1} \land c$
				\EndIf
			\EndFor
		\EndFor
		\State $i \gets i+1$
	\EndWhile
	\State \Return $\Frame_i$
\EndProcedure
\end{footnotesize}
\end{algorithmic}
\end{algorithm}
\end{minipage}
\vspace{-0.4cm}
\end{wrapfigure}
\Cref{alg:eepdr} presents $\Lambda$-PDR.
Briefly, it constructs frames one after the other, by including all possible lemmas that any execution of PDR might learn; $\Frame_{i+1}$ is the conjunction of \emph{all} \emph{clauses} that \emph{block} a state from $\bkwrch{k}$---the set of states that \emph{can reach a bad state} in at most $k$ steps---yet \emph{retain} the states \emph{reachable in one step} from $\Frame_i$.
The algorithm's essentials are similar to PDR's, with important changes.

First, it is useful for our purpose to decouple two roles that frames serve in PDR. One is as a sequence of approximations to the invariant until the frame where an invariant is found (which is usually somewhere in the middle of the sequence). The other is a way to find counterexamples---which are states that can reach a bad state---without unrolling the transition relation~\cite{ic3}.
\begin{changebar}
In $\Lambda$-PDR we instead imagine that we are provided, through some arbitrary means (such as unrolling), with $\bkwrch{k}$, the set of states that can reach a state in $\Bad$ along some execution of length at most $k$ steps (\cref{ln:eepdr-bk}).
This allows us to focus on the other role of frames as approximations that converge to the invariant.
The number $k$ is chosen in advance, independently of the number of frames $N$.\footnote{
	At first sight $\Framepdr_i$ consists of clauses that exclude counterexamples from a lower backward reachability bound, $\bkwrch{N-i}$; but in fact, pushing lemmas forward means that even $\Framepdr_N$ can include  clauses learned at $\Framepdr_1$ from counterexamples in $\bkwrch{N}$.
}
\end{changebar}

Second, frames are computed without backtracking to refine previous frames; lemmas to support future frames are learned eagerly, in advance.
In particular, convergence is always at the last frame.

Third, whereas PDR ``samples'' $k$-backward reachable states and blocks each counterexample in $\Frame_{i+1}$ using a single, arbitrary (minimal) clause that does not exclude states in $\postimage{\tr}{\Frame_i}$,  $\Lambda$-PDR generates \emph{all} such clauses---for \emph{any} counterexample state from $\bkwrch{k}$ (\cref{ln:eepdr-bmc}) and \emph{any} order of dropping literals (\cref{ln:eepdr-for-clause}). This ``determinization'' makes the algorithm easier to analyze.

\begin{changebar}
Overall, the algorithm computes each frame $\Frame_{i+1}$ iteratively, from the previous $\Frame_i$, without ever refining previous frames. Each frame is the conjunction of all the clauses that can be obtained as lemmas from blocking any counterexample from $\bkwrch{k}$ while still overapproximating $\postimage{\tr}{\Frame_i}$. This process continues until an inductive invariant is found (\cref{ln:eepdr-while-not-inductive})---unless the current frame does include a counterexample from $\bkwrch{k}$, which prompts an increase of $k$ (\cref{ln:eepdr-restart}) in order to distinguish between spurious overapproximation and truly unsafe systems (detected in~\cref{ln:eepdr-unsafe} by an initial state that can reach a bad state in $k$ steps).
\end{changebar}

\begin{changebar}
As an example, this is how $\Lambda$-PDR proceeds on the example of~\Cref{fig:skip-counter} with (say) $k=1$:
The $k$-backward reachable states $\bkwrch{k}$ are those where $\vec{x} = 10\ldots00 \land y_n y_{n-1} \ldots y_1 = 11\ldots1 \land z=1$ (every value of $y_0$ yields a backward reachable state).
The frame sequence is initialized with $\Frame_0 = \Init$. As $\postimage{\tr}{\Frame_0}$ does not intersect $\bkwrch{k}$, the algorithm proceeds to computing $\Frame_1$. It starts as $\true$, and the algorithm iterates through the states in $\bkwrch{k}$ to generate clauses. Suppose that the first counterexample $\sigma_b$ is $\vec{x} = 10\ldots00 \land \vec{y} = 11\ldots10 \land z=1$. We write
$\neg \sigma_b = (x_n \neq 1) \lor (x_{n-1} \neq 0) \lor \ldots \lor (x_{1} \neq 0) \lor (x_{0} \neq 0) \lor (y_n \neq 1) \lor (y_{n-1} \neq 1) \lor \ldots \lor (y_1 \neq 1) \lor (y_0 \neq 0) \lor (z \neq 1)$, and consider every possible sub-clause $c$ thereof, checking whether $\postimage{\tr}{\Frame_0} \implies c$, namely, $c$ includes both $\vec{x}=00\ldots00,\vec{y}=00\ldots00,z=0$ and $\vec{x}=00\ldots01,\vec{y}=00\ldots00,z=0$. In this case, there are several incomparable (and minimal) such $c$'s: $x_n \neq 1$, $y_i \neq 1$ for every $i>1$, and $z \neq 1$. \TODO{did I miss something?}
In $\Lambda$-PDR, \emph{all} these potential clauses are conjoined to $\Frame_1$.
The algorithm performs the same procedure for all the counterexamples in $\bkwrch{k}$.
Once this is done, $\Frame_1$ never changes again in the course of the algorithm, and it becomes the basis for constructing $\Frame_2$ in the same way, and so on until an inductive invariant is found or a restart becomes necessary. (We later show the resulting $\Frame_1,\Frame_2,\ldots$ in this example.)
\end{changebar}

\begin{changebar}
\subsubsection{PDR Overapproximates at Least as Much as $\Lambda$-PDR}
The importance of $\Lambda$-PDR for our investigation of abstraction in PDR stems from the fact that $\Frame_i \implies \Framepdr_i$, when $k=N$ is the final number of frames in PDR (\Cref{cor:lambda-pdr-underapproximates-pdr}).
This implies that whatever overapproximation $\Lambda$-PDR performs also transfers to PDR: the overapproximation in $\Lambda$-PDR is a lower bound for the overapproximation in PDR.
The relationship $\Frame_i \implies \Framepdr_i$ holds because every clause $c$ that PDR can use to strengthen $\Framepdr_i$ is used to strengthen $\Frame_i$ of $\Lambda$-PDR (roughly, in PDR, for $c$ to added to $\Framepdr_i$, it must block a counterexample from $\bkwrch{k}$ and overapproximate the post-image of the previous frame; the same properties would hold for $c$ in $\Lambda$-PDR, thus ensuring that $c$ is conjoined to $\Frame_i$---see~\Cref{cor:lambda-pdr-underapproximates-pdr}).

Our goal then is to show that $\Lambda$-PDR performs significant overapproximation over exact forward reachability, and thereby establish the same for PDR.
Our first step is to characterize the overapproximation that $\Lambda$-PDR performs, and for this we need tools developed in exact concept learning.
\end{changebar}

\subsection{Abstraction from The Monotone Theory}\label{sec:overview-monotone}
The main technical enabler of our paper is the observation (\Cref{lem:justify-overview-next-frame}) that in $\Lambda$-PDR, there is a well-defined relation between successive frames, through what we call the \emph{monotone hull}:
\vspace{-0.15cm}
\begin{tcolorbox}[boxsep=-4pt]
\begin{equation}
\label{eq:overview-next-frame}
	\Frame_{i+1} = \mhull{\postimage{\tr}{\Frame_i}}{\bkwrch{k}} \eqdef \bigwedge_{b \in \bkwrch{k}}{\monox{\postimage{\tr}{\Frame_i}}{b}},
\end{equation}
\end{tcolorbox}
\vspace{-0.15cm}
\noindent
where $\monox{\varphi}{b}$ is the central operator in the monotone theory~\cite{DBLP:journals/iandc/Bshouty95}, the \emph{least $b$-monotone overapproximation of $\varphi$} (``$b$-monotonization'' in short). A Boolean function $f$ is $b$-monotone, when $b$ is a state, if whenever $f(\sigma_1)=1$ and $\sigma_1 \leq_b \sigma_2$, which means that $\sigma_2$ is obtained from $\sigma_1$ by flipping bits on which $\sigma_1,b$ agree, then also $f(\sigma_2)=1$.
$\monox{\varphi}{b}$ is the \emph{least} $b$-monotone formula (function) implied by $\varphi$.
(We elaborate on the technical details in~\Cref{sec:monotone-background}.)
The insight of~\Cref{eq:overview-next-frame} is that, as we show, every lemma in PDR is implied by the monotone hull, and the conjunction of all possible lemmas is exactly the monotone hull.
Technically, the observation builds on an equivalent formulation of $\monox{\varphi}{b}$ in a conjunctive form, which is not explicit in Bshouty's paper (\Cref{lem:monox-conjunction-clauses}).

Our central observation is that the monotone hull operator introduces \emph{overapproximation} to the sequence of frames---$\mhull{\postimage{\tr}{\Frame_i}}{\bkwrch{k}}$ can include many more states than $\postimage{\tr}{\Frame_i}$.
This is an interesting deviation from Bshouty's use of monotonizations in \emph{exact} learning of an unknown $\varphi$, %
where a set $B$ is chosen such that $\varphi \equiv \mhull{\varphi}{B}$; in that case $B$ is said to be a monotone basis for $\varphi$, denoted $\varphi \in \mspan{B}$ (\Cref{def:monotone-span}).
In contrast, here the monotone hull is applied to intermediate frames, and we are interested in the cases that the set $\bkwrch{k}$ is \emph{not} a monotone basis for $\postimage{\tr}{\Frame_i}$, for then \Cref{eq:overview-next-frame} indicates a strict \emph{overapproximation of exact forward reachability} that $\Lambda$-PDR performs.

\label{sec:overview-example-frame}
Consider again the running example (\Cref{fig:skip-counter}).
One step of exact forward reachability discovers the state $\vec{x}=00\ldots01, \vec{y}=00\ldots00, z=0$ in addition to the initial state.
In contrast, by~\Cref{eq:overview-next-frame}, the first frame of $\Lambda$-PDR with $k=1$ is $\Frame_1 = \mhull{\postimage{\tr}{\Init}}{\bkwrch{k}}$, resulting in
\begin{equation}
\label{eq:running-frame-1}
	\Frame_1 \, = \, x_n=0 \land \vec{y}=00\ldots00 \land z=0;
\end{equation}
in a single iteration the algorithm has leaped over an exponential number of steps of $\tr$!

To compute $\mhull{\postimage{\tr}{\Init}}{\bkwrch{k}}$, in order to arrive at~\Cref{eq:running-frame-1}, we  compute the monotone overapproximations. %
In our example, $\bkwrch{k}$ is a single cube:
\begin{align*}
	\bkwrch{k} &= (x_n=1 \land x_{n-1}=0 \land \ldots \land x_1=0 \land x_0=0 \land y_n=1 \land \ldots \land y_1=1 \land z=1),
\\
\intertext{in which case $\mhull{\postimage{\tr}{\Frame_0}}{\bkwrch{k}}$ can be calculated by writing $\postimage{\tr}{\Frame_0}$ in DNF (see~\Cref{lem:mhull-dnf-base,lem:bshouty-mon-mindnf}):}
	\postimage{\tr}{\Frame_0} &=
		(x_n=0 \land \litabs{x_{n-1}=0} \land \ldots \land \litabs{x_1=0} \land \litabs{x_0=0} \land y_n=0 \land \ldots \land y_1=0 \land y_0=0 \land z=0)
		\\
		&\lor
		(x_n=0 \land \litabs{x_{n-1}=0} \land \ldots \land \litabs{x_1=0} \land x_0=1 \land y_n=0 \land \ldots \land y_1=0 \land y_0=0 \land z=0),
\end{align*}
and in each term omitting every literal that agrees with the cube of $\bkwrch{k}$ (appearing in color).
(When $\bkwrch{k}$ is not a single cube, $\Frame_{i+1}$ is computed as a conjunction of the above for each cube in $\bkwrch{k}$.)

\begin{figure}[t]
  \centering
  \begin{subfigure}[t]{0.35\textwidth}
    \includegraphics[width=\textwidth]{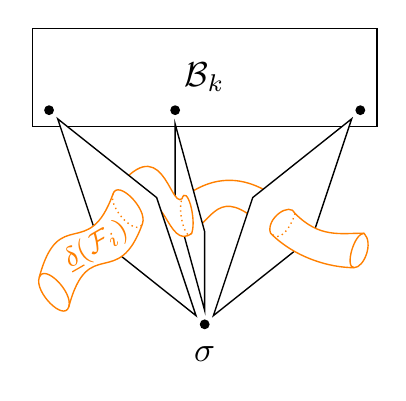}
    \caption{}
    \iflong\else\vspace{-0.2cm}\fi
    \label{fig:protected-single}
  \end{subfigure}
  \hspace{0.15\textwidth}
  \begin{subfigure}[t]{0.3\textwidth}
    \includegraphics[width=\textwidth]{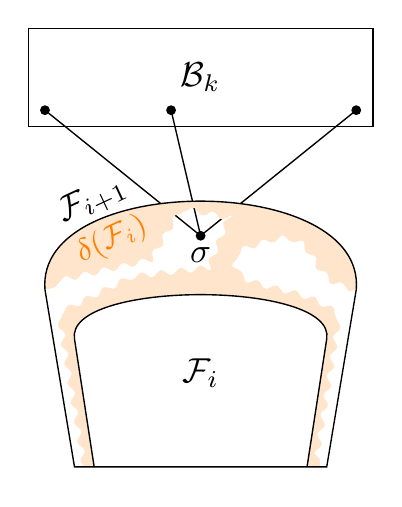}
    \caption{}
    \iflong\else\vspace{-0.2cm}\fi
  \end{subfigure}
  \caption{(a) $\sigma$ is ``protected'' by $\postimage{\tr}{\Frame_i}$ from exclusion due to blocking $\bkwrch{k}$,
    thus (b) $\sigma$ is in $\Frame_{i+1} = \mhull{\postimage{\tr}{\Frame_i}}{\bkwrch{k}}$.
    \TODO{Say that we are ``punning'' Euclidean geometry and Hamming geometry.}}
  \label{fig:protected-hull}
\iflong\else\vspace{-0.45cm}\fi
\end{figure}

What is especially significant about this overapproximation in $\Lambda$-PDR is that it exists \emph{in each step}, as we describe in the next subsection (\Cref{sec:overview-successive}). But before that, we explain the intuition for where this overapproximation stems from.

The cause of overapproximation in $\Lambda$-PDR is the special constraints on the lemmas the algorithm can generate.
Recall that the states that remain in $\Frame_{i+1}$ are those that \emph{cannot be excluded} by \emph{any} lemma starting from \emph{any} counterexample, due to the need to satisfy property~\ref{it:frames-onestep-overapprox}.
Since lemmas are \emph{not} arbitrary formulas, perhaps surprisingly, such states exist beyond the exact post-image.
We demonstrate what these states are using the running example.
Consider a state $\sigma$ that satisfies~\Cref{eq:running-frame-1}.
Why does no lemma $c$ learned by the algorithm exclude $\sigma$, i.e., $\sigma \not\models c$? The reason is that every $c$ excludes some counterexample $b \in \bkwrch{k}$, and, furthermore, $c$ is a clause---a negation of a cube. A cube is a very rigid geometric shape; if $\neg c$ contains both $\sigma$ and $b$ then it necessarily contains many other states---it must include all states that are within the smallest cube that contains both $\sigma,b$, a.k.a.\ the Hamming interval between $\sigma,b$~\cite[e.g.][]{wiedemann1987hamming}.
For example, the Hamming interval between $\sigma = (\vec{x}=011\ldots101,\vec{y}=00\ldots00,z=0)$ and $b = (\vec{x}=10\ldots00,\vec{y}=11\ldots10,z=1)$ is $x_1=0 \land y_0=0$---the conjunction of the literals (or constraints) that hold in both $\sigma,b$.
However, $\neg c$ must not contain any state in $\postimage{\tr}{\Frame_i}$, so the Hamming interval between $\sigma,b$ cannot intersect $\postimage{\tr}{\Frame_i}$.
In our example, the Hamming interval between $\sigma$ and $b$ includes the state $\widetilde{\sigma}=(\vec{x}=00\ldots00, \vec{y}=00\ldots00, z=0)$, and $\widetilde{\sigma} \in \postimage{\tr}{\Frame_0}$, so a lemma $c$ that excludes $\sigma$ and originates from blocking $b$ cannot be conjoined to $\Frame_1$.

Put differently, $\widetilde{\sigma} \in \postimage{\tr}{\Frame_0}$ ``protects'' $\sigma \in \Frame_{1}$ from being excluded due to $b$.
In general, a state $\sigma$ is included in $\Frame_{i+1}$ if a protector state $\widetilde{\sigma} \in \Frame_i$ exists \emph{for every} $b \in \bkwrch{k}$, namely, the Hamming interval between $\sigma,b$ crosses $\postimage{\tr}{\Frame_i}$ for all $b$'s (\Cref{fig:protected-hull}). In our example, the same $\widetilde{\sigma}$ actually protects every $\sigma \in \Frame_1$ from exclusion due to any $b \in \bkwrch{k}$, but multiple protector states may be necessary in general.

The idea of protector states explains why $\Frame_{i+1}$ is $\mhull{\postimage{\tr}{\Frame_i}}{\bkwrch{k}}$ (\Cref{eq:overview-next-frame}).
Every state $\widetilde{\sigma} \in \postimage{\tr}{\Frame_0}$ is a protector state. The states that $\widetilde{\sigma}$ protects from $b$ are the states $\sigma$ such that $\widetilde{\sigma}$ is in the Hamming interval between $\sigma$ and $b$; these states are ``farther away'' from $b$ than $\widetilde{\sigma}$, in the sense of $\widetilde{\sigma} \leq_b \sigma$ as defined above (and formally in~\Cref{def:b-monotone-order}).
The set protected from $b$ by $\postimage{\tr}{\Frame_i}$ is therefore $\monox{\postimage{\tr}{\Frame_i}}{b}$, and the states that are protected from all $b \in \bkwrch{k}$ are the conjunction over all $b$'s, namely $\mhull{\postimage{\tr}{\Frame_{i}}}{\bkwrch{k}}$.

\subsection{Successive Overapproximation: Abstract Interpretation}
\label{sec:overview-successive}
The overapproximation of~\Cref{eq:overview-next-frame} is present between each two consecutive frames; it is thus performed repeatedly, using the previous overapproximation as the starting point of the next. In~\Cref{sec:ai} we show that $\Lambda$-PDR can be cast as \textbf{\textit{abstract interpretation in a new logical domain}}, of the formulas in $\bkwspan{k}$, the formulas $\varphi$ s.t.\ $\mhull{\varphi}{\bkwrch{k}} \equiv \varphi$, which are the formulas expressible by a conjunction of clauses that each excludes a state from $\bkwrch{k}$ (\Cref{def:monotone-span}). \textbf{\textit{The frames of $\Lambda$-PDR are completely characterized as Kleene iterations with the best abstract transformer}} in this domain (\Cref{lem:ai-eepdr-sandwich}), when correcting for the slightly different initial frame ($\Frame_0=\Init$ vs.\ $\mhull{\Init}{\bkwrch{k}}$; we show that the resulting difference in the number of frames is at most one).

\newcounter{reorder-fig-idx}
\begin{figure}[t]
\begin{minipage}{.37\textwidth}
\begin{minipage}{\textwidth}
  \centering
  \captionsetup{width=.9\textwidth}
  \includegraphics[width=\textwidth]{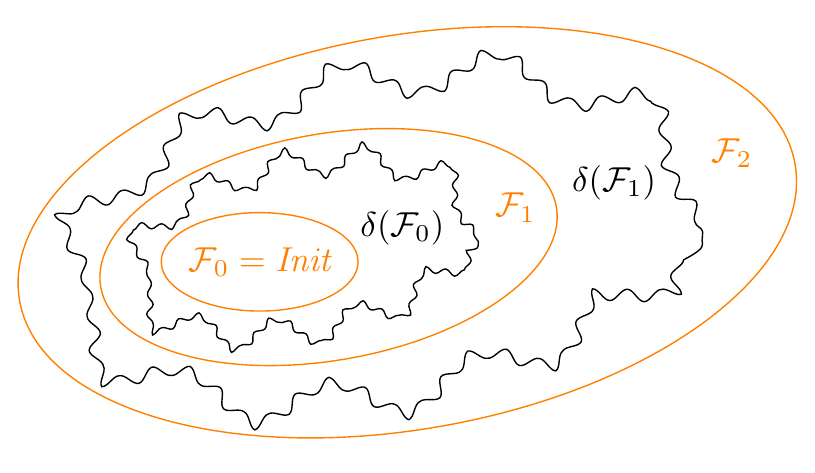}
  \vspace{-1cm}
  \caption{\iflong\footnotesize\fi
  Repeatedly interleaving the post-image and monotone hull operators results in
    successive overapproximation.}
  \label{fig:successive-approx}
\vspace{0.3cm}
\end{minipage}
\begin{minipage}{\textwidth}
  \setcounter{reorder-fig-idx}{\value{figure}}
  \setcounter{figure}{\value{reorder-fig-idx}+1}
  \centering
\begin{tikzpicture}
  \def\w{0.5in}
  \def\h{1in}
  \def\b{3pt}
  \def\mytiny#1{\scalebox{.5}{#1}}
  \node[draw,shape=rectangle,minimum width=\w,minimum height=\h] (Bk) {$\bkwrch{k}$};
  \draw[gray,->,>={Classical TikZ Rightarrow[]},double,double distance=1.5pt] ($(Bk.north west)-(\b,\b)$) -- +(-\w, 0);
  \draw[gray,->,>={Classical TikZ Rightarrow[]},double,double distance=1.5pt] ($(Bk.west)-(\b,0)$) -- node[align=center,above=8pt] {\mytiny{direction of}\\[-6pt]\mytiny{monotonization}} +(-\w, 0);
  \draw[gray,->,>={Classical TikZ Rightarrow[]},double,double distance=1.5pt] ($(Bk.south west)+(-\b,\b)$) -- +(-\w, 0);

  \node[circle,fill=black,inner sep=0.75pt,outer sep=3pt, label={[label distance=-7pt]below right:{\tiny{$\sigma_2$}}}] (t1s) at ($(Bk.south west)!0.25!(Bk.north west)+1.5*(-\w,0)$) {};
  \node[circle,fill=black,inner sep=0.75pt,outer sep=3pt, label={[label distance=-7pt]above left:{\tiny{$\sigma_2'$}}}] (t1t) at ($(Bk.south west)!0.75!(Bk.north west)+2.5*(-\w,0)$) {};
  \node[circle,fill=black,inner sep=0.75pt,outer sep=3pt, label={[label distance=-7pt]above right:{\tiny{$\sigma_1'$}}}] (t2t) at ($(Bk.south west)!0.75!(Bk.north west)+1.5*(-\w,0)$) {};
  \node[circle,fill=black,inner sep=0.75pt,outer sep=3pt, label={[label distance=-7pt]below left:{\tiny{$\sigma_1$}}}] (t2s) at ($(Bk.south west)!0.25!(Bk.north west)+2.5*(-\w,0)$) {};
  \draw[->] (t1s) -- node[right, pos=.3]{\tiny{$t_1$}} (t1t);
  \draw[->] (t2s) -- node[left, pos=.3]{\tiny{$t_2$}} (t2t);
  \draw[gray,->,>={Classical TikZ Rightarrow[]},double,double distance=0.5pt] (t2t) -- (t1t);
  \draw[gray,->,>={Classical TikZ Rightarrow[]},double,double distance=0.5pt] (t1s) -- (t2s);
\end{tikzpicture}
  \caption{\iflong\footnotesize\fi The monotonization of a transition $t_1=(\sigma_1,\sigma'_1) \in \tr $ subsumes that of $t_2=(\sigma_2,\sigma'_2) \in \tr$ if $\sigma_1$ is farther from the backward reachable cube than $\sigma_2$, and $\sigma'_1$ is closer than $\sigma'_2$. If few transitions subsume the rest, $\dnfsize{\absr{\tr}}$ is small, which implies convergence with few frames for $\Lambda$-PDR.}
  \label{fig:arrow-subsumption}
\end{minipage}
\end{minipage}\hspace{0.2cm}%
\begin{minipage}{.59\textwidth}
  \centering
  \setcounter{figure}{\value{reorder-fig-idx}}
  \setcounter{reorder-fig-idx}{\value{figure}}
  \captionsetup{width=.95\textwidth}
  \includegraphics[width=\textwidth]{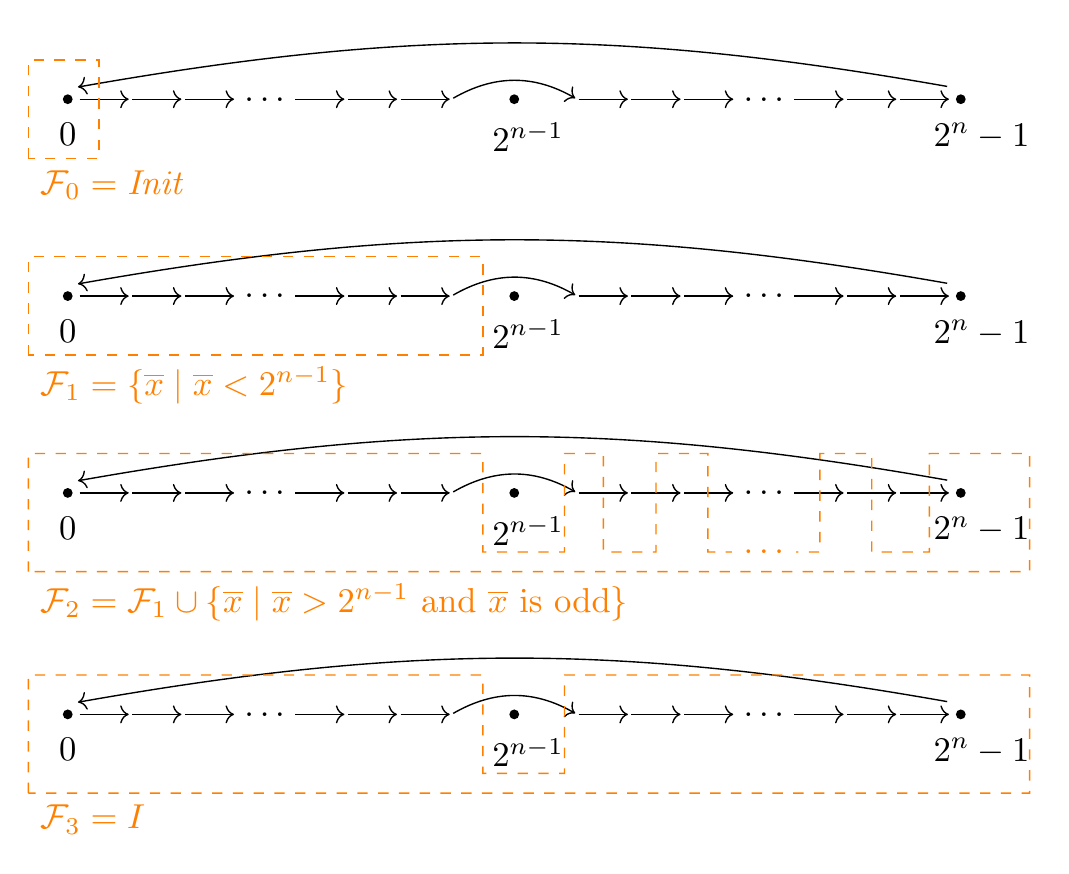}
  \caption{\iflong\footnotesize\fi The frames of $\Lambda$-PDR on the running example (only values of $\vec{x}$ are displayed, always $\vec{y}=0\ldots0,z=0$).
    The number of frames required for convergence is always 4, independent of the parameter $n$.
    }
  \label{fig:lambda-pdr-frames}
\end{minipage}
\vspace{-0.4cm}
\end{figure}
\setcounter{figure}{\value{reorder-fig-idx}+1}

Repeatedly applying the abstraction can cause $\Lambda$-PDR to converge much faster than exact forward reachability (illustrated schematically in~\Cref{fig:successive-approx}).
For example, the frames of $\Lambda$-PDR on the running example are displayed in~\Cref{fig:lambda-pdr-frames} (we perform the calculation in detail in~\Cref{ex:running-all-frames}), and
$\Frame_3$ is none other than the inductive invariant of~\Cref{eq:skip-counter-invariant}.
In this way, \textbf{\textit{successive overapproximation can lead to convergence in a smaller number of frames than exact forward reachability}}; in this example, $\Lambda$-PDR converges in $4$ frames, rather than an exponential number as exact forward reachability would use.\footnote{
	The operations in the abstract domain are not efficient; we focus on the number of iterations until convergence.
}
In~\Cref{sec:itp-friends}, we show that $\Lambda$-PDR holds a similar advantage over the unrolling depth of an interpolation-based algorithm.

\subsection{Convergence Bounds via (Hyper)Diameter Bounds}
\label{sec:overview-diameter-bound}
When does the successive overapproximation of $\Lambda$-PDR terminate in a small number of iterations? The lattice height of the abstract domain is exponential in the number of variables (see~\Cref{sec:ai}), and rapid convergence depends on properties of the transition system rather than of the abstract domain.

We prove a convergence bound using one such property: \textit{\textbf{when the DNF size}} (the number of terms in the smallest DNF representation) \textit{\textbf{of specific monotone overapproximations of the transition relation are small, then $\Lambda$-PDR terminates after a small number of iterations}} (\Cref{thm:abstract-diamter-bound,thm:abstract-hyperdiamter-bound}).
The central idea is to relate $\Lambda$-PDR to \textit{\textbf{exact forward reachability in an ``abstract'' transition system}}, and then bound the number of iterations using techniques for analyzing the diameter of transition systems.

Let us consider an example where~\Cref{thm:abstract-diamter-bound} derives an efficient convergence bound (specifically, a linear bound, as opposed to the potential exponential number of iterations) through a syntactic analysis of the transition system in question.
Consider the same example of~\Cref{fig:skip-counter}, but with additional transitions that ``bounce back'' from $\vec{x}=01\ldots11$ to $\vec{x}=00\ldots00$ and from $\vec{x}=11\ldots11$ to any $\vec{x}=10\ldots010\ldots0$ (number with msb and other variable $1$).  \yotamsmall{turn every subset of $1$ variables to $0$---not so easy to describe in a transition relation formula. where every number can bounce back---commented out}
(The new transitions have no effect on the behavior of either exact forward reachability or $\Lambda$-PDR;\footnote{
 \yotamsmall{new} The algorithms are affected by the new transitions when they arrive to the transitions' pre-states, but at this point both algorithms will have already arrived to the smaller numbers in the post-states of the new transitions, resulting in the same frames as without these transitions.
 }
 we explain below why this is needed to obtain a good bound through~\Cref{thm:abstract-diamter-bound}.\footnote{
	\yotamsmall{new} An earlier version of this paper stated that \Cref{thm:abstract-diamter-bound} proves $\Lambda$-PDR's convergence in a linear number of frames for the system in~\Cref{fig:skip-counter}. However, this is incorrect; see~\Cref{ex:skip-counter-pure-exponential}. A similar error is now corrected also in~\Cref{ex:multiskip-counter-poly}.
})

In this transition system, we bound the number of iterations of $\Lambda$-PDR by applying~\Cref{thm:abstract-diamter-bound} to the part of $\tr$ restricted to states where $\vec{y}=0\ldots0,z=0$ (this is valid because both the transitions of $\tr$ and monotonization
w.r.t.\ $\bkwrch{k}$, $\monox{\cdot}{\vec{x}=10\ldots0,\vec{y}=1\ldots1,z=1}$, leave $\vec{y}=0\ldots0,z=0$ unchanged---see~\Cref{lem:diameter-bound-project-irrelevant}). %
The transition relation $\tilde{\tr} = \restrict{\tr}{\vec{y}\gets00\ldots00, z\gets0}$ can be written in DNF (as a double-vocabulary formula, with unprimed variables for the pre-state and primed variables for the post-state) as the disjunction of individual transitions, as appears on the left hand side here (colors and boxes should be ignored at this point): %
\newcommand*{\bbox}{%
  \tcboxmath[colback=white, colframe=black, size=fbox, arc=0pt, boxrule=0.4pt]%
}
\newcommand*{\nbox}{%
  \tcboxmath[colback=white, colframe=white, size=fbox, arc=0pt, boxrule=0.4pt]%
}
\small
\begin{align}
	\tilde{\tr}
		=
		  \label{eq:it:00start}
		  &\nbox{(\vec{x}=\litabs{0}00\ldots0000 \land \vec{x}'=0\litabs{00\ldots000}1)}
		  && \absr{\tr} = \monox{\tilde{\tr}}{\vec{x}=011\ldots11,\vec{x}'=100\ldots00} =
		\\
		\lor
		  &\nbox{(\vec{x}=\litabs{0}00\ldots000\litabs{1} \land \vec{x}'=0\litabs{00\ldots00}1\litabs{0})}
		  &&
	    \\
	    \lor
		  &
		  \nbox{(\vec{x}=\litabs{0}00\ldots00\litabs{1}0 \land \vec{x}'=0\litabs{00\ldots00}11)}
		  &&
	    \\
		\lor
		  &\nbox{(\vec{x}=\litabs{0}00\ldots00\litabs{11} \land \vec{x}'=0\litabs{00\ldots0}1\litabs{00})}
		  &&
	    \\
	    \label{eq:it:00end}
	    \lor
		  &
		  \nbox{\ldots}
		  &&
		\\
		\label{eq:it:0bounce}
		\lor
		  &\bbox{(\vec{x}=\litabs{011\ldots1111} \land \vec{x}'=0\litabs{00\ldots0000})}
		  && \quad \lor x'_n=0
		\\
		\label{eq:it:01skip}
		\lor
		  &\bbox{(\vec{x}=\litabs{011\ldots1111} \land \vec{x}'=\litabs{100\ldots000}1)}
		  && \quad \lor x'_0=1
		\\
		\lor
		\label{eq:it:11start}
		  &\nbox{(\vec{x}=100\ldots000\litabs{1} \land \vec{x}'=\litabs{100\ldots00}1\litabs{0})}
		  &&
		\\
		\lor
		  &\nbox{(\vec{x}=100\ldots00\litabs{1}0 \land \vec{x}'=\litabs{100\ldots00}11)}
		  &&
		\\
		\lor
		  &
		  \nbox{\ldots}
		  &&
		\\
		\lor
		\label{eq:it:11end}
		  &\nbox{(\vec{x}=1\litabs{11\ldots111}0 \land \vec{x}'=\litabs{1}11\ldots1111)}
		  &&
		\\
		\lor
		\label{eq:it:10wrap}
		  &
		  \nbox{(\vec{x}=1\litabs{11\ldots1111} \land \vec{x}'=0\litabs{00\ldots0000})}
		  &&
		\\
		\label{eq:it:1bounce-start}
		\lor
		  &\nbox{(\vec{x}=1\litabs{11\ldots1111} \land \vec{x}'=\litabs{100\ldots000}1)}
		  && 
		\\
		\label{eq:it:1bounce-start-really}
		\lor
		  &\bbox{(\vec{x}=1\litabs{11\ldots1111} \land \vec{x}'=\litabs{100\ldots00}1\litabs{0})}
		  && \quad \lor (x_n = 1 \land x'_1 = 1)
		\\
		\lor
		  &\bbox{(\vec{x}=1\litabs{11\ldots1111} \land \vec{x}'=\litabs{100\ldots0}1\litabs{00})}
		  && \quad \lor (x_n = 1 \land x'_2 = 1)
		\\
		\lor
		  &\bbox{(\vec{x}=1\litabs{11\ldots1111} \land \vec{x}'=\litabs{100\ldots}1\litabs{000})}
		  && \quad \lor (x_n = 1 \land x'_3 = 1)
		\\
		\lor
		  \label{eq:it:1bounce-end}
		  &\bbox{\ldots}
		  && \quad \lor \ldots
\end{align}
\normalsize
\yotamsmall{omitting: Obviously, the number of terms in $\tilde{\tr}$ is $\Omega\left(2^n\right)$.}
To compute a bound for the number of frames in $\Lambda$-PDR,
we perform a monotonization of the (two-vocabulary) transition relation.
Recall that in this example $\bkwrch{k}$ consists of a single cube, in which case we need only consider one monotonization (the case of more complex syntactic forms of $\bkwrch{k}$ is discussed later):
examine the monotonization that omits literals that agree with $\bkwrch{k}$ in the post-state, and \emph{conversely} in the pre-state,
$\absr{\tr} = \monox{\tilde{\tr}}{\vec{x}=01\ldots11,\vec{x}'=10\ldots00}$. The literals in $\tilde{\tr}$ that are omitted in $\absr{\tr}$ appear colored.
As we show in~\Cref{thm:absract-reach}, the resulting transition relation captures the behavior of $\Lambda$-PDR:
the set of $i$-reachable states of $\absr{\tr}$
matches the $i$'th frame of the %
Kleene iterations with the best transformer for $\tr$ in the $\bkwspan{k}$ domain.
Hence, \textbf{\textit{bounds on the diameter of the abstract system result in bounds on the number of frames}} of $\Lambda$-PDR.

To bound the diameter, we consider the DNF size $\absr{\tr}$.
The monotonization term-by-term from $\tilde{\tr}$
creates many redundant terms; the terms that originate from the transitions marked by boxes above subsume all the others.\footnote{
	The term arising from the ``bounce back'' transition with msb $0$ in~\cref{eq:it:0bounce}
	subsumes all other terms that originate from transitions where the msb is $0$ in both the pre-state and the post-state (\crefrange{eq:it:00start}{eq:it:00end}), as well as the term originating from the ``wraparound'' transition in~\cref{eq:it:10wrap};
	the term arising from the ``skip'' transition in~\cref{eq:it:01skip} %
	\yotamsmall{changed}subsumes the term originating from the transition in~\Cref{eq:it:1bounce-start};
	and the terms arising from ``bounce back'' transitions with msb $1$ in~\crefrange{eq:it:1bounce-start}{eq:it:1bounce-end} (\yotamsmall{new}including~\Cref{eq:it:1bounce-start} which is subsumed by~\Cref{eq:it:01skip})
	subsume all the terms that arise from transitions where the msb is $1$ both in the pre-state and the post-state (\crefrange{eq:it:11start}{eq:it:11end}). 	
} %
This generates a DNF representation of $\absr{\tr}$ with linear number of terms (appearing in the right-hand side above)---even though the original $\tilde{\tr}$ has an exponential number of terms in its DNF representation.
By~\Cref{thm:abstract-diamter-bound}
we deduce from
the linear DNF size of $\absr{\tr}$ that
$\Lambda$-PDR converges in at most a linear number of frames. %

One way to think about the difference between $\tr,\absr{\tr}$ is by the way transitions in $\tr$ give rise to the transitions in $\absr{\tr}$, illustrated in~\Cref{fig:arrow-subsumption}.
A transition of $\absr{\tr}$ can \emph{abstract} and move away from $\bkwrch{k}$, then follow a \emph{concrete} transition of $\tr$, and from the resulting post-state again \emph{abstract} and move in the direction away from $\bkwrch{k}$.
In this way, it may be possible for $\absr{\tr}$ to use the transition $(\sigma_1,\sigma'_1)$ in order to arrive from $\sigma_2$ to $\sigma'_2$, even if $(\sigma_2,\sigma'_2)$ were not a transition of $\tr$ (see~\Cref{fig:arrow-subsumption}).
When this is the case for $(\sigma_1,\sigma'_1),(\sigma_2,\sigma'_2) \in \tr$, the transitions of $\absr{\tr}$ that use the concrete transition $(\sigma_2,\sigma'_2)$ are also possible using the concrete transition $(\sigma_1,\sigma'_1)$; hence
the term generated from $(\sigma_2,\sigma'_2)$ can be discarded in the monotonization of $\tr$ to obtain $\absr{\tr}$, because it is subsumed by the term generated from $(\sigma_1,\sigma'_1)$. %
Roughly, \Cref{thm:abstract-diamter-bound} shows that $\Lambda$-PDR converges rapidly whenever there is a small number of transitions that subsume the others, by going from a pre-state $\sigma$ that is ``very far'' from $\bkwrch{k}$ in Hamming distance compared to the pre-states of other transitions, to the post-state $\sigma'$ that is ``very close'' to $\bkwrch{k}$ compared to the post-states of other transitions.
This is an intuition for how a small $\dnfsize{\absr{\tr}}$ can arise from the monotonization of the fully-expanded DNF representation of $\tr$.
(The starting point for monotonization can also be a more succinct DNF representation of $\tr$, in which case the intuition for an even shorter DNF representation of $\absr{\tr}$ is similar.) %

If we were not to add the ``bounce back'' transitions to the example of~\Cref{fig:skip-counter}, still monotonization of the transition relation produces an abstract transition system whose reachable states coincide with $\Lambda$-PDR's frames, and whose diameter corresponds to the number of iterations in which $\Lambda$-PDR converges. However, in that case the bound of~\Cref{thm:abstract-diamter-bound} is poor, because in this case the DNF-size is a poor estimate of the abstract system's diameter (see~\Cref{ex:skip-counter-pure-exponential}).

For when $\bkwrch{k}$ consists of multiple cubes, we generalize \Cref{thm:abstract-diamter-bound} to~\Cref{thm:abstract-hyperdiamter-bound}, bounding the number of frames by a product of monotonizations of $\tr$ and $\Init$. In the proof, the construction involves not an ordinary transition system, but a \emph{hyper}transition system: the hypertransition relation $\absr{\tr}$ arrives through concrete transitions to a \emph{set} of states, and abstracts from them to a state ``protected'' by that set, because the abstraction requires a ``protector'' state from every state in $\bkwrch{k}$ (see~\Cref{fig:protected-hull}).
A similar diameter bound using the DNF size of the hypertransition relation $\absr{\tr}$ applies.
We show that $\absr{\tr}$ can be written as a conjunction of per-cube monotonizations, leading to a bound by the product of DNF sizes of monotonizations of $\tr$ and $\Init$ (see~\Cref{sec:hyper-all}).

This technique does not explain rapid convergence of $\Lambda$-PDR in full generality, but provides one explanation for how the abstraction can bring this about.

\subsection{PDR, Revisited}
Through $\Lambda$-PDR, we have shown how PDR's frames perform an abstract interpretation in a domain founded on the monotone theory, and how such an abstraction can lead to faster convergence. We observe that these important characteristics of PDR are concealed in a simple property of PDR's frames: that they can be written in CNF so that every clause excludes at least one state from $\bkwrch{N}$ (\Cref{lem:pdr-also}).
In the monotone theory from above this reads that for every $1 \leq i \leq N$,
\iflong
\else
\vspace{-0.2cm}
\fi
\begin{tcolorbox}[boxsep=0pt]
		\begin{enumerate}
			\setcounter{enumi}{\value{overview-frame-props}}
			\item \label{it:frames-mbasis} $\Frame_i \in \mspan{\bkwrch{N}}$.
		\end{enumerate}
\end{tcolorbox}
\iflong
\else
\vspace{-0.2cm}
\fi
\noindent
The frames of $\Lambda$-PDR are the least set of states that satisfy this property together with properties~\ref{it:frames-start}--\ref{it:frames-end} from~\Cref{sec:overview-frame-props} (\Cref{lem:lambda-frames-minimality}), and the frames of PDR overapproximate them (\Cref{cor:lambda-pdr-underapproximates-pdr}).
Property~\ref{it:frames-mbasis} is the regularization in our abstract domain (\Cref{sec:ai}), and we have shown that it can lead to faster convergence than exact post-image computations---although PDR does not necessarily converge in the same number of frames as $\Lambda$-PDR, due to its additional overapproximation and heuristics.
The fact that PDR's frames are not the least to satisfy the above properties can
\yotamforlater{This could also lead PDR to be ``distracted'' and converge slower than the abstraction of $\Lambda$-PDR. However,}%
have some benefits. We show two:
\begin{itemize}
	\item \textbf{Faster convergence}: In some cases $\Lambda$-PDR performs little or no abstraction over exact forward reachability, but the fact that PDR only samples a subset of the possible lemmas can guarantee fast convergence. We show an example where $\Lambda$-PDR requires an exponential number of frames, whereas a linear number always suffices for standard PDR.

	\item \textbf{Frame size}: $\Lambda$-PDR's frames may be (needlessly) complex to represent as a formula. We show an example where some frames of $\Lambda$-PDR necessarily have an exponential DNF or CNF size, whereas standard PDR can converge in the same number of frames that include only a small number of important lemmas.
\end{itemize}
\yotamforlater{However, in some cases PDR can be led astray whereas $\Lambda$-PDR's more systematic exploration converges quickly.}

\subsubsection{Discussion: Additional PDR Features}
Our study focuses on what is, in our view, the most basic version of PDR. Our approach provides an interesting starting point to a discussion of two common, more advanced features of PDR.

\para{Other forms of generalization}
\label{sec:overview-inductive-generalization}
Inductive generalization~\cite{ic3} minimizes lemmas using a stronger criterion:
a lemma $c$ can be learned in $\Framepdr_{i+1}$ if it is inductive \emph{relative} to $\Frame_i$---whether $\postimage{\tr}{\Framepdr_i \land c} \implies c$, namely, checking whether $c$ holds in the post-state while also assuming $c$ in the pre-state. At first sight, this feature is not present in $\Lambda$-PDR, which uses the standard check (\Cref{alg:eepdr}, \cref{ln:eepdr-lemma-check}). Surprisingly, lemmas that PDR can generate using inductive generalization are also present in $\Lambda$-PDR (with $k=N$). This is a consequence of the fact that PDR with inductive generalization still satisfies properties~\ref{it:frames-start}--\ref{it:frames-end}~\cite{ic3,pdr}, as well as property~\ref{it:frames-mbasis}. The optimization of inductive generalization becomes important only when lazily backtracking to refine previous frames.
Other techniques, such as ternary simulation~\cite{pdr}, propagate sets of states to block together. If any of the states is in $\bkwrch{k}$, the resulting lemma is present also in $\Lambda$-PDR.

\para{May-counterexamples}
\label{sec:overview-may-cexs}
Some variants of PDR produce lemmas by blocking may counterexamples~\cite{DBLP:conf/fmcad/GurfinkelI15} that are not necessarily backward reachable, as a way to encourage pushing existing lemmas to later frames.
In $\Lambda$-PDR, all admissible lemmas from $\bkwspan{k}$ are always included, hence may counterexamples are not useful for pushing such lemmas.
However, may counterexamples also mean that lemmas no longer necessarily block states in $\bkwrch{N}$, which could be beneficial if a large $N$ is required to have an inductive invariant $I \in \bkwspan{N}$.
(It is a necessary condition PDR; in $\Lambda$-PDR it is both necessary and sufficient, see~\Cref{lem:eepdr-lfp}.)
This could be thought of as (heuristically) increasing the set $\bkwrch{k}$. In $\Lambda$-PDR, this results in a richer abstract domain that includes more inductive invariants but leads to tighter frames with less overapproximation. The theoretical ramifications of this beyond $\Lambda$-PDR merit more study.

\subsubsection{Outline}
The rest of the paper is organized as follows:
\Cref{sec:prelim} introduces preliminary notation.
\Cref{sec:monotone} introduces notions from the monotone theory and establishes their connection to $\Lambda$-PDR.
\Cref{sec:ai} proves the connection to abstract interpretation.
\Cref{sec:abstract-diameter-all} develops a bound on the number of iterations of $\Lambda$-PDR for the case that $\bkwrch{k}=\bkcube$ a single cube, and \Cref{sec:hyper-all} generalizes these results to arbitrary $\bkwrch{k}$. %
\Cref{sec:itp-friends} contrasts forward reachability in $\Lambda$-PDR with exact forward reachability and a dual interpolation-based algorithm.
\Cref{sec:vs-pdr} compares $\Lambda$-PDR to standard PDR.
\Cref{sec:related} discusses related work, and~\Cref{sec:conclusion} concludes.  
\section{Preliminaries}
\label{sec:prelim}
We consider the safety of transition systems defined over propositional vocabularies. %

\para{States, transition systems, inductive invariants}
Given a vocabulary $\voc = \set{p_1,\ldots,p_n}$ of $n$ Boolean variables, a \emph{state} is a \emph{valuation} to $\voc$. The set of states over $\voc$ is denoted $\States$.
If $x$ is a state, $x[p]$ is the value ($\true/\false$ or $1/0$) that $x$ assigns to the variable $p \in \voc$.
A \emph{transition system} is a triple $(\Init,\tr,\Bad)$ where $\Init,\Bad$ are formulas over $\voc$ denoting the set of initial and bad states (respectively), and the \emph{transition relation} $\tr$ is a formula over $\voc \uplus \voc'$, where $\voc' = \{ \prop' \mid \prop \in \voc\}$ is a copy of the vocabulary used to describe the post-state of a transition.
If $\tilde{\voc},\tilde{\voc}'$ are distinct copies of $\voc$, $\tr[\tilde{\voc},\tilde{\voc}']$ denotes the substitution in $\tr$ of each $p \in \voc$ by its corresponding in $\tilde{\voc}$ and likewise for $\voc',\tilde{\voc'}$.
Given a set of states $S \subseteq \States$, the post-image $\tr(S) = \set{\sigma' \mid \exists \sigma \in S. \ (\sigma,\sigma') \models \tr}$.
The reflexive post-image is $\postimage{\tr}{S} = \tr(S) \cup S$.
The reachability diameter of a system is the least $s$ s.t.\ if $\sigma$ is reachable from $\Init$ by $\tr$ in any number of steps, it is reachable in at most $s$ steps. %
\begin{changebar}
The set of \emph{$k$-backward reachable states}---those that can reach a state in $\Bad$ along some execution of length at most $k$---are $\bkwrch{k} = \set{\sigma \ | \ \bmc{\tr}{\set{\sigma}}{k} \cap \Bad \neq \emptyset}$.
\end{changebar}
A transition system is \emph{safe} if all the states that are reachable from $\Init$ via any number of steps of $\tr$ satisfy $\neg \Bad$. %
An \emph{inductive invariant} is a formula $I$ over $\voc$ such that
\begin{inparaenum}[(i)]
	\item $\Init \implies I$,
	\item $I \land \tr \implies I'$, and
	\item $I \implies \neg\Bad$, where $I'$ %
denotes the result of substituting each $\prop \in \voc$ for $\prop' \in \voc'$ in $I$,
\end{inparaenum}
and $\varphi \implies \psi$ denotes the validity of the formula $\varphi \to \psi$. In the context of propositional logic, a transition system is safe iff it has an inductive invariant. %

\para{CNF, DNF, and cubes}
A \emph{literal} $\ell$ is a variable $p$ or its negation $\neg p$.
A \emph{clause} $c$ is a disjunction of %
literals. %
The empty clause is $\false$.
A formula is in \emph{conjunctive normal norm (CNF)} if it is a conjunction of clauses.
A \emph{cube} or \emph{term} $d$ is a conjunction of a consistent set of literals; at times, we also refer directly to the set and write $\ell \in d$. The empty cube is $\true$.
A formula is in \emph{disjunctive normal form (DNF)} if it is a disjunction of terms.
The \emph{domain}, $\cubdom{d}$, of a cube $d$ is the set of variables that appear in it (positively or negatively).
Given a state $\sigma$, we use the state and the (full) cube that consists of all the literals that are satisfied in $\sigma$ interchangeably; the only satisfying valuation of the cube is $\sigma$.
We identify a formula with the set of its valuations, and a set of valuations with an arbitrary formula that represents it, chosen arbitrarily (which always exists in propositional logic).
\section{The Monotone Theory for $\Lambda$-PDR}
\label{sec:monotone}
In this section, we present the monotone theory by~\citet{DBLP:journals/iandc/Bshouty95} and our extensions, and use it to derive the relation between successive frames in $\Lambda$-PDR (\Cref{eq:overview-next-frame}).
\Cref{sec:monotone-background} defines \emph{least $b$-monotone overapproximations} $\monox{\cdot}{b}$.
\Cref{sec:monotone-hull} defines the \emph{monotone hull} $\mhull{\cdot}{B}$ w.r.t.\ a set of states $B$, which is a conjunction of monotone overapproximations, and then relates this to $\Lambda$-PDR.
\iflong
\else
All omitted proofs appear in the extended version~\cite{extendedVersion}.
\fi

\subsection{Monotone Overapproximations}
\label{sec:monotone-background}
Our definitions and claims concerning $b$-monotone overapproximations %
generalize \citet{DBLP:journals/iandc/Bshouty95} by considering a partial cube $b$, and coincide with the original in the case of a full cube.
\begin{changebar}
\begin{definition}[$b$-Monotone Order~\cite{DBLP:journals/iandc/Bshouty95}]
\label{def:b-monotone-order}
Let $b$ be a cube. We define a partial order over states where $v \leq_b x$ when $x,v$ agree on all variables not present in $b$, and $x$ disagrees with $b$ on all variables on which also $v$ disagrees with $b$:
$\forall p \in \voc. \ x[p] \neq v[p] \mbox{ implies } p \in \cubdom{b} \land v[p]=b[p]$.
\end{definition}
Geometrically, $v \leq_b x$ indicates that $x$ is ``farther away'' from $b$ in the Hamming cube than $v$ from $b$, namely, that there is a shortest path w.r.t.\ Hamming distance between $x$ and its projection onto $b$ that goes through $v$.
\begin{definition}[$b$-Monotonicity~\cite{DBLP:journals/iandc/Bshouty95}]
\label{def:b-monotonicity}
A formula $\psi$ is $b$-monotone for a cube $b$ if
$
\forall v \leq_b x. \ v \models \psi \mbox{ implies } x \models \psi.
$
\end{definition}
That is, if $v$ satisfies $\psi$, so do all the states that are farther away from $b$ than $v$.
For example, if $\psi$ is $000$-monotone and $100 \models \psi$, then because $100 \leq_{000} 111$ (starting in $100$ and moving away from $000$ can reach $111$), also $111 \models \psi$.
In contrast, $100 \not\leq_{000} 011$ (the same process cannot flip the $1$ bit that already disagrees with $000$), so $011$ does not necessarily belong to $\psi$.

The importance of $b$-monotonicity for us is that if $\sigma \models b$, then, as we shall see (in~\Cref{thm:mhull-conjunctive}), any clause $c \subseteq \neg \sigma$ is a $b$-monotone formula.
The reason is that if $v \models c$ and we flip variables to disagree more with $b$, we can only make $v$ satisfy more literals of $c$ than before: all the variables in $b$ appear in the opposite polarity in $\neg \sigma$ and hence also in $c$, so flipping them in $v \models c$ to disagree with $b$ makes them agree with $c$ even more, and the result also satisfies $c$. %

Further, the lemmas that PDR learns are not just clauses, but clauses that overapproximate a certain set, hence they are $b$-monotone overapproximations. There are many $b$-monotone overapproximations, but our main workhorse is the least such formula:
\begin{definition}[Least $b$-Monotone Overapproximation~\cite{DBLP:journals/iandc/Bshouty95}]
\label{def:monox}
Given a formula $\varphi$ and a cube $b$, the \emph{least $b$-monotone overapproximation} of $\varphi$ is a formula $\monox{\varphi}{b}$ defined by
\begin{equation*}
	x \models \monox{\varphi}{b} \mbox{ iff } \exists v. \ v \leq_b x 	\land	 v \models \varphi.
\end{equation*}
\end{definition}
For example, if $100 \models \varphi$, then $100 \models \monox{\varphi}{000}$ because $\monox{\varphi}{000}$ is an overapproximation, and hence $111 \models \monox{\varphi}{000}$ because it is $000$-monotone, as above.
Here, thanks to minimality, $011$ does not belong to $\monox{\varphi}{000}$, unless $000$, $001$, $010$, or $011$ belong to $\varphi$.

$\monox{\varphi}{b}$ is a well-defined overapproximation of $\varphi$. Its main significance for learning theory is that it can be computed efficiently (through the DNF representation we also show below), obtaining the original $\varphi$ as the conjunction of least $b$-monotone overapproximations (with different $b$'s).
Surprisingly, the same overapproximation is related to PDR; we show below that $\monox{\varphi}{b}$ is exactly the conjunction of all the clauses $c$ that overapproximate $\varphi$ and can arise from blocking a state in $b$ (\Cref{thm:mhull-conjunctive}), matching generalization in the construction of frames of $\Lambda$-PDR (\Cref{cor:lambda-pdr-underapproximates-pdr}).

\end{changebar}

A technical observation that will prove useful several times is that $\monox{\cdot}{b}$ is a monotone operator:
\begin{lemma}
\label{lem:bshouty-monox-monotone}
If $\varphi_1 \implies \varphi_2$ then $\monox{\varphi_1}{b} \implies \monox{\varphi_2}{b}$.
\end{lemma}
\toolong{
\begin{proof}
Immediate from the definition of $\monox{\cdot}{b}$.
\end{proof}	
}

\para{Disjunctive form}
The monotone overapproximation can be related to a DNF representation of the original formula, a fact that we use extensively in \Cref{sec:abstract-diameter-all} and also when we analyze $\Lambda$-PDR on specific examples.
Starting with a DNF representation of $\varphi$, we can derive a DNF representation of $\monox{\varphi}{b}$
by dropping in each term the literals that agree with $b$.
Intuitively, a ``constraint'' that $\sigma \models \ell$ in order to have $\sigma \models \monox{t}{b}$ where $\ell$ agrees with $b$ is dropped because if $\sigma \models \monox{t}{b}$ then flipping a bit $\sigma$ to disagree with $b$ results in a state $\tilde{\sigma}$
such that also
$\tilde{\sigma} \models \monox{t}{b}$, as $\sigma \leq_{b} \tilde{\sigma}$.

\begin{lemma}[Generalization of~\citet{DBLP:journals/iandc/Bshouty95}, Lemma 1(7)]
\label{lem:bshouty-mon-mindnf}
Let $\varphi = t_1 \lor \ldots \lor t_m$ in DNF. Then the monotonization $\monox{\varphi}{b} \equiv \monox{t_1}{b} \lor \ldots \lor \monox{t_m}{b}$ where $\monox{t_i}{b} \equiv t_i \setminus b = \bigwedge \set{\ell \in t_i \land \ell \not\in b}$.
\end{lemma}
\toolong{
\begin{proof}
First we argue that for any term $t$, $\monox{t}{b} \equiv t \setminus b$. Denote the rhs by $\psi$.
Let $x \in \monox{t}{b}$. Then there is $v \models t$ such that $v \leq_b x$. Let $\ell \in t$; $v \models \ell$. Consider a literal $\ell \in \psi$, then $\ell \in t$ and $\ell \not\in b$. Since $v \models t$, the former means that in particular $v \models \ell$. The latter means that to satisfy $v \leq_b x$ necessarily $v,x$ agree on the variable in $\ell$, and hence also $x \models \ell$. This proves $\monox{t}{a} \implies \psi$.
For the other direction, let $x \models \psi$. Let $v$ be obtained from $x$ by setting every variable $p_i \in \cubdom{b}$ that does not appear in $\psi$ to disagree with the corresponding value in $b$; then $v \leq_a x$. Now $v \models \psi$ (since these variables do not appear in $\psi$), and, furthermore, $v \models \ell$ for every $\ell \in t$ that was dropped from $t$ to $\psi$, because $v$ disagrees with $b$ on those literals, which are those that appear in $b$ in a negated form compared to $\ell$. Overall, $v \models t$, which implies $x \models \monox{t}{a}$.

We now claim, more generally, that $\monox{\psi_1 \lor \psi_2}{b} \equiv \monox{\psi_1}{b} \lor \monox{\psi_2}{a}$:
Let $x \models \monox{\psi_1 \lor \psi_2}{b}$. Then there is $v \models \psi_1 \lor \psi_2$ such that $x \leq_b x$. If $v \models \psi_1$, by definition we must have $x \models \monox{\psi_1}{b}$ and in particular $x \models \monox{\psi_1}{b} \lor \monox{\psi_2}{b}$; similarly for $\psi_2$. This shows $\monox{\psi_1 \lor \psi_2}{b} \implies \monox{\psi_1}{b} \lor \monox{\psi_2}{b}$.
As for the other direction, let $x \models \monox{\psi_1}{b} \lor \monox{\psi_2}{b}$. Without loss of generality, assume $x \models \monox{\psi_1}{b}$. Then there is $v \models \psi_1$, and in particular $v \models \psi_1 \lor \psi_2$, such that $v \leq_b x$. So we must have $x \models \monox{\psi_1 \lor \psi_2}{b}$.
\end{proof}
}

A corollary provides a canonical (inefficient) disjunctive form for $\monox{\varphi}{b}$:
\begin{corollary}
\label{lem:monox-disjunction-cubes}
Given a state $v$, we denote
$
	\cubemon{v}{b} \eqdef \monox{v}{b} =
					 \bigwedge{\set{p \ | \ v[p]=\true, \, p \not\in b}} \land
					 \bigwedge{\set{\neg p_i \ | \ v[p]=\false, \, \neg p \not\in b}}.
$
Then $\monox{\varphi}{b} \equiv \bigvee_{v \models \varphi}{\cubemon{v}{b}}$.
\end{corollary}
\toolong{
\begin{proof}
Apply~\Cref{lem:bshouty-mon-mindnf} to the representation of $\varphi$ as the disjunction of all satisfying states.
\end{proof}
}

\subsection{Monotone Hull}
\label{sec:monotone-hull}
We now define the monotone hull, which is a conjunction of $b$-monotone overapproximations over all $b$ from a fixed set of states $B$ (in the case of $\Lambda$-PDR, $B = \bkwrch{k}$). We start with the definition that uses a conjunction over monotone-overapproximations w.r.t.\ individual states, and then extend this to the union of (partial) cubes.
\begin{definition}[Monotone Hull]
The \emph{monotone hull} of a formula $\varphi$ w.r.t.\ a set of states $B$ is $\mhull{\varphi}{B} = \bigwedge_{b \in B}{\monox{\varphi}{b}}$. %
\end{definition}

The monotone hull can be simplified to use a succinct DNF representation of the basis $B$ instead of a conjunction over all states.
(This is the motivation for generalizing $\monox{\cdot}{b}$ to a cube $b$ in \Cref{sec:monotone-background}.)
\begin{lemma}
\label{lem:mhull-dnf-base}
If $B = b_1 \lor \ldots \lor b_m$ and $b_1,\ldots,b_m$ are cubes, then $\mhull{\varphi}{B} \equiv \monox{\varphi}{b_1} \land \ldots \land \monox{\varphi}{b_m}$.
\end{lemma}
\toolong{
\begin{proof}
It follows from the definition that $\mhull{\cdot}{B}$ distributes over union in $B$.
Hence, $\mhull{\varphi}{B} \equiv \mhull{\varphi}{b_1} \land \ldots \land \mhull{\varphi}{b_m}$.
Further, $\mhull{\varphi}{b_i} = \bigwedge_{\sigma_b \in b_i}{\mhull{\varphi}{\set{\sigma_b}}} = \bigwedge_{\sigma_b \in b_i}{\monox{\varphi}{\sigma_b}}$.
It remains to argue that $\monox{\varphi}{b_i} \equiv \bigwedge_{\sigma_b \in b_i}{\monox{\varphi}{\sigma_b}}$.

\noindent
$\subseteq$:
By definition, if $\sigma \models \monox{\varphi}{b_i}$ then there is $x \models \varphi$ s.t.\ $x \leq_{b_i} \sigma$. In particular, $x \leq_{\sigma_b} \sigma$ for every $\sigma_b \models b_i$ (because $\sigma_b$ agrees with all the literals in $b_i$), and hence $\sigma \models \monox{\varphi}{\sigma_b}$ for every such $\sigma_b$.

\noindent
$\supseteq$:
Suppose that $\sigma$ is a model of the rhs. Let $\sigma_b$ be the state obtained from $\sigma$ by setting the variables present in $b_i$ to be as in $b_i$ (geometrically, this the projection of $\sigma$ onto the cube $b_i$). We have $\sigma_b \models b_i$. Thus $\sigma \models \monox{\varphi}{\sigma_b}$, so there exists $x \models \varphi$ such that $x \leq_{\sigma_b} \sigma$. But because $\sigma_b$ agrees with $\sigma$ on all literals except those also present in $b_i$, this implies also that $x \leq_{b_i} \sigma$. Hence also $\sigma \models \monox{\varphi}{b_i}$.
\end{proof}
}
Note that when $B = b$ is a single cube, $\mhull{\varphi}{b} = \monox{\varphi}{b}$.

Our main technical observation, connecting the monotone hull to~\Cref{alg:eepdr}, is that the monotone hull has an equivalent CNF form, as the conjunction of all overapproxmating clauses that exclude a state in $B$:
\begin{theorem}
\label{thm:mhull-conjunctive}
\label{lem:monox-conjunction-clauses}
	$\mhull{\varphi}{B} \equiv \bigwedge{\set{c \ | \ \mbox{$c$ is a clause, } \varphi \implies c \mbox{, and } \exists b \in B. \, b \not\models c}}.$
\end{theorem}
\begin{proof}
\noindent
$\Longrightarrow$: Let $c$ be a clause as in the rhs. Then there exists $b \in B$ s.t.\ $b \not\models c$. It suffices to show that $\monox{\varphi}{b} \implies c$. Recall that $c$ is a disjunction of literals; since $b \not\models c$, all those literals are falsified in $b$. Hence, by~\Cref{lem:bshouty-mon-mindnf}, $\monox{c}{b} \equiv c$.
Now $\varphi \implies c$ (by the choice of $c$), and~\Cref{lem:bshouty-monox-monotone} yields $\monox{\varphi}{b} \implies \monox{c}{b} \equiv c$.

\noindent
$\Longleftarrow$:
Let $\sigma$ be a model of the rhs. We want to prove that $\sigma \models \monox{\varphi}{b}$ for every $b \in B$.
Assume otherwise.
Take $d$ as the conjunction of all literals that hold in both $\sigma$ and $b$ (if this set is empty, $d=\true$). Clearly $d$ is a term and $b,\sigma \models d$.
Take $c = \neg d$; then $c$ is a clause and $b \not\models c$.
It remains to show that $\varphi \implies c$, because then $c$ belongs to the rhs but $\sigma \not\models c$, in contradiction to the premise.
To see this, note that $c = \monox{\neg \sigma}{b}$ by~\Cref{lem:bshouty-mon-mindnf}---all the literals from the clause $\neg \sigma$ on which $b$ disagrees (because $b$ agrees on them with $\sigma$).
$\varphi \implies \neg \sigma$ because $\sigma \not\models \varphi$ (as $\sigma \not\models \monox{\varphi}{b}$, which is an overapproximation of $\varphi$). Hence \Cref{lem:bshouty-monox-monotone} implies $\monox{\varphi}{b} \implies \monox{\neg \sigma}{b} = c$, and in particular $\varphi \implies c$. The claim follows.
\end{proof}
We can now derive~\Cref{eq:overview-next-frame}:
\begin{corollary}
\label{lem:justify-overview-next-frame}
In $\Lambda$-PDR (\Cref{alg:eepdr}), $\Frame_{i+1} = \mhull{\postimage{\tr}{\Frame_i}}{\bkwrch{k}}$.
\end{corollary}
\begin{proof}
$\Frame_{i+1}$ is the conjunction formed by the process that for each $\sigma_b \in \bkwrch{k}$ (\cref{ln:eepdr-bmc} of~\Cref{alg:eepdr}) iterates in~\cref{ln:eepdr-for-clause} over all clauses that exclude $\sigma_b$, and conjoins $c$ if it overapproximates $\postimage{\tr}{\Frame_i}$ (\cref{ln:eepdr-lemma-check}). By~\Cref{thm:mhull-conjunctive} this is $\mhull{\postimage{\tr}{\Frame_i}}{\bkwrch{k}}$.
\end{proof}

\begin{example}
\label{ex:running-all-frames}
The above lemmas are the basis for our presentation in~\Cref{sec:overview-example-frame} of $\Frame_1$ on~\Cref{fig:skip-counter} as~\Cref{eq:running-frame-1}.
Let us now use these lemmas to describe later frames in that execution.
Recall that $\bkwrch{k}$ is the cube $\vec{x} = 10\ldots0 \land y_n y_{n-1} \ldots y_1 = 11\ldots1 \land z=1$, denote it by $b$.
For the next frame,
$\postimage{\tr}{\Frame_1}=\Frame_1 \lor (\vec{x}=10\ldots01 \land \vec{y}=0\ldots0 \land z=0)$. Then
\begin{align*}
		\Frame_2
	&\underset{\rm \Cref{thm:mhull-conjunctive}}{=}
		\mhull{\postimage{\tr}{\Frame_1}}{\bkwrch{k}}
	\underset{\rm \Cref{lem:mhull-dnf-base}}{=}
		\monox{\postimage{\tr}{\Frame_1}}{\bkcube}
	\\
	&\underset{\rm \Cref{lem:bshouty-mon-mindnf}}{=}
		\monox{\Frame_1}{\bkcube} \lor
		\monox{\vec{x}=10\ldots01 \land \vec{y}=0\ldots0 \land z=0}{\bkcube}
	\\
	&\underset{\rm \Cref{lem:bshouty-mon-mindnf}}{=}
		\Frame_1 \lor (x_0=1 \land \vec{y}=0\ldots0 \land z=0).
\end{align*}
For the next frame, $\postimage{\tr}{\Frame_2} = \Frame_2 \lor (x_0=0 \land \vec{y}=0\ldots0 \land z=0) \land \neg(\vec{x}=1\ldots0 \land \vec{y}=0\ldots0 \land z=0)$. This is equivalent to $(\vec{y}=0\ldots0 \land z=0) \land \neg(\vec{x}=1\ldots0 \land \vec{y}=0\ldots0 \land z=0)$, which is already $\bkcube$-monotone and hence this is also $\Frame_3 = \mhull{\postimage{\tr}{\Frame_2}}{\bkcube}$.
$\postimage{\tr}{\Frame_3}=\Frame_3$, and so $\mhull{\postimage{\tr}{\Frame_3}}{\bkwrch{k}} = \Frame_3$ (see~\Cref{lem:mhull-idempotence} below), and the algorithm converges.
\end{example}
\begin{example}
\label{ex:monotone-several-cubes}
In the previous example, $\bkwrch{k}$ consisted of a single cube. To exemplify the more general case, consider a system over $n$ variables $x_1,\ldots,x_n$, with $\Init = (x_1=\ldots=x_n=0)$, $\Bad$ the set of states with exactly one variable $1$, and $\tr$ that non-deterministically chooses some $i \neq j$ with $x_i=x_j=0$ and sets $x_i \gets 1$ and $x_j \gets 1$.
Take $k=0$. Then $\bkwrch{k}=\bkcube_1 \lor \ldots \lor \bkcube_n$ where $\bkcube_i$ is $x_i=1 \land \bigwedge_{j \neq i}{(x_j=0)}$.
After one step, $\postimage{\tr}{\Frame_0} = (0\ldots0) \lor \bigvee_{i_1 \neq i_2}{(x_{i_1}=x_{i_2}=1 \land \bigwedge_{j \not\in \set{i_1,i_2}}(x_j=0)})$ is the set of states where there are zero or two variables $1$.
Then for every $b_i$,
\begin{align*}
		\monox{\postimage{\tr}{\Frame_0}}{b_i}
	\underset{\rm \Cref{lem:bshouty-mon-mindnf}}{=}
			(x_i=1)
		\lor
			\left(\bigvee_{i_1 \neq i_2, i\not\in \set{i_1,i_2}}{x_{i_1}=x_{i_2}=x_i=1}\right)
		\lor
			\left(\bigvee_{i_1 \neq i_2, i=i_1}{x_{i_2}=1}\right)
	= \bigvee_{j}{(x_j=1)}.
\end{align*}
Hence,
$
		\Frame_1 \underset{\rm \Cref{thm:mhull-conjunctive}}{=} \mhull{\postimage{\tr}{\Frame_0}}{\bkwrch{k}} = \bigvee_{i}{\monox{\postimage{\tr}{\Frame_0}}{b_i}} = \bigvee_{j}{(x_j=1)}.
$
After another step, $\postimage{\tr}{\Frame_1}=\Frame_1$ and so $\mhull{\postimage{\tr}{\Frame_1}}{\bkwrch{k}} = \Frame_1$ (see~\Cref{lem:mhull-idempotence} below), and the algorithm converges.
\end{example}

\para{Additional lemmas}
Before proceeding, we state a few helpful, simple lemmas that we use later:
\begin{lemma}
\label{lem:mhull-overapproximation}
$\varphi \implies \mhull{\varphi}{B}$.
\end{lemma}
\toolong{
\begin{proof}
$\varphi \implies \monox{\varphi}{b}$ for every $b \in B$, from the definition of $b$-monotone overapproximation. Hence also $\varphi \implies \bigwedge_{b \in B}{\monox{\varphi}{b}}$.
\end{proof}
}
\begin{lemma}
\label{lem:mhull-monotonicity}
If $\varphi_1 \implies \varphi_2$ then $\mhull{\varphi_1}{b} \implies \mhull{\varphi_2}{b}$.
\end{lemma}
\toolong{
\begin{proof}
$\monox{\varphi_1}{b} \implies \monox{\varphi_2}{b}$ for every $b \in B$ by~\Cref{lem:bshouty-monox-monotone}, so if $\sigma \models \bigwedge_{b \in B}{\monox{\varphi_1}{b}}$ it also satisfies $\sigma \models \bigwedge_{b \in B}{\monox{\varphi_2}{b}}$.
\end{proof}
}
\begin{lemma}
\label{lem:mhull-idempotence}
The monotone hull is idempotent, that is, $\mhull{\mhull{\varphi}{B}}{B} \equiv \mhull{\varphi}{B}$.
\end{lemma}
\toolong{
\begin{proof}
We claim that for every $b \in B$, $\monox{\mhull{\varphi}{B}}{b} \equiv \monox{\varphi}{b}$, which implies $\mhull{\mhull{\varphi}{B}}{B} = \bigwedge_{b \in B}{\monox{\mhull{\varphi}{B}}{b}} \equiv \bigwedge_{b \in B}{\monox{\varphi}{b}} = \mhull{\varphi}{B}$.

Let $b \in B$. $\monox{\varphi}{b} \subseteq \monox{\mhull{\varphi}{B}}{b}$ because $\varphi \subseteq \mhull{\varphi}{B}$ (\Cref{lem:mhull-overapproximation}) and by~\Cref{lem:mhull-monotonicity}.
Since from the definition $\mhull{\varphi}{B} \subseteq \monox{\varphi}{b}$, again by~\Cref{lem:mhull-idempotence}
$\monox{\mhull{\varphi}{B}} \subseteq \monox{\monox{\varphi}{b}}{b}$ and $\monox{\cdot}{b}$ is idempotent from the definition as least $b$-monotone overapproximation.
\end{proof}
}

\section{Abstract Interpretation in The Monotone Span of $\bkwrch{k}$}
\label{sec:ai}

In this section, we cast $\Lambda$-PDR as an abstract interpretation in a new logical abstract domain of the \emph{monotone span of $\bkwrch{k}$}.
We first discuss abstract interpretation in general, and then develop the notion of a monotone span. We then define the abstract domain and show the connection to $\Lambda$-PDR.

\subsection{Background: Abstract Interpretation}
\label{sec:ai-background}
In this section, we review the basics of abstract interpretation~\cite{DBLP:conf/popl/CousotC77}; see~\cite{urban2015static,rival2020introduction} for complete and general presentations.
A \emph{complete join-semilattice} is a tuple $\langle D, \sleq, \join, \bot \rangle$ where $D$ is partially-ordered by $\sleq$, $\bigjoin X$ is the least upper bound of every $X \subseteq D$ ($\forall x \in X. \ x \sleq \bigjoin X$ and $X$ is the smallest w.r.t.\ $\sleq$ that satisfies this), and $\bot$ is the minimal element ($\forall x \in D. \ \bot \sleq x$).
A \emph{chain} is a sequence of elements from $D$ satisfying $x_1 \sleq x_2 \sleq \ldots$, and a \emph{strictly ascending chain} if additionally $x_i \neq x_{i+1}$ for every $i$. The lattice's \emph{height} is the maximal length of a strictly ascending chain.
We consider finite domains, where in particular the height is also finite.

A function $F: D \to D$ is Scott-continuous if for every chain $X=x_0,x_1,\ldots \subseteq D$ it holds $f(\bigjoin_{x \in X} x) = \bigjoin_{x \in X}{f(x)}$.
By the Knaster-Tarski theorem, such $F$ has a least fixed-point (lfp)---the least $x$ such that $f(x)=x$---and by Kleene's theorem it is $\bigjoin_{i \geq 0}{F^i(\bot)}$. When %
the domain's height is also finite, the sequence $\set{F^i(\bot)}_{i \geq 0}$ converges to the lfp.

In our setting, the \emph{concrete domain} is the join-semilattice powerset domain of the set of states, $\mathcal{C} = \langle D=2^{\States}, \subseteq, \cup, \emptyset \rangle$.
An \emph{abstract domain} is a join-semilattice $\mathcal{A} = \langle \abs{D}, \abs{\sleq}, \abs{\join}, \abs{\bot} \rangle$.
An \emph{abstraction function} is a monotone function $\alpha: D \to \abs{D}$, that is, $\forall S_1 \subseteq S_2 \in D. \ \alpha(S_1) \abs{\sleq} \alpha(S_2)$.
A \emph{conretization function} is a monotone function $\gamma: \abs{D} \to D$, that is, $\forall a_1 \abs{\sleq} a_2 \in \abs{D}. \ \gamma(a_1) \subseteq \gamma(a_2)$.
There is a \emph{Galois connections} between $(D, \subseteq)$ and $(\abs{D}, \abs{\sleq})$ through $(\alpha,\gamma)$, denoted $(D, \subseteq) \galois{\alpha}{\gamma} (\abs{D}, \abs{\sleq})$, if $\forall S \in D, a \in \abs{D}. \ \alpha(S) \abs{\sleq} a \Leftrightarrow S \subseteq \gamma(a)$.

Let $(\Init,\tr)$ be a transition system.
Define the concrete transformer $F_{\Init,\tr}: D \to D$ by $F_{\Init,\tr}(S) = \tr(S) \cup \Init$. It is Scott-continuous, since
for every increasing $\subseteq$-chain $X \subseteq D$, we have $F_{\Init,\tr}(\bigcup_{S \in X} S) = \bigcup_{S \in X}{F_{\Init,\tr}(S)}$.
Its fixed-point $\textit{lfp}(F_{\Init,\tr})$ is the set of reachable states of $(\Init,\tr)$.
The corresponding \emph{best abstract transformer} is given by $\abs{F}_{\Init,\tr}(a) = \alpha(F_{\Init,\tr}(\gamma(a)))$. $\abs{F}_{\Init,\tr}$ is also Scott-continuous, and there is
a \emph{fixed-point transfer} from $F$ to $\abs{F}$: $\textit{lfp}(\abs{F}_{\Init,\tr}) = \alpha(\textit{lfp}(F_{\Init,\tr}))$.
It follows that $\textit{lfp}(\abs{F}_{\Init,\tr})$ is the least abstraction of the set of reachable states, and it is obtained by the chain
\begin{equation*}
	\abs{\bot} \ \abs{\sleq} \ (\abs{F}_{\Init,\tr})^{1}(\abs{\bot}) \ \abs{\sleq} \ (\abs{F}_{\Init,\tr})^2(\abs{\bot}) \ \abs{\sleq} \ \ldots
\end{equation*}
at its convergence point in a finite $i^*$ with $(\abs{F}_{\Init,\tr})^{i^*} = (\abs{F}_{\Init,\tr})^{i^*+1}$ (due to the finite height of $\abs{D}$).
We call the chain the \emph{Kleene iterations with the best transformer}, and overall it converges to the most precise sound inductive invariant in the abstract domain.
Another way to phrase the same chain, denoting $\xi_i=(\abs{F}_{\Init,\tr})^{i}(\abs{\bot})$, is by
\begin{align*}
	\xi_1 = \alpha(\Init) \qquad \qquad \xi_{i+1}=\xi_i \mathbin{\abs{\join}} \alpha(\tr(\gamma(\xi_i))=\alpha(\postimage{\tr}{\gamma(\xi_i)}),
\end{align*}
because
$\xi_i \abs{\sleq} \xi_{i+1}$ implies that $\gamma(\xi_i) \subseteq \gamma(\xi_{i+1})$ and therefore
$\xi_{i+1}=\abs{F}_{\Init,\tr}(\xi_i) = \abs{F}_{\Init,\tr}(\xi_i) \mathbin{\abs{\join}} \alpha(\tr(\gamma(\xi_i)) \cup \Init) = \alpha(\tr(\gamma(\xi_i)) \cup \Init \cup \gamma(\xi_i)) = \alpha(\postimage{\tr}{\gamma(\xi_i)})$.

\subsection{Monotone Basis and Monotone Span}
\label{sec:monotone-basis}
We define the abstract domain in which $\Lambda$-PDR operates using the notion of a monotone span.
\begin{definition}[Monotone Basis~\cite{DBLP:journals/iandc/Bshouty95}]
\label{def:monotone-basis}
A \emph{monotone basis} is a set of states $B$.
It is a basis \emph{for a formula} $\varphi$ if
$
	\varphi \equiv \mhull{\varphi}{B}.
$
\end{definition}

\begin{definition}[Monotone Span]
\label{def:monotone-span}
$\mspan{B} = \set{\mhull{\varphi}{B} \mid \varphi \mbox{ over } \voc}$,
the set of formulas for which $B$ is a monotone basis.
\end{definition}
\citet{DBLP:journals/iandc/Bshouty95} showed that $\varphi \in \mspan{B}$ iff
there exist clauses $c_1,\ldots,c_s$ such that $\varphi \equiv c_1 \land \ldots \land c_s$ and for every $1 \leq i \leq s$ there exists $b_j \in B$ such that $b_j \not \models c_i$. (This is also a corollary of~\Cref{thm:mhull-conjunctive}.)
The monotone span is thus the set of all formulas that can be written in CNF using clauses that exclude states from the basis.
A consequence of this is that it is closed under conjunction:
\begin{lemma}
\label{lem:mspan-conjunctive-closed}
If $\varphi_1, \varphi_2 \in \mspan{B}$ then also $\varphi_1 \land \varphi_2 \in \mspan{B}$.
\end{lemma}
\begin{proof}
By~\Cref{thm:mhull-conjunctive}, $\varphi_1 \equiv \bigwedge{\set{c \ | \ \mbox{$c$ is a clause, } \varphi_1 \implies c \mbox{ and } \exists b \in B. \, b \not\models c}}$, likewise for $\varphi_2$. The set $\set{c \ | \ \mbox{$c$ is a clause, } \varphi_1 \land \varphi_2 \implies c \mbox{, and } \exists b \in B. \, b \not\models c}$ includes the conjuncts associated with both $\varphi_1,\varphi_2$ (and more), and so $\mhull{\varphi_1 \land \varphi_2}{B} \implies \varphi_1 \land \varphi_2$. The other implication is by~\Cref{lem:mhull-overapproximation}.
\end{proof}

\subsection{Abstract Interpretation in the Monotone Span}
For a set of states $B$, we define the abstract domain $\madom{B}=\langle \mspan{B}, \implies, \join_{B}, \false \rangle$, a logical abstract domain~\cite{DBLP:conf/popl/GulwaniMT08} consisting of the set of (propositional) formulas for which $B$ is a monotone basis (\Cref{def:monotone-basis}), ordered by logical implication, with bottom element $\false$. The existence of $\join_B$ relies on~\Cref{lem:mspan-conjunctive-closed} (the least upper-bound is the conjunction of all upper-bounds), and $\false \in \mspan{B}$ because $\mhull{\false}{B}=\false$, seeing that $\monox{\false}{b}=\false$ for every $b$.

To define a Galois connection~\cite{DBLP:conf/popl/CousotC77} between sets of concrete states and formulas in $\mspan{B}$, we use the \emph{concretization} $\gamma(\varphi) = \set{\sigma \, | \, \sigma \models \varphi}$ (in the sequel, we refer to $\gamma$ as the identity function, by our convention of equating formulas with the set of states they represent). %
The best abstraction is expressed by $\malpha{B}(S)=\mhull{S}{B}$:
\begin{lemma}
\label{lem:best-abstraction}
Let $S \subseteq \States$. Then $\mhull{S}{B}$ is the least overapproximation of $S$ in $\mspan{B}$, namely, $\mhull{S}{B} \implies \varphi$ for every $\varphi \in \mspan{B}$ s.t.\ $S \implies \varphi$.
\end{lemma}
\begin{proof}
Since $\varphi \in \mspan{B}$, by~\Cref{thm:mhull-conjunctive}, $\varphi \equiv \bigwedge_{i}{c_i}$ where each clause $c_i$ excludes some $a_i \in B$. If $\malpha{k}(S) \notimplies \varphi$, then there is some $c_i$ and a corresponding $a_i \in B$ that it excludes such that $\malpha{B}(S) \notimplies c_i$. This means that $\monox{S}{b} \notimplies c_i$ for every $b \in B$. But then in particular $\monox{S}{a_i} \notimplies c_i$, even though $S \implies c_i$ and $a_i \not\models c_i$, which is a contradiction to~\Cref{lem:monox-disjunction-cubes} for $\monox{S}{a_i}$.
\end{proof}

\begin{lemma}
There is a Galois connection $(2^{\States}, \subseteq) \galois{\malpha{B}}{\gamma} (\mspan{B}, \implies)$.\footnote{
	Quotient over logical equivalence, this is a Galois \emph{insertion}, as $\malpha{B}(\gamma(\psi))=\mhull{\psi}{B}\equiv\psi$ for $\psi \in \mspan{B}$.
}
\end{lemma}
\begin{proof}
$\malpha{B}$ is monotone by~\Cref{lem:mhull-monotonicity}.
$\gamma$ is also monotone.
Let $\varphi \in \mspan{B}$ and $S \subseteq \States$.
If $\malpha{B}(S) \implies \varphi$ then $S \subseteq \gamma(\varphi)$ since $S \subseteq \mhull{S}{B}$ (\Cref{lem:mhull-overapproximation}).
If $S \subseteq \gamma(\varphi)$ then $\malpha{B}(S) \implies \varphi$ by~\Cref{lem:best-abstraction}.
\end{proof}

\begin{remark}[Disjunctive completion]
When $B = \bkcube_1 \lor \ldots \lor \bkcube_m$
is a disjunction of multiple cubes, the domain is not disjunctively-complete~\cite{POPL:CC79}: if $\varphi_1,\varphi_2 \in \mspan{B}$, it could be that $\varphi_1 \lor \varphi_2 \not\in \mspan{B}$.
However, for a single cube $\bkcube$, the join operation of $\madom{\bkcube}$ is disjunction, as follows from~\Cref{lem:bshouty-mon-mindnf}.
In this case, a definition of the abstraction $\malpha{\bkcube}$ through the representation function $\beta_{\bkcube}(\sigma)=\moncube{\sigma}{\bkcube}$ is straightforward, since $\malpha{\bkcube}(S) = {\bigjoin_{\sigma \in S}}{\beta_{\bkcube}(\sigma)}$ reads as~\Cref{lem:monox-disjunction-cubes}.
\end{remark}

\begin{remark}[As a reduced product domain]
One way to understand $\madom{B}$ when $B = \bkcube_1 \lor \ldots \lor \bkcube_m$ is as a reduced product~\cite{POPL:CC79} of the per-cube domains: $\madom{B}$ (quotient on logical equivalence) is isomorphic to $\bigotimes_{i}{\madom{\bkcube_i}}$.
The Cartesian product domain is over $m$-tuples of formulas, $\bigtimes_{i}{\mspan{\bkcube_i}}$, ordered by $(\varphi^1_1,\ldots,\varphi^1_m) \sleq (\varphi^2_1,\ldots,\varphi^2_m) \iff \bigwedge_{i=1}^{m}{(\varphi^1_i \implies \varphi^2_i)}$, with concretization $\gamma_{\times \bkcube_i}(\varphi_1,\ldots,\varphi_m)=\bigcap_{i=1}^{m}{\gamma(\varphi_i)}$.
The reduced product quotients the Cartesian product w.r.t.\ having the same concretization.
This is isomorphic to $\madom{B}$ because $\gamma_{\times \bkcube_i} = \gamma(\bigwedge_{i=1}^{m}{\varphi_i})$, and if $\varphi_i \in \mspan{\bkcube_i}$ then $\bigwedge_{i=1}^{m}{\varphi_i} \in \mspan{B}$.
\end{remark}

\begin{wrapfigure}{r}{0.33\textwidth}
\vspace{-0.8cm}
\begin{minipage}{0.33\textwidth}
\begin{algorithm}[H]
\caption{Kleene Iterations in $\madom{\bkwrch{k}}$}
\label{alg:eepdr-ai}
\begin{algorithmic}[1]
\begin{footnotesize}
\Procedure{$\Lambda$-PDR}{$\Init$, $\tr$, $\Bad$, $k$}
	\State $i \gets 0$
	\State $\Frameai_{-1} \gets \false$
	\State $\Frameai_0 \gets \malpha{\bkwrch{k}}(\Init)$ $\label{ln:eepdr-ai:frame0}$
	\While{$\Frameai_{i} \notimplies \Frameai_{i-1}$}
		\State $\Frameai_{i+1} = \malpha{\bkwrch{k}}(\postimage{\tr}{\Frameai_i})$ $\label{ln:eepdr-ai:transformer}$
		\State $i \gets i+1$
	\EndWhile
	\State \Return $\Frameai_i$
\EndProcedure
\end{footnotesize}
\end{algorithmic}
\end{algorithm}
\end{minipage}
\vspace{-0.6cm}
\end{wrapfigure} %
\para{$\Lambda$-PDR as Kleene iterations}
\Cref{alg:eepdr-ai} shows Kleene iterations in $\madom{\bkwrch{k}}$ with the best abstract transformer for $(\Init,\tr)$.
The next iterate (\cref{ln:eepdr-ai:transformer}) is always $\Frameai_{i+1}=\malpha{\bkwrch{k}}(\postimage{\tr}{\Frameai_i})$, which exactly matches the relation between successive frames in $\Lambda$-PDR (\Cref{lem:justify-overview-next-frame}).
This means that $\Lambda$-PDR's frames exactly match the Kleene iterates, at least when the initial states are in $\bkwspan{k}$ themselves (in which case the first iterate in~\cref{ln:eepdr-ai:frame0} is simply $\Init$); otherwise there is a difference of at most one frame:
\begin{theorem}
\label{lem:ai-eepdr-sandwich}
$\Frameai_i \subseteq \Frame_{i+1} \subseteq \Frameai_{i+1}$ for every $i$ where $\Frame_{i+1}$ exists.\yotamforlater{assuming that Kleene iterations are always defined, it's not this way in the \Cref{alg:eepdr-ai} but it is in the basic exposition. we don't care?}
Further, if $\Init \in \bkwspan{k}$, then $\Frame_{i}=\Frameai_i$.
\end{theorem}
\begin{proof}
By induction on $i$, first prove that $\Frame_i \subseteq \Frameai_i$ for every $i$ where $\Frame_i$ exists.
Initially, $\Frame_0 = \Init \subseteq \malpha{\bkwrch{k}}(\Init) =\Frameai_0$.
For the step, assume that $\Frame_i \subseteq \Frameai_i$.  %
Using~\Cref{lem:mhull-monotonicity}, this implies that $\mhull{\postimage{\tr}{\Frame_i}}{\bkwrch{k}} \subseteq \mhull{\postimage{\tr}{\Frameai_i}}{\bkwrch{k}}$, which by~\Cref{lem:justify-overview-next-frame} and the Kleene iterations means that $\Frame_{i+1} \subseteq \Frameai_{i+1}$. %
For the other inclusion, likewise, since $\Frameai_0 = \malpha{\bkwrch{k}}(\Init) = \mhull{\Init}{\bkwrch{k}} \subseteq \mhull{\postimage{\tr}{\Init}}{\bkwrch{k}} = \Frame_1$ \iflong(where the inclusion uses~\Cref{lem:mhull-monotonicity})\fi, we have by induction that $\Frameai_i \subseteq \Frame_{i+1}$ for every $i$ s.t.\ $\Frame_{i+1}$ exists.
Similarly, if $\Init \in \bkwspan{k}$ then $\malpha{\bkwrch{k}}(\Init) = \Init$, and by induction $\Frame_{i}=\Frameai_i$ for every $i$.
\end{proof}
We can now relate the number of frames in $\Lambda$-PDR and the number of Kleene iterations:
\begin{corollary}
\label{lem:eepdr-ai-iterations}
$\eepdr(\Init,\tr,\Bad,k)$ (\Cref{alg:eepdr}) converges or fails (\cref{ln:eepdr-restart}) in a frame whose index is at most one greater than the number of Kleene iterations in $\madom{\bkwrch{k}}$ on $(\Init,\tr)$ (\Cref{alg:eepdr-ai}).
\end{corollary}
\begin{proof}
Let $i^*$ be the iteration where \Cref{alg:eepdr-ai} converges to the least-fixed point, i.e. $\Frameai_{i^*+1}=\Frameai_{i^*}$.
If \Cref{alg:eepdr} does not terminate with error (\cref{ln:eepdr-restart}) before $i = i^*+1$, by~\Cref{lem:ai-eepdr-sandwich} $\Frameai_i \subseteq \Frame_{i+1} \subseteq \Frameai_{i+1}$ for every $i \leq i^*$, and for $i=i^*$ we have $\Frameai_{i^*} \subseteq \Frame_{i+1} \subseteq \Frameai_{i^*+1}=\Frameai_{i^*}$, so $\Frame_{i^*+1}$ is an inductive invariant.
Therefore, $\postimage{\tr}{\Frame_{i^*+1}} = \Frame_{i^*+1}$, and thus, because $\Frame_{i^*+1} \in \bkwspan{k}$, also $\mhull{\postimage{\tr}{\Frame_{i^*+1}}}{\bkwrch{k}} = \Frame_{i^*+1}$ and the algorithm converges.
\end{proof}
Further, $\Lambda$-PDR converges whenever the Kleene iterations converge to an inductive invariant:
\begin{corollary}
\label{lem:eepdr-lfp}
If there exists an inductive invariant $I \in \bkwspan{k}$ for $(\Init,\tr,\Bad)$, then $\eepdr(\Init,\tr,\Bad,k)$ (\Cref{alg:eepdr}) converges to an inductive invariant.
\end{corollary}
\begin{proof}
\Cref{alg:eepdr-ai} converges from below to the abstract lfp for $(\Init,\tr)$ in $\madom{\bkwrch{k}}$, which from the premise is strong enough to prove safety. Using~\Cref{lem:ai-eepdr-sandwich}, %
the same is true for \Cref{alg:eepdr}.
\end{proof}

\para{Increasing $k$ refines $\madom{\bkwrch{k}}$}
Increasing the backward exploration bound $k$ refines $\madom{\bkwrch{k}}$ by increasing $\bkwspan{k}$ ($\bkwrch{k} \subsetneq \bkwrch{k'}$ implies $\bkwspan{k} \subsetneq \bkwspan{k'}$: every $\varphi \in \bkwspan{k}$ is also $\varphi \in \bkwspan{k'}$, and for instance the clause $\neg \sigma_b \in \bkwspan{k'} \setminus \bkwspan{k}$ %
if the state
$\sigma_b \in \bkwrch{k'} \setminus \bkwrch{k}$). %
Restarting $\Lambda$-PDR with a larger $k$ (\cref{ln:eepdr-increase-k} in~\Cref{alg:eepdr})
thus refines the domain until it includes an inductive invariant that establishes safety. Such a $k$ always exists because the set of all backward reachable states (attained by some finite $k$ in the setting of propositional systems) is sufficient to express the weakest safe inductive invariant.

\para{Efficient convergence}
Unfortunately, the lattice height of $\mspan{B}$ is exponential. For example, all the formulas in the strictly ascending chain of formulas $\set{\vec{x} \leq i}_{0 \leq i \leq 2^n-1}$ over $n$ variables $\vec{x}=x_{n-1},\ldots,x_0$ are in $\mspan{\set{0\ldots0}}$ (see~\Cref{sec:slow-convergence}).
Therefore, to bound the number of iterations we need to consider properties of the transition system, a task on which we embark next.  
\section{Convergence Bounds via Abstract Diameter}
\label{sec:abstract-diameter-all}
In this section and the next, we prove a bound on the number of iterations of $\Lambda$-PDR on a given transition system via the DNF size of a monotonization of the transition relation.
In the current section we assume that $\bkwrch{k}$ can be expressed as a single cube $\bkcube$, and generalize it in~\Cref{sec:hyper-all}.

To formulate the bound, we define a monotonization of the transition relation, which is a formula over $\voc \cup \voc'$. For cubes $c_1$ and $c_2$ over $\voc$, we denote by $\monox{\tr}{(c_1,c_2)}$ the monotonization $\monox{\tr}{a}$ where $a = c_1 \land c_2'$ and $c_2'$ is obtained from $c_2$ by substituting each $p \in \voc$ by the corresponding $p' \in \voc'$.

The monotonization we perform on the pre-state vocabulary $\voc$ uses the \emph{reflection} of the backward reachable cube:
\begin{definition}[Reflection]
For a cube $\bkcube = \ell_1 \land \ldots \land \ell_r$, the \emph{reflection} is $\reflect{\bkcube} = \neg \ell_1 \land \ldots \land \neg \ell_r$.
\end{definition}

Our main theorem in this section is as follows.
\begin{theorem}
\label{thm:abstract-diamter-bound}
Let $(\Init,\tr,\Bad)$ be a transition system, and $\bkwrch{k}=\bkcube$ a cube.
Then $\eepdr(\Init,\tr,\Bad,k)$ converges or fails in a frame whose index  is bounded by $\dnfsize{\monox{\tr}{(\reflect{\bkcube},\bkcube)}}+1$.
\end{theorem}

\begin{example}
\label{ex:even-counter-constant}
\yotamsmall{new}
For an example where~\Cref{thm:abstract-diamter-bound} yields a tight bound, consider a counter over $\vec{x}=x_n,\ldots,x_0$ where $\Init$ is $\vec{x}=00\ldots00$, $\tr$ increments even numbers by two, $\Bad$ is $\vec{x}=10\ldots01$, and $k=0$.
The monotonization $\absr{\tr} = \monox{\tr}{\vec{x}=01\ldots10,\vec{x}'=10\ldots01}$ is $x'_0=0$ (see the calculation below), so $\dnfsize{\monox{\tr}{\vec{x}=01\ldots10,\vec{x}'=10\ldots01}}=1$. By~\Cref{thm:abstract-diamter-bound}, $\Lambda$-PDR converges in $\Frame_2$, and indeed in this example $\Frame_1 = (x_n=0 \land x_0=0)$, and $\Frame_2 = (x_0=0)$ is the inductive invariant (all the even numbers).

To see that indeed $\absr{\tr} = (x'_0=0)$, note that $(01\ldots10,10\ldots0) \in \tr$, the monotonization $\moncube{\vec{x}=01\ldots10,\vec{x}'=10\ldots01}{\vec{x}=01\ldots10,\vec{x}'=10\ldots00)} = (x'_0=0)$ and use~\Cref{lem:monox-disjunction-cubes}; for the other direction, $x'_0=0$ holds in every transition of $\tr$ and remains in every such monotone cube, invoking again~\Cref{lem:monox-disjunction-cubes}.
\end{example}

\begin{example}
\yotamsmall{new}
\label{ex:skip-counter-bounce-polynomial}
An example where~\Cref{thm:abstract-diamter-bound} yields a polynomial yet non-tight bound appears in~\Cref{sec:overview-diameter-bound}.
\end{example}

\para{Outline}
We prove~\Cref{thm:abstract-diamter-bound} by constructing an ``abstract'' transition system whose diameter is the number of iterations required for \Cref{alg:eepdr-ai} to converge (\Cref{sec:abstract-transition}), and bound its diameter %
(\Cref{sec:diameter-bound}).
Throughout, we fix a transition system $(\Init,\tr,\Bad)$ and a backwards bound $k \in \mathbb{N}$, denoting the cube $\bkwrch{k}$ by $\bkcube$.

\subsection{Abstract Transition System}
\label{sec:abstract-transition}
Given $(\Init,\tr,\Bad)$, the \emph{abstract transition system} $(\abs{\Init},\abs{\tr},\abs{\Bad})$ is defined over the same set of states as the original, and extends its transitions:
\begin{definition}[Abstract Transition System]
\label{def:abs-tr} %
The abstract transition system of $(\Init,\tr,\Bad)$ w.r.t.\ $\bkcube$
is defined as a transition system $(\absr{\Init},\absr{\tr},\absr{\Bad})$ over $\States$ with $\absr{\Init} = \monox{\Init}{\bkcube}$, $\absr{Bad} = \Bad$, and $\absr{\tr}=\monox{\tr}{(\reflect{\bkcube},\bkcube)}$.
\end{definition}
The monotonization in $\absr{\tr}$ of the pre-image is understood using the following technical lemma about monotonization w.r.t.\ a reflection.
\begin{lemma}
\label{lem:moncube-reflect}
$\sigma_1 \models \moncube{\sigma_2}{\reflect{\bkcube}} \iff \sigma_2 \models \moncube{\sigma_1}{\bkcube}$ for every cube $b$ and states $\sigma_1,\sigma_2$.
\end{lemma}
\toolong{
\begin{proof}
Suppose that $\sigma_1 \models \moncube{\sigma_2}{\reflect{\bkcube}}$.
Then for every $p$ present in $\reflect{\bkcube}$ (equivalently, in $\bkcube$), $\sigma_1,\reflect{\bkcube}$ \emph{dis}agree on $p$ whenever $\sigma_2,\reflect{\bkcube}$ do. Contrapositively, $\sigma_2,\reflect{\bkcube}$ \emph{agree} on $p$ whenever $\sigma_1,\reflect{\bkcube}$ do. Equivalently, $\sigma_2,\bkcube$ \emph{dis}agree on $p$ whenever $\sigma_1,\bkcube$ do. This shows that $\sigma_2 \models \moncube{\sigma_1}{\bkcube}$. The other direction of the implication is symmetric because $\reflect{\reflect{\bkcube}}=\bkcube$.
\end{proof}
}

The central property of the abstract transition system is that its $i$-reachable state capture executions of iterations in the $\madom{\bkwrch{k}}$ abstract domain.
\begin{lemma}
\label{thm:absract-reach}
Let $\absr{R}_i$ be the set of states reachable in $(\absr{\Init},\absr{\tr},\absr{Bad})$ w.r.t.\ $\bkwrch{k}=\bkcube$ (\Cref{def:abs-tr}) in at most $i$ steps.
Then $\absr{R}_i \equiv \Frameai_{i}$ where $\Frameai_{i}$ is the $i$'th iterate of~\Cref{alg:eepdr-ai} in $\madom{\bkcube}$ on $(\Init,\tr,\Bad)$.
\end{lemma}
This will imply that the diameter of the abstract transition system---the minimal $i$ where $i$-reachability converges to all the reachable states---equals the number of iterations needed for convergence of the abstract interpretation in $\madom{\bkwrch{k}}$.

Before proving this connection, let us explain the intuition for the abstract transition system and the relation to the algorithm.

Through~\Cref{lem:monox-disjunction-cubes,lem:moncube-reflect} we can see that a transition $(\sigma,\sigma')$ of $\absr{\tr}$ consists of three steps:
\begin{itemize}
	\item monotonization of $\sigma$ to $\widetilde{\sigma}$ %
($\sigma$ is the ``protector'' state for $\widetilde{\sigma}$); then
	\item a concrete transition $(\widetilde{\sigma},\widetilde{\sigma}') \in \tr$; and
	\item the monotonization of $\widetilde{\sigma}'$ to $\sigma'$ %
($\widetilde{\sigma}'$ is the ``protector'' state for $\sigma'$).
\end{itemize}

The monotonization \emph{after} the concrete transition mimics a step of the algorithm, which computes the post-image and then a monotone overapproximation.
A potentially critical difference between the algorithm and $\absr{\tr}$ is that the transition system performs this state-by-state, while the algorithm computes these operations over sets. The insight is that when $\bkwrch{k}$ is a single cube, the abstraction of $\Lambda$-PDR factors to individual states as well, and so can be captured using an ordinary transition system, so that its $i$-reachable states correspond to iterations of the algorithm, which is the objective of the next lemma.

As the proof shows, the monotonization \emph{before} the concrete transition does not change reachability, because in execution traces of $\absr{\tr}$, this monotonization is absorbed by the monotonization at the end of the previous transition (the first step in the trace is handled by taking the monotonization of the initial states, $\abs{\Init}=\absr{\Init}=\monox{\Init}{b}$).
Even though the monotonization of the pre-state does not change the diameter, it can improve (and never worsen) our diameter \emph{bound}, which is derived in~\Cref{sec:diameter-bound}.

\begin{proof}[Proof of~\Cref{thm:absract-reach}]
We first prove a similar result for a slightly simpler, ``less abstract'' transition system, where monotonization is performed in the post-state but not in the pre-state.
Define a transition system $(\abs{\Init},\abs{\tr},\abs{\Bad})$ over $\States$ by $\abs{\Init} = \monox{\Init}{\bkcube}$, $\abs{Bad} = \Bad$, and
\begin{equation*}
	(\sigma,\sigma') \in \abs{\tr} \iff \exists \sigma''. \ (\sigma,\sigma'') \in \tr \land \sigma' \in \cubemon{\sigma''}{\bkcube}.
\end{equation*}
(In fact, $\abs{\tr}=\monox{\tr}{(\true,\bkcube)}$.)

Denote by $\abs{R}_i$ the set of states reachable in $(\abs{\Init},\abs{\tr},\abs{Bad})$ in at most $i$ steps.
We argue that $\abs{R}_i \equiv \Frameai_{i}$,
by induction on $i$.
Initially, $\abs{R}_0 = \abs{\Init} = \monox{\Init}{\bkcube} = \Frameai_0$.
For the step, by the definition of the abstract system, by the induction hypothesis $\abs{R}_i \equiv \Frameai_i \in \mspan{\bkcube}$. Hence,
\begin{align*}
	\abs{R}_{i+1} &= \postimage{\abs{\tr}}{\abs{R}_i}
	= \abs{R}_i \cup \bigvee_{\sigma'' \in \tr({\abs{R}_i})}{\cubemon{\sigma''}{\bkcube}}
	\underset{\rm \Cref{lem:monox-disjunction-cubes}}{=}
	\abs{R}_i \cup \monox{{\tr}({\abs{R}_i})}{\bkcube}
	\\
	&\underset{\abs{R}_i \in \mspan{\bkcube}}{=}
	\monox{\abs{R}_i}{\bkcube} \cup \monox{{\tr}({\abs{R}_i})}{\bkcube}
	\underset{\rm \Cref{lem:bshouty-mon-mindnf}}{=}
	\monox{\postimage{\tr}{\abs{R}_i}}{\bkcube}
	\underset{\rm ind.}{=}
	\monox{\postimage{\tr}{\Frameai_i}}{\bkcube}
	= \Frameai_{i+1}.
\end{align*} %

It remains to show that $\absr{R}_i = \abs{R}_i$, i.e., that the $i$-reachable states of $\abs{\tr},\absr{\tr}$ coincide (although they are not in general bisimilar).

First, $\abs{\tr} \subseteq \absr{\tr}$. This is because if $(\sigma,\sigma') \in \abs{\tr}$, then by definition there is $\sigma''$ such that $(\sigma,\sigma'') \in \tr$ and $\sigma' \in \moncube{\sigma''}{\bkcube}$. Considering the product monotone order, $(\sigma,\sigma'') \leq_{(\cdot,\bkcube)} (\sigma,\sigma')$,
and so $(\sigma,\sigma'') \in \tr \implies (\sigma,\sigma') \in \monox{\tr}{(\reflect{\bkcube},\bkcube)} = \absr{\tr}$, as required.

Second, we show that for any $S \in \mspan{\bkcube}$ it holds that $\postimage{\absr{\tr}}{S} \subseteq \postimage{\abs{\tr}}{S}$.
$\postimage{\absr{\tr}}{S} = S \cup \absr{\tr}(S)$ and $\postimage{\abs{\tr}}{S} = S \cup \abs{\tr}(S)$, so we need to show that $\absr{\tr}(S) \subseteq \abs{\tr}(S)$.
Let $(\sigma,\sigma') \in \absr{\tr}, \sigma \in S$. By the definition of $\absr{\tr}$ and~\Cref{lem:monox-disjunction-cubes}, there exists $(\widetilde{\sigma},\widetilde{\sigma}') \in \tr$ such that $\sigma \models \moncube{\widetilde{\sigma}}{\reflect{\bkcube}}$ and $\sigma' \models \moncube{\widetilde{\sigma}'}{\bkcube}$.
The former implies, by~\Cref{lem:moncube-reflect}, that $\widetilde{\sigma} \models \moncube{\sigma}{\bkcube}$, hence $\widetilde{\sigma} \in S$ as well (because $S \in \mspan{\bkcube}$). Writing $\sigma'' = \widetilde{\sigma}'$ shows that $(\widetilde{\sigma},\sigma') \in \abs{\tr}$, and hence $\sigma' \in \postimage{\abs{\tr}}{S}$, as required.

The first part of the argument (and induction on $i$) shows that $\abs{R}_i \subseteq \absr{R}_i$.
We have shown that $\abs{R}_i = \Frameai_i$, which in particular implies that
always $\abs{R}_i \in \mspan{\bkcube}$; therefore, the second argument above shows that $\absr{R}_i \subseteq \abs{R}_i$.
The claim follows.
\end{proof}
\begin{corollary}
\label{thm:abstract-diameter-eepdr}
Let $(\absr{\Init},\absr{\tr},\absr{\Bad})$ be the abstract transition system w.r.t.\ $\bkwrch{k}=\bkcube$ (\Cref{def:abs-tr}).
If $(\absr{\Init},\absr{\tr},\absr{\Bad})$ is safe and its reachability diameter is $s$, then
$\eepdr(\Init,\tr,\Bad,k)$ converges in frame at most $s+1$.
If $(\absr{\Init},\absr{\tr},\absr{\Bad})$ reaches a bad state in $s$ steps, then $\eepdr(\Init,\tr,\Bad,k)$ fails (\cref{ln:eepdr-restart}) in frame at most $s+1$.
\end{corollary}
\begin{proof}
From~\Cref{thm:absract-reach}, $\Frameai_{s} \equiv \Frameai_{s+1}$ iff $\absr{R}_{s} \equiv \absr{R}_{s+1}$, %
and the least $s$ in which the latter holds is the diameter. For the unsafe case, $\Frameai_s \cap \Bad \neq \emptyset$ iff $\absr{R}_s \cap \Bad \neq \emptyset$. Apply~\Cref{lem:eepdr-ai-iterations} in both cases to deduce convergence, resp. failure, of $\Lambda$-PDR in frame at most $s+1$.
\end{proof}

\subsection{Diameter Bound via Abstract DNF Size}
\label{sec:diameter-bound}
In this section, we bound the diameter of the abstract transition system %
in order to obtain the convergence bound of~\Cref{thm:abstract-diamter-bound}.
We use a simple, general bound on the diameter of transition systems\iflong, by the DNF size of the transition relation\fi:
\begin{lemma}
\label{lem:diam-dnf}
The reachability diameter of a transition system $(\Init,\tr,\Bad)$ is bounded by $\dnfsize{\tr}$.
\end{lemma}
\begin{proof}
Fix a minimal DNF representation of $\tr$. Thinking about each disjunct of $\tr$ as an action $a$, every transition can be labeled by at least one action. Whenever in an execution $\sigma_1,\sigma_2,\ldots$ an action $a$ labels two transitions $\sigma_{i_1} \overset{a}{\rightarrow} \sigma_{i_1+1},\sigma_{i_2} \overset{a}{\rightarrow} \sigma_{i_2+1}$, the segment between the occurrences, $\sigma_{i_1+1},\ldots,\sigma_{i_2}$ can be dropped and the resulting trace is still valid (and terminates at the same state)---this is because if $(\sigma_{i_1},\sigma_{i_1+1}) \models a$ and likewise $(\sigma_{i_2},\sigma_{i_2+1}) \models a$ then also $(\sigma_{i_1},\sigma_{i_2+1}) \models a$, because $a$, which is a cube, can be decomposed to $a_\textit{pre} \land a_\textit{post}$ where all the literals in $a_\textit{pre}$ are in $\voc$ and those in $a_\textit{post}$ are in $\voc'$.
Overall, every state that can be reached from another state can do so by an execution where each action appears at most once, and thus the diameter is bounded by $\dnfsize{\tr}$.
\end{proof}

Combining the above results yields a proof of this section's main theorem:
\begin{proof}[Proof of~\Cref{thm:abstract-diamter-bound}]
By~\Cref{thm:abstract-diameter-eepdr}, the number of iterations before convergence or failure of $\Lambda$-PDR is bounded by %
1 plus
the reachability diameter of $(\absr{\Init},\absr{\tr},\absr{\Bad})$,
which by~\Cref{lem:diam-dnf} is at most $\dnfsize{\monox{\tr}{(\reflect{\bkcube},\bkcube)}}$.
\end{proof}

\para{Complexity}
Finding whether there is an equivalent DNF representation with at most $s$ terms is complete for the second level of the polynomial hierarchy $\psigma{2}$~\cite{DBLP:journals/jcss/Umans01}.
This is on par with deciding whether the diameter is bounded by $s$~\cite{DBLP:journals/tcs/HemaspaandraHTW10} (see also~\cite{schaefer2002completeness}).
Thus the bound in~\Cref{thm:abstract-diamter-bound} is not an efficiently-computable upper bound on the number of frames of $\Lambda$-PDR. Instead, we view the result of~\Cref{thm:abstract-diamter-bound} as a conceptual explanation of how smaller diameters can originate from the abstraction.

\begin{example}
\yotamsmall{new}
\label{ex:skip-counter-pure-exponential}
In the running example from~\Cref{sec:overview} (without the additional transitions in~\Cref{sec:overview-diameter-bound}), \Cref{thm:abstract-diamter-bound} yields a trivial, exponential, bound which is not tight.
Consider the system in~\Cref{fig:skip-counter} restricted to the $\vec{x}$ part (see~\Cref{lem:diameter-bound-project-irrelevant} and the justification in~\Cref{sec:overview-diameter-bound}).
Then $\absr{\tr} = \monox{\tr}{\vec{x}=01\ldots11,\vec{x}'=10\ldots00} = (x_n=1 \land x'_n=0) \lor (x_n=0 \land x'_0=1) \lor (x_n=0 \land x'_n=0 \land \vec{x}' > \vec{x}) \lor (x_n=1 \land x'_n=1 \land \vec{x}' > \vec{x})$, which has an exponential DNF size,\footnote{
	\yotam{proof adapted from an email} $\vec{x}' > \vec{x}$ is unate~\cite{DBLP:journals/jacm/AngluinHK93}---it is closed under turning variables in $\vec{x}$ from $1$ to $0$, and turning variables in $\vec{x}'$ from $0$ to $1$. Hence its unique and minimal DNF representation consists of the disjunction of all its prime implicants~\cite{quine1954two} (a term $t$ is a prime implicant of $\varphi$ if $t \implies \varphi$ but this ceases to hold when dropping any literal from $t$). It suffices to show that there is an exponential number of prime implicants.
	To see this, let $v$ be an assignment to all the variables except the least significant, let $\sigma=(v,0),\sigma'=(v',1)$ (so in $\sigma$, $\vec{x} = 2v$, and in $\sigma'$, $\vec{x} = 2v+1$). Then $(\sigma,\sigma') \models \vec{x}' > \vec{x}$.
	The only prime implicant that can be obtained by dropping literals from $(\sigma,\sigma')$ is the conjunction that includes every $x_i=0$ when $v[x_i]=0$ and $\vec{x}'=1$ when $v[x_i]=1$---dropping any one of these variables would result in a term that is satisfied by $(v[x_i \mapsto 1],0),(v'[x_i \mapsto 0],1)$, which does not satisfy $\vec{x}' > \vec{x}$, and so the resulting term would not be an implicant.
	Thus, the prime implicant associated with $v$ can be used to reconstruct $v$, setting $v[x_i]=0$ when $x_i=0$ is present in the prime implicant, and $v[x_i]=1$ when $x'_i=1$ is present. There are exponentially many choices of $v$ and we have shown that each induces a different prime implicant, and the claim follows.
}
yielding an exponential bound on the number of frames of $\Lambda$-PDR. However, as~\Cref{ex:running-all-frames} shows, $\Lambda$-PDR converges in this case in a constant number of frames.

To see that indeed $\absr{\tr}$ is as stated above, it is easiest to think about the behavior of $\absr{\tr}$ as an abstract transition system (\Cref{def:abs-tr}). \yotam{is this actually right?}
\begin{itemize}
	\item Starting in a state $\sigma$ with $x_n=0$, $\absr{\tr}$ can lead us to any state $\sigma'$ with $x'_n=0$ and $\vec{x}' > \vec{x}$---to reach $\vec{x}'$ the abstraction step turns all the bits below its leading $1$, a concrete step increments the counter, so the resulting state agrees with the leading $1$ and everything else is $0$, and then another abstraction step can generate the other $1$'s present in $\vec{x'}$.

	Additionally, we can reach in the abstraction the state $01\ldots11$, a step then skip and arrives at $10\ldots01$, which is abstracted to all numbers with $x_0=1$.

	These are all the states we may reach this way; the first abstract step cannot change $x_n$, and we cannot arrive at smaller numbers, because both the concrete and abstract steps (which turn $0$'s to $1$'s) can only strictly increase the number.

	\item Starting in a state $\sigma$ with $x_n=1$, $\absr{\tr}$ can lead us to any state with $x'_n=1,\vec{x}' > \vec{x}$, similarly to above for the case $x_n=0$. Additionally, we can abstract to $11\ldots11$, from which a concrete step wraparounds to $00\ldots00$, which then abstracts to $x'_n=0$.

	These are all the states that we may reach this way; we cannot reach smaller numbers with $x'_n=1$, along a similar argument to the case above of $x_n=0$.
\end{itemize}
\end{example}

\section{Convergence Bounds via Abstract Hypertransition Systems}
\label{sec:hyper-all}
In this section, we generalize the results of~\Cref{sec:abstract-diameter-all} to the case that $\bkwrch{k}$ is not expressible as a single cube.
In this case, our bound is the product of monotonizations w.r.t.\ the different cubes that comprise $\bkwrch{k} = \bkcube_1 \lor \ldots \lor \bkcube_m$ in the post-state, and w.r.t.\ (the reflection of) the least cube that contains all $\bkwrch{k}$ in the pre-state, defined as follows:
\begin{definition}
If $\bkwrch{k}=\bkcube_1 \lor \ldots \lor \bkcube_m$, we denote by $\cubejoin{\bkwrch{k}} = \bigcap_{i=1}^{m}{\bkcube_i}$ (as sets of literals) the cube that consists of the literals that appear in all $\bkcube_1,\ldots,\bkcube_m$.
\end{definition}

Fix a representation $\bkwrch{k}=\bkcube_1 \lor \ldots \lor \bkcube_m$.
Our main theorem is as follows:
\begin{theorem}
\label{thm:abstract-hyperdiamter-bound}
Let $(\Init,\tr,\Bad)$ be a transition system.
Then $\eepdr(\Init,\tr,\Bad,k)$ converges or fails in a frame whose index is bounded by
\begin{equation*}
	\prod_{i=1}^{m}{\left(
					\dnfsize{\monox{\tr}{(\reflect{\cubejoin{\bkwrch{k}}},\bkcube_i)}}
					+
					\dnfsize{\monox{\Init}{\bkcube_i}}
					\right)
				}
	+
	1.
\end{equation*}
\end{theorem}
The reasons for $\monox{\Init}{\bkcube_i}$ and $\cubejoin{\bkwrch{k}}$ will become clear in~\Cref{sec:abstract-hypertransition}.
Often $\Init$ is a cube, in which case $\dnfsize{\monox{\Init}{\bkcube_i}}=1$.

\begin{example}
\label{ex:multiskip-counter-poly}
For an example where~\Cref{thm:abstract-hyperdiamter-bound} yields a polynomial convergence bound,
consider a counter over $\vec{x}=x_n,\ldots,x_0$ (similar in spirit to~\Cref{fig:skip-counter}) with $\Init \ = \ \vec{x}=0$, $\Bad \ = \vec{x} = 10\ldots0$, and $\tr$ that
\begin{inparaenum}[(i)]
	\item skips every multiple of $2^r$ except $\vec{0}=0\ldots0$;
	\item from every state with $x_r = x_{r-1} = \ldots = x_1 = x_0 = 1$ may also ``bounce back'' to a state with the same upper bits ($i > r$) and exactly one lower bit ($i \leq r$) is $1$, or to $\vec{x}=0$ if the upper bits are already $0$; and
	\item transitions from any multiple of $2^r$ but $\vec{0}$ to any other multiple of $2^r$ (including the bad state).
\end{inparaenum}
We call the set of states between consecutive multiples of $\vec{x}=2^r,\vec{x}=2^{r+1}$ a ``segment''.

Assume that $r=n-\textit{polylog}(n)$ (the counter skips relatively few times).
We now compute the bound resulting from the theorem.
For every $k \geq 1$,
$\bkwrch{k}=\bigvee_{i=s}^{n}{\bkcube_i}$ where $\bkcube_i = (x_{r-1}=0 \land \ldots \land x_0=0 \land x_i=1)$ (all multiples of $2^r$ except $0\ldots0$).
$\cubejoin{\bkwrch{k}} = (x_{r-1}=0 \land \ldots \land x_0=0)$ (all multiples of $2^r$).
We find a DNF representation for $\monox{\tr}{(\reflect{\cubejoin{\bkwrch{k}}},\bkcube_i)}$ using~\Cref{lem:bshouty-mon-mindnf} similarly to~\Cref{sec:overview-diameter-bound}:
The number of segments is $2^{n-r}$, and in each segment
the abstraction of the ``bounce back'' transitions subsume the transitions between numbers in the same segments. \yotam{validate me}
This amounts to $\bigO(r2^{n-r})$ terms. \yotam{validate me}
The number of disjuncts $\bkcube_i$ is $n-r$, the number of terms in $\monox{\Init}{\bkcube_i}$ is $1$ because $\Init$ is itself a term, and overall \Cref{thm:abstract-hyperdiamter-bound} yields the bound $\bigO(r (n-r) 2^{n-r})=\textit{poly}(n)$.
\end{example}

\begin{example}
For an example where the theorem yields an exponential convergence bound, consider the same system as in the previous example (\Cref{ex:multiskip-counter-poly}) but when $r=\textit{polylog}(n)$. The above calculation still yields the same bound but now it is $\Omega(2^n)$.
This exponential bound reflects true exponential behavior of the algorithm, because each post-image crosses to at most one new segment, and the abstraction never produces states in a segment beyond those represented in the current frame, mandating at least $2^{n-r}$ frames\iflong.\footnote{
	To see that $\mhull{\Frame_i}{\bkwrch{k}}$ never introduces states in a segment that was not already present in $\Frame_i$, note that because always $0\ldots0 \in \postimage{\tr}{\Frame_i}$ (the initial state), $\monox{\Frame_i}{0\ldots0}=\true$, and thus
	$\mhull{\postimage{\tr}{\Frame_i}}{\bkwrch{k} \cup \set{0\ldots0}}=\mhull{\postimage{\tr}{\Frame_i}}{\bkwrch{k}} \land \monox{\postimage{\tr}{\Frame_i}}{0\ldots0} = \mhull{\postimage{\tr}{\Frame_i}}{\bkwrch{k}}$. Hence, $\mhull{\postimage{\tr}{\Frame_i}}{\bkwrch{k}} = \mhull{\postimage{\tr}{\Frame_i}}{\bkwrch{k} \cup \set{0\ldots0}} = \monox{\postimage{\tr}{\Frame_i}}{x_{r-1}=0 \land \ldots \land x_0=0}$. But the cube $x_{r-1}=0 \land \ldots \land x_0=0$ does not mention the upper bits, and thus the monotonization does not alter these bits, and includes only segments that were already present in $\postimage{\tr}{\Frame_i}$.
}
\else
~(see the extended version~\cite{extendedVersion} for details).
\fi
\end{example}

\begin{remark}
\label{lem:diameter-bound-project-irrelevant}
There are cases where it is possible to apply \Cref{thm:abstract-hyperdiamter-bound} on a restriction of $\tr$ to specific values to some of the variables, and this produces a better bound.
Suppose that for a set of variables $\vec{x}$ and some $\vec{v}$ valuation thereof, $\vec{x}=\vec{v}$ is an inductive invariant for the system, and $\exists b \in \bkwrch{k}. \ b[\vec{x}]=\reflect{\vec{v}}$. (In~\Cref{sec:overview-diameter-bound}, actually $\bkwrch{k} \implies \vec{x}=\reflect{\vec{v}}$.)
Then applying \Cref{thm:abstract-hyperdiamter-bound} to $\restrict{\tr}{\vec{x}\gets\vec{v}}=\tr[\vec{v}/\vec{x}]$, eliminating $\vec{x}$ by substituting $\vec{v}$ for it, also yields an upper bound on the number of iterations of $\Lambda$-PDR. The benefit is that
$\dnfsize{\monox{\restrict{\tr}{\vec{x}\gets\vec{v}}}{(\ldots)}}$
can be smaller than with the original $\tr$.
It is correct to apply the theorem to the restriction and deduce a bound for the original, because under the above premises, always $\Frame_i[\tr] = \Frame_i[\restrict{\tr}{\vec{x}\gets\vec{v}}] \land \vec{x}=\vec{v}$, where $\Frame_i[\tau]$ is the $i$th frame of $\Lambda$-PDR w.r.t.\ transition relation $\tau$.\footnote{
	This is because $\Frame_i \eqdef \Frame_i[\tr] \implies \vec{x}=\vec{v}$ by induction on $i$---since $\vec{x}=\vec{v}$ is an inductive invariant, this holds initially, as well as $\postimage{\tr}{\Frame_i} \implies \vec{x}=\vec{v}$.
	Now $\Frame_{i+1}=\mhull{\postimage{\tr}{\Frame_i}}{\bkwrch{k}} \implies \vec{x}=\vec{v}$ as well, because from the assumption there is $b \in \bkwrch{k}$ s.t.\ $b[\vec{x}]=\reflect{\vec{v}}$, so $\cubemon{\sigma}{b} \implies \vec{x}=\vec{v}$ for every $\sigma \in \postimage{\tr}{\Frame_i}$ because in $\sigma$, $\vec{x}=\vec{v}$, which are opposite in $b$ and thus retained. Hence $\monox{\postimage{\tr}{\Frame_i}}{b} = \bigvee_{\sigma \in \postimage{\tr}{\Frame_i}}{\cubemon{\sigma}{b}} \implies \vec{x}=\vec{v}$, and thus also the conjunction
	$\mhull{\postimage{\tr}{\Frame_i}}{\bkwrch{k}} \implies \vec{x}=\vec{v}$.
}
\end{remark}

\para{Outline}
To prove \Cref{thm:abstract-hyperdiamter-bound}, the first step is to define an analog to the abstract transition system from~\Cref{def:abs-tr} that captures \Cref{alg:eepdr-ai} in the general case.
If $\bkwrch{k}$ has a DNF form with $m$ cubes, this can be done using a hypertransition system of width $m$ (\Cref{sec:abstract-hypertransition}). We then proceed to bound its diameter (\Cref{sec:hyperdiameter-bound}).

\subsection{Abstract Hypertransition System}
\label{sec:abstract-hypertransition}
We consider hypertransition systems that are dual to the classical definition~\cite[e.g.][]{DBLP:conf/lics/LarsenX90}, in that the pre-state, instead of the post-state, of a hypertransition consists of a set of states.
\begin{definition}[Hypertransition System]
A \emph{hypertransition system} (of width $m \in \mathbb{N}$) over $\States$ is a tuple $(\Init,\tr,\Bad)$ where
\begin{itemize}
	\item $\Init \subseteq \States$ is the set of initial states,
	\item $\Bad \subseteq \States$ is the set of bad states, and
	\item $\tr \subseteq \States^m \times \States$ is a hypertransition relation.
		   As a formula, it is defined over $m$ copies of $\voc$ for the pre-states, $\voc_1,\ldots,\voc_m$, and a copy $\voc'$ for the post-state.
\end{itemize}
An \emph{execution} of the system is a tree in which the leaves are states from $\Init$, the relationship between a node $\sigma'$ and its children $\sigma_1,\ldots,\sigma_m$ is that $(\sigma_1,\ldots,\sigma_m,\sigma') \models \tr$. %
A state $\sigma$ is \emph{reachable in at most $i$ steps} if there is an execution with root $\sigma$ and height at most $i$.
A state is \emph{reachable} if it is reachable in at most $i$ steps for some $i \in \mathbb{N}$.
The \emph{reachability diameter} of the system is the least $i$ such that every reachable state is reachable in $i$ steps.
\end{definition}
A standard transition system is a hypertransition system with width $m=1$.

\begin{definition}[Abstract Hypertransition System]
\label{def:abs-hypertr}
The abstract hypertransition system $(\absr{\Init},\absr{\tr},\absr{\Bad})$ of a (standard) transition system $(\Init,\tr,\Bad)$ w.r.t.\ $\bkwrch{k} = \bkcube_1 \lor \ldots \lor \bkcube_m$
is defined over $\States$ by $\absr{\Init} = \mhull{\Init}{\bkwrch{k}}$, $\absr{Bad} = \Bad$, and
\begin{equation*}
	(\sigma_1,\ldots,\sigma_m,\sigma') \in \absr{\tr} \iff
			(\sigma_1,\sigma') \in \monox{\tr \lor \Init'}{(\reflect{\cubejoin{\bkwrch{k}}},\bkcube_1)}
		\land
			\ldots
		\land
			(\sigma_m,\sigma') \in \monox{\tr \lor \Init'}{(\reflect{\cubejoin{\bkwrch{k}}},\bkcube_m)}.
\end{equation*}
\end{definition}

The central property of the abstract hypertransition system is that its $i$-reachable state capture the Kleene iterations in the $\madom{\bkwrch{k}}$ abstract domain.
\begin{lemma}
\label{thm:hyperabsract-reach}
Let $\absr{R}_i$ be the set of states reachable in $(\absr{\Init},\absr{\tr},\absr{Bad})$ w.r.t.\ $\bkwrch{k}$ (\Cref{def:abs-hypertr}) in at most $i$ steps.
Then $\absr{R}_i \equiv \Frameai_{i}$ where $\Frameai_{i}$ is the $i$'th iterate of~\Cref{alg:eepdr-ai} on $(\Init,\tr,\Bad)$ in $\madom{\bkwrch{k}}$.
\end{lemma}
This will imply that the diameter of the abstract hypertransition system equals the number of iterations needed for convergence of the abstract interpretation in $\madom{\bkwrch{k}}$.

Before proving this connection, let us explain the intuition for the abstract hypertransition system and the relation to the algorithm.

Through~\Cref{lem:monox-disjunction-cubes,lem:moncube-reflect} we can see that a hypertransition $(\sigma_1,\ldots,\sigma_m,\sigma')$ of $\absr{\tr}$ consists of three segments:
\begin{itemize}
	\item the monotonization w.r.t.\ $\cubejoin{\bkwrch{k}}$ of each $\sigma_i$ to $\widetilde{\sigma}_i$ ($\sigma_i$ is the ``protector'' state for $\widetilde{\sigma}_i$); then
	\item from each resulting state $\widetilde{\sigma}_i$, either a concrete transition $(\widetilde{\sigma}_i,\widetilde{\sigma}'_i) \in \tr$, or going back to an initial state $\widetilde{\sigma}'_i \in \Init$; and
	\item application of the monotone hull to arrive at $\sigma'$, using $\widetilde{\sigma}'_1,\ldots,\widetilde{\sigma}'_m$ as ``protector'' states, each $\widetilde{\sigma}'_i$ showing that $\sigma'$ is in the monotone overapproximation w.r.t.\ one of the cubes $\bkcube_i$ composing $\bkwrch{k}$.
\end{itemize}

The abstraction in the \emph{last} step connects the reachable states of $\absr{\tr}$ and the Kleene iterations of~\Cref{alg:eepdr-ai}.
The idea is that $\sigma' \in \mhull{\set{\widetilde{\sigma}'_1,\ldots,\widetilde{\sigma}'_m}}{\bkwrch{k}}$, and if $\widetilde{\sigma}_1,\ldots,\widetilde{\sigma}_m \in \Frameai_i$ then by~\Cref{lem:mhull-dnf-base} this implies that $\sigma' \in \mhull{{\tr}({\Frameai_i}) \cup \Init}{\bkwrch{k}}$, which, using the results of~\Cref{sec:ai}, is the next iterate $\Frameai_{i+1}$.
The key point is that the converse also holds---the monotone hull $\mhull{{\tr}({\Frameai_i}) \cup \Init}{\bkwrch{k}}$ ``factors'' to $\mhull{\set{\widetilde{\sigma}'_1,\ldots,\widetilde{\sigma}'_m}}{\bkwrch{k}}$ on all $m$ choices of protectors
$\widetilde{\sigma}'_1,\ldots,\widetilde{\sigma}'_m \in {\tr}({\Frameai_i}) \cup \Init$ (this is reminiscent of Carath\'{e}odory's theorem in convex analysis).
Unlike in~\Cref{def:abs-tr}, in this definition the protector states also come directly from $\Init$, not only from a transition of $\tr$, and essentially this is the reason for $\monox{\Init}{\bkcube_i}$ in the bound of~\Cref{thm:abstract-hyperdiamter-bound}, unlike in~\Cref{thm:abstract-diamter-bound}. This is necessary here to ``mix'' the different protector states, which do not necessarily all originate from the same frame.

As in~\Cref{def:abs-tr}, the abstraction in the \emph{first} step, which uses $\cubejoin{\bkwrch{k}}$, does not change %
reachability
and the diameter, but it can improve our diameter \emph{bound}, which is derived in~\Cref{sec:hyperdiameter-bound}.
This is achieved by also allowing a hypertransition from $(\sigma_1,\ldots,\sigma_m)$ to $\sigma'$ if other steps of $\absr{\tr}$ (concrete/init, monotone hull) can arrive at $\sigma'$ from $(\widetilde{\sigma}_1,\ldots,\widetilde{\sigma}_m)$ and if we know for certain that whenever $\sigma_1,\ldots,\sigma_m$ are reachable, then so are $\widetilde{\sigma}_1,\ldots,\widetilde{\sigma}_m$.
This is the case when $\widetilde{\sigma}_1,\ldots,\widetilde{\sigma}_m$ belong to $\mhull{\sigma_1,\ldots,\sigma_m}{\bkwrch{k}}$, because, as explained above, a monotone hull is performed in the last step of $\absr{\tr}$. %
Reachability is not extended,
because the additional abstraction in the pre-state could be mimicked by an abstraction in the post-state of the previous (abstract) step.
One way to ensure that $\widetilde{\sigma}_1,\ldots,\widetilde{\sigma}_m$ belong to $\mhull{\sigma_1,\ldots,\sigma_m}{\bkwrch{k}}$ is that each $\widetilde{\sigma}_i \in \mhull{\set{\sigma_i}}{\bkwrch{k}}$, namely, that
for every $b_j$, $\sigma_i \in \monox{\sigma_i'}{b_j}$. This is achieved when $\sigma_i \in \monox{\sigma_i'}{\cubejoin{\bkwrch{k}}}$, i.e., $\sigma_i' \in \monox{\sigma_i}{\reflect{\cubejoin{\bkwrch{k}}}}$, which is the reason for using $\cubejoin{\bkwrch{k}}$ in the monotonization of the pre-states.
Yet:
\begin{example}
In some cases, the abstraction using $\cubejoin{\bkwrch{k}}$ is weak, and is another source of non-tightness in the the bound of~\Cref{thm:abstract-hyperdiamter-bound}.
Consider the system from~\Cref{ex:monotone-several-cubes}.
In this case $\cubejoin{\bkwrch{k}}=\true$, resulting in no abstraction of the pre-state vocabulary. The DNF size of $\monox{\tr}{(\true,\bkcube_i)}$ is superpolynomial\iflong\footnote{
	Let $\sigma$ be a state
	$x_i=0$, and $\sigma'$ obtained by applying $\tr$ with two $x_{j_1},x_{j_2} \neq x_i$ variables that are $0$ in $\sigma$.
	A DNF representation must have a term $d_{\sigma}$ such that $(\sigma,\sigma') \models d$. However, $d$ must include all variables $x_r \not\in \set{x_i,x_{j_1},x_{j_2}}$ that are $0$ in $\sigma$; otherwise $\tilde{\sigma}$ where they are turned off also yields $(\tilde{\sigma},\sigma') \models d_{\sigma}$. But this can't be because the $x_r$'s are also $0$ in $\sigma'$, and $\monox{\tr}{(\true,\bkcube_i)}$ doesn't allow turning $1$ bits to $0$ (except for $x_i$).
	Consider now two states $\sigma_1,\sigma_2$ with $x_i=1,x_{j_1}=x_{j_2}=0$, and each has additional $n/2$ variables $x^{s}_{r_1},\ldots,x^{s}_{r_{n/2}} \not\in \set{x_i,x_{j_1},x_{j_2}}$ of value $0$ ($s \in \set{1,2}$), and the rest of the variables are $1$. By the above argument, $d_{\sigma_1}$ requires that all literals $x^{1}_{r_p}$ are $0$, which implies that $\sigma_2 \not\models d_{\sigma_1}$ if $\set{x^{2}_{r_1},\ldots,x^{2}_{r_{n/2}}} \not\subseteq \set{x^{1}_{r_1},\ldots,x^{1}_{r_{n/2}}}$.
	Therefore, every choice of $n/2$ variables yields a non-comparable term, and there are $\binom{n}{n/2} = \Theta\left(4^n/\sqrt{n}\right)$ such choices.
},
\else
~(see the extended version~\cite{extendedVersion}),
\fi
leading to a superpolynomial bound on the number of frames.
However, $\Lambda$-PDR with $k=0$ converges in $\Frame_1$ (see~\Cref{ex:monotone-several-cubes}).
\end{example}

\toolong{
\begin{remark}%
At first sight, it would seem that the product of monotonizations in the bound of~\Cref{thm:abstract-hyperdiamter-bound} is unnecessary, and that one could study the convergence of $\Lambda$-PDR w.r.t.\ $\bkwrch{k}$ by the convergence w.r.t.\ the simpler (and larger) set $\cubejoin{\bkwrch{k}}$. Since $\bkwrch{k} \subseteq \cubejoin{\bkwrch{k}}$, the overapproximation is tighter with $\cubejoin{\bkwrch{k}}$ (\Cref{lem:mhull-monotonicity}), so it would seem that the number of iterations with $\mhull{\cdot}{\bkwrch{k}}$ must be less than with $\mhull{\cdot}{\cubejoin{\bkwrch{k}}}$. However, this is not so. The reason is that $\Lambda$-PDR with $\mhull{\cdot}{\cubejoin{\bkwrch{k}}}$ might be converging to an inductive invariant that is not present in $\mspan{\bkwrch{k}}$: $\bkwrch{k} \subseteq \cubejoin{\bkwrch{k}}$ implies $\mspan{\cubejoin{\bkwrch{k}}} \subseteq \bkwrch{k}$. Thanks to such ``new'' invariants, convergence could be faster with $\mhull{\cdot}{\cubejoin{\bkwrch{k}}}$ than with $\mhull{\cdot}{\bkwrch{k}}$.
\end{remark}

We now formally prove the connection between reachable states of the abstract transition system and iterations of the algorithm.
\begin{proof}[Proof of~\Cref{thm:hyperabsract-reach}]
We first prove a similar result for a slightly simpler, ``less abstract'' hypertransition system, where abstraction is performed in the post-state but not in the pre-state.
Define a hypertransition system $(\abs{\Init},\abs{\tr},\abs{\Bad})$ over $\States$
$\abs{\Init} = \mhull{\Init}{\bkwrch{k}}$, $\abs{Bad} = \Bad$, and
\begin{equation*}
\begin{split}
	(\sigma_1,\ldots,\sigma_m,\sigma') \in \abs{\tr} \iff
	\exists \sigma''_1,\ldots,\sigma''_m. \qquad
	& ((\sigma_1,\sigma''_1) \in \tr \lor \sigma''_1 \in \Init) \land \sigma' \in \cubemon{\sigma''_1}{\bkcube_1} \qquad \land
	\\
	& \ldots
	\\
	& ((\sigma_m,\sigma''_m) \in \tr \lor \sigma''_m \in \Init) \land \sigma' \in \cubemon{\sigma''_m}{\bkcube_m}.
\end{split}
\end{equation*}
(In fact, $\abs{\tr}=\bigwedge_{i=1}^{m}{(\monox{\tr \lor \Init'}{(\true,\bkcube_i)})[\voc_i,\voc']}$).

Denote by $\abs{R}_i$ the set of states reachable in $(\abs{\Init},\abs{\tr},\abs{Bad})$ in at most $i$ steps.
We argue that $\abs{R}_i \equiv \Frameai_{i}$,
by induction on $i$.
Initially, $\abs{R}_0 = \abs{\Init} = \mhull{\Init}{\bkwrch{k}} = \Frameai_0$.
For the step, the set $\abs{R}_{i+1}$ is the set of states reachable in at most $i+1$ steps in the hypertransition system, which is $\abs{R}_{i+1} = \abs{R}_i \cup \tr(\abs{R}_i)$ where $\abs{\tr}(\abs{R}_i)$ is the set of states $\sigma'$ so that there are $\sigma_1,\ldots,\sigma_m \in \abs{R}_i$ such that $(\sigma_1,\ldots,\sigma_m,\sigma') \in \abs{\tr}$. By the definition of $\abs{\tr}$,
\begin{align*}
	{\abs{\tr}}({\abs{R_i}})
	&= \bigvee_{\sigma''_1,\ldots,\sigma''_m \in {\tr}({\abs{R}_i}) \cup \Init}{\left(\cubemon{\sigma''_1}{\bkcube_1} \land \ldots \land \cubemon{\sigma''_m}{\bkcube_m}\right)}
	\\
	\intertext{By distributivity of conjunction over disjunction,}
	\\
	&= \left(\bigvee_{\sigma''_1 \in {\tr}({\abs{R}_i}) \cup \Init}{\cubemon{\sigma''_1}{\bkcube_1}}\right)
		\land
		\ldots
		\land
		\left(\bigvee_{\sigma''_m \in {\tr}({\abs{R}_i}) \cup \Init}{\cubemon{\sigma''_m}{\bkcube_m}}
		\right)
	\\
	\intertext{By~\Cref{lem:monox-disjunction-cubes}, this is}
	&= \monox{{\tr}({\abs{R}_i}) \cup \Init}{\bkcube_1} \land \ldots \land \monox{{\tr}({\abs{R}_i}) \cup \Init}{\bkcube_m}
	\\
	\intertext{By~\Cref{lem:mhull-dnf-base}, this amounts to}
	&= \mhull{{\tr}({\abs{R}_i}) \cup \Init}{\bkwrch{k}}
	\\
	\intertext{which by the induction hypothesis is}
	&= \mhull{{\tr}({\Frameai_{i}}) \cup \Init}{\bkwrch{k}}.
\end{align*}
In the terminology of~\Cref{sec:ai-background}, as $\Frameai_{i}=(\abs{F})^{i+1}(\abs{\bot})$, we have obtained $\abs{\tr}(\abs{R}_i)=\malpha{\bkwrch{k}}(\tr(({\abs{F}_{\Init,\tr}})^{i+1}(\abs{\bot}))\cup \Init)=(\abs{F}_{\Init,\tr})^{i+2}(\abs{\bot})$. Because
$(\abs{F}_{\Init,\tr})^{i+1}(\abs{\bot}) \abs{\sleq} (\abs{F}_{\Init,\tr})^{i+2}(\abs{\bot})$,
the result is that $\abs{R}_i \subseteq \tr(\abs{R}_i)$, and
$\abs{R}_{i+1} = \tr(\abs{R}_i) \cup \abs{R}_i = \tr(\abs{R}_i) = \Frameai_{i+1}$, as required.

It remains to show that $\absr{R}_i = \abs{R}_i$, i.e., that the $i$-reachable states of $\abs{\tr},\absr{\tr}$ coincide.

First, for every set of states $S$ it holds that $\postimage{\abs{\tr}}{S} \subseteq \postimage{\absr{\tr}}{S}$. This is because if $(\sigma_1,\ldots,\sigma_m,\sigma') \in \abs{\tr}$, then by definition there are $\sigma''_1,\ldots,\sigma''_m$ such that $(\sigma_i,\sigma''_i) \in \tr \lor \Init'$ for every $i$ and $\sigma' \in \moncube{\sigma''_i}{\bkcube_i}$.
Considering the product monotone order, $(\sigma_i,\sigma''_i) \leq_{(\cdot,\bkcube_i)} (\sigma_i,\sigma')$, and so $(\sigma_i,\sigma''_i) \in \tr \lor \Init' \implies (\sigma_i,\sigma') \in \monox{\tr \lor \Init'}{(\reflect{\cubejoin{\bkwrch{k}}},\bkcube_i)}$. This for every $i$; by the definition of $\absr{\tr}$ this implies that $(\sigma_1,\ldots,\sigma_m,\sigma') \in \absr{\tr}$. This means that $\sigma' \in \postimage{\absr{\tr}}{S}$ since $\sigma_1,\ldots,\sigma_m \in S$.

Second, we show that for any $S \in \mspan{\bkwrch{k}}$ it holds that $\postimage{\absr{\tr}}{S} \subseteq \postimage{\abs{\tr}}{S}$. Let $(\sigma_1,\ldots,\sigma_m,\sigma') \in \absr{\tr}$, where $\sigma_1,\ldots,\sigma_m \in S$. By the definition of $\absr{\tr}$, $(\sigma_i,\sigma') \models \monox{\tr \lor \Init'}{(\reflect{\cubejoin{\bkwrch{k}}},\bkcube_i)}$, and so there exist $(\widetilde{\sigma}_1,\widetilde{\sigma}'_1),\ldots,(\widetilde{\sigma}_m,\widetilde{\sigma}'_m) \in \tr \lor \Init'$ such that for every $i$,
\begin{itemize}
	\item $\sigma_i \models \moncube{\widetilde{\sigma}_i}{\reflect{\cubejoin{\bkwrch{k}}}}$---by~\Cref{lem:moncube-reflect}, this means that $\widetilde{\sigma}_i \models \moncube{\sigma_i}{\cubejoin{\bkwrch{k}}}$.
	Since $\sigma_i \in S$, by~\Cref{lem:monox-disjunction-cubes}, $\widetilde{\sigma}_i \in \monox{S}{\cubejoin{\bkwrch{k}}}$. It follows that $\widetilde{\sigma}_i \in \mhull{S}{\bkwrch{k}}$ (because $\bkwrch{k} \subseteq \cubejoin{\bkwrch{k}}$ implies $\mhull{S}{\cubejoin{\bkwrch{k}}} \subseteq \mhull{S}{\bkwrch{k}}$ and $\mhull{S}{\cubejoin{\bkwrch{k}}} \equiv \monox{S}{\cubejoin{\bkwrch{k}}}$ by~\Cref{lem:mhull-dnf-base}). From the premise that $S \in \mspan{\bkwrch{k}}$, $\mhull{S}{\bkwrch{k}} \equiv S$, and we have $\widetilde{\sigma}_i \in S$.

	\item $\sigma' \models \moncube{\widetilde{\sigma}'_i}{\bkcube_i}$.
\end{itemize}
Writing $\sigma''_i = \widetilde{\sigma}'_i$ shows that $(\widetilde{\sigma}_1,\ldots,\widetilde{\sigma}_m,\sigma') \in \abs{\tr}$, because for every $i$ we have $\widetilde{\sigma}_i \in S$, $(\widetilde{\sigma}_i,\widetilde{\sigma}'_i) \in \tr \lor \Init'$, and $\sigma' \models \moncube{\widetilde{\sigma}'_i}{\bkcube_i}$.
This shows that $\sigma' \in \postimage{\abs{\tr}}{S}$, as required.

The first part of the argument (and induction on $i$) shows that $\abs{R}_i \subseteq \absr{R}_i$.
We have shown that $\abs{R}_i = \Frameai_i$, which in particular implies that
always $\abs{R}_i \in \bkwspan{k}$; therefore, the second argument above shows that $\absr{R}_i \subseteq \abs{R}_i$.
The claim follows.
\end{proof}
}
\begin{corollary}
\label{thm:abstract-hyperdiameter-eepdr}
Let $(\absr{\Init},\absr{\tr},\absr{\Bad})$ be the abstract hypertransition system w.r.t.\ $\bkwrch{k}$ (\Cref{def:abs-hypertr}).
If $(\absr{\Init},\absr{\tr},\absr{\Bad})$ is safe and its reachability diameter is $s$, then
$\eepdr(\Init,\tr,\Bad,k)$ converges in frame at most $s+1$.
If $(\absr{\Init},\absr{\tr},\absr{\Bad})$ reaches a bad state in $s$ steps, then $\eepdr(\Init,\tr,\Bad,k)$ fails (\cref{ln:eepdr-restart}) in frame at most $s+1$.
\end{corollary}
\toolong{
\begin{proof}
Follows from~\Cref{thm:hyperabsract-reach} similarly to the proof of~\Cref{thm:abstract-diameter-eepdr} from~\Cref{thm:absract-reach}.
\end{proof}
}

\subsection{Hyperdiameter Bounds via a Joint Abstract Cover}
\label{sec:hyperdiameter-bound}

In this section, we bound the diameter of the abstract transition in order to obtain the convergence bound of~\Cref{thm:abstract-hyperdiamter-bound}.
The proof is based on a diameter bound similar to the case of standard transition systems.
\begin{lemma}
\label{lem:hyperdiam-dnf}
The reachability diameter of a hypertransition system $(\Init,\tr,\Bad)$ is bounded by $\dnfsize{\tr}$.
\end{lemma}
\toolong{
\begin{proof}
Fix a minimal DNF representation of $\tr$. Thinking about each disjunct of $\tr$ as an action, every transition from $m$ children to a parent can be labeled by at least one action.
Consider a path from the root to the leaves in an execution tree. With these actions, if an action $a$ labels two (hyper)transitions $\sigma_{i_1} \overset{a}{\rightarrow} \sigma_{i_1+1},\sigma_{i_2} \overset{a}{\rightarrow} \sigma_{i_2+1}$, the segment of the tree between the occurrences, $\sigma_{i_1+1},\ldots,\sigma_{i_2}$ can be dropped, replacing the and the resulting trace is still valid (and terminates at the same state)---this is because if $(\sigma^1_{i_1},\ldots,\sigma^m_{i_1},\sigma_{i_1+1}) \models a$ and likewise $(\sigma^1_{i_2},\ldots,\sigma^m_{i_2},\sigma_{i_2+1}) \models a$ then also $(\sigma^1_{i_1},\ldots,\sigma^m_{i_1},\sigma_{i_2+1}) \models a$, because $a$, which is a cube, can be decomposed to $a_{\textit{pre}_1} \land \ldots \land a_{\textit{pre}_m} \land a_\textit{post}$ where all the literals in $a_{\textit{pre}_i}$ are in the $i$'th pre-state copy $\voc_i$ and those in $a_\textit{post}$ are in $\voc'$.
Overall, every state that can be reached from a set of leaf states can do so by an execution where each action appears at most once on each path of the tree, and thus the diameter is bounded by $\dnfsize{\tr}$.
\end{proof}
}

\toolong{
\begin{proof}[Proof of~\Cref{thm:abstract-hyperdiamter-bound}]
Denote the bound in the theorem by $q + 1$.
From the distributivity of the conjunction in $\absr{\tr}$, $\dnfsize{\absr{\tr}} \leq \prod_{i=1}^{m}{\dnfsize{\monox{\tr \lor \Init'}{(\reflect{\cubejoin{\bkwrch{k}}},\bkcube_i)}}}$,
and
\begin{align*}
\dnfsize{\monox{\tr \lor \Init'}{(\reflect{\cubejoin{\bkwrch{k}}},\bkcube_i)}}
&\underset{\rm \Cref{lem:bshouty-mon-mindnf}}{=}
\dnfsize{\monox{\tr}{(\reflect{\cubejoin{\bkwrch{k}}},\bkcube_i)}
	\lor \monox{\Init'}{(\reflect{\cubejoin{\bkwrch{k}}},\bkcube_i)}}
\\
&\leq
	\dnfsize{\monox{\tr}{(\reflect{\cubejoin{\bkwrch{k}}},\bkcube_i)}}
	+
	\dnfsize{\monox{\Init'}{(\reflect{\cubejoin{\bkwrch{k}}},\bkcube_i)}}
\\
&=
	\dnfsize{\monox{\tr}{(\reflect{\cubejoin{\bkwrch{k}}},\bkcube_i)}}
	+
	\dnfsize{\monox{\Init}{\bkcube_i}},
\end{align*}
overall yielding that $\dnfsize{\absr{\tr}} \leq q$.

By~\Cref{thm:abstract-hyperdiameter-eepdr}, the number of iterations before convergence or failure of $\Lambda$-PDR is bounded by %
1 plus
the reachability diameter of $(\absr{\Init},\absr{\tr},\absr{\Bad})$,
which by~\Cref{lem:hyperdiam-dnf} is at most $q$.
\end{proof}
}

\section{Forward Reachability in $\Lambda$-PDR and Others}
\label{sec:itp-friends}
This section highlights the importance of the successive overapproximation embodied in the Kleene iterations of $\Lambda$-PDR by contrasting $\Lambda$-PDR with the treatment of forward reachability in other invariant inference algorithms.

\para{Exact forward reachability}
Exact forward reachability iterates $R_0=\Init, R_{i+1}=\postimage{\tr}{R_i}$, so that $R_i$ is the set of states reachable in at most $i$ steps (without any overapproximation).
We have shown that in some cases $\Lambda$-PDR can converge in a significantly lower number of iterations than exact forward reachability, stated formally in the following lemma.
\begin{lemma}
There exists a family of transition systems $(\Init,\tr,\Bad)$ over $\voc$ with $\card{\voc}=n$ and $k=\bigO(1)$ such that $\eepdr(\Init,\tr,\Bad,k)$ converges in $\textit{poly}(n,k)$ iterations, whereas exact forward reachability converges in $\Omega(2^n)$ iterations.
\end{lemma}
\begin{proof}
See e.g.~\Cref{sec:overview-successive} and~\Cref{ex:running-all-frames}. %
\end{proof}
This gap reflects a gap between the diameter of the original system $(\Init,\tr)$ and the diameter of the abstract system $(\absr{\Init},\absr{\tr})$ (\Cref{def:abs-tr,thm:abstract-diameter-eepdr}).

\para{Dual interpolation}
\label{sec:dual-itp-compare}
The essence of \emph{interpolation-based inference} (ITP)~\cite{DBLP:conf/cav/McMillan03} is generalizing from proofs of \emph{bounded} unreachability.
We consider the time-dual~\cite[e.g.][Appendix A]{DBLP:journals/pacmpl/FeldmanISS20} of this approach, generalizing from bounded unreachabilty \emph{from} the initial states, rather than unreachability \emph{to} the bad states,
in line with our focus here on the treatment of forward reachability.
Specifically, \Cref{alg:dual-itp-termmin} is based on (the time-dual of) a model-based ITP algorithm~\cite{DBLP:conf/hvc/ChocklerIM12,DBLP:conf/lpar/BjornerGKL13} whose generalization procedure was inspired by PDR.

\iflong
\begin{wrapfigure}{r}{0.45\textwidth}
\vspace{-0.4cm}
\begin{minipage}{0.45\textwidth}
\begin{algorithm}[H]
\caption{Dual Model-Based ITP, based on~\cite{DBLP:conf/hvc/ChocklerIM12,DBLP:conf/lpar/BjornerGKL13}}
\label{alg:dual-itp-termmin}
\begin{algorithmic}[1]
\begin{footnotesize}
\Procedure{Dual-Model-Based-ITP}{$\Init$,$\tr$,$\Bad$,$s$}
	\State $\varphi \gets \neg \Bad$ $\label{ln:dual-itp-frame0}$
    \While{$\varphi$ not inductive}
		\State take $\sigma_b \in \varphi$ s.t.\ $\tr(\sigma_b) \not\subseteq \varphi$ $\label{ln:dual-itp-cex}$
		\If{$\sigma_b \in \mathcal{R}_s$} $\label{ln:dual-itp-restart}$
			\State \textbf{restart} with larger $s$
		\EndIf
		\State take minimal $c \subseteq \neg \sigma_b$ s.t.\ $\mathcal{R}_s \implies c$  $\label{ln:dual-itp-bmc}$
		\State $\varphi \gets \varphi \land c$ $\label{ln:dual-itp-learn}$
	\EndWhile
	\State \Return $\varphi$
\EndProcedure
\end{footnotesize}
\end{algorithmic}
\end{algorithm}
\end{minipage}
\vspace{-0.5cm}
\end{wrapfigure}
The algorithm is parametrized by a forward-exploration bound $s$. It refines a candidate $\varphi$ starting from the candidate that excludes just the bad states (\cref{ln:dual-itp-frame0}).
In each iteration,
the algorithm samples a pre-state $\sigma_b$ of a counterexample to the induction of $\varphi$, a state in $\varphi$ that in one step reaches states outside $\varphi$ (\cref{ln:dual-itp-cex}).
Instead of excluding just the counterexample---similarly to PDR---
\iflong\else
\begin{wrapfigure}{r}{0.45\textwidth}
\vspace{-0.4cm}
\begin{minipage}{0.45\textwidth}
\begin{algorithm}[H]
\caption{Dual Model-Based ITP, based on~\cite{DBLP:conf/hvc/ChocklerIM12,DBLP:conf/lpar/BjornerGKL13}}
\label{alg:dual-itp-termmin}
\begin{algorithmic}[1]
\begin{footnotesize}
\Procedure{Dual-Model-Based-ITP}{$\Init$,$\tr$,$\Bad$,$s$}
	\State $\varphi \gets \neg \Bad$ $\label{ln:dual-itp-frame0}$
    \While{$\varphi$ not inductive}
		\State take $\sigma_b \in \varphi$ s.t.\ $\tr(\sigma_b) \not\subseteq \varphi$ $\label{ln:dual-itp-cex}$
		\If{$\sigma_b \in \mathcal{R}_s$} $\label{ln:dual-itp-restart}$
			\State \textbf{restart} with larger $s$
		\EndIf
		\State take minimal $c \subseteq \neg \sigma_b$ s.t.\ $\mathcal{R}_s \implies c$  $\label{ln:dual-itp-bmc}$
		\State $\varphi \gets \varphi \land c$ $\label{ln:dual-itp-learn}$
	\EndWhile
	\State \Return $\varphi$
\EndProcedure
\end{footnotesize}
\end{algorithmic}
\end{algorithm}
\end{minipage}
\vspace{-0.5cm}
\end{wrapfigure}
\fi
the algorithm seeks a minimal clause $c$ over the literals that are falsified in $\sigma_b$ that does not exclude a state from $\mathcal{R}_s$, the set of states that the system can reach in $s$ steps (\cref{ln:dual-itp-bmc}), and conjoins $c$ to the %
candidate (\cref{ln:dual-itp-learn}).

The complexity of this algorithm was recently studied by~\citet{DBLP:journals/pacmpl/FeldmanSSW21}, who showed that the forward-exploration bound $s$ is sufficient to discover an inductive invariant when $s$ steps reach the entire \emph{inner boundary} of $I$,
\iflong\else
\vspace{0.3cm}
\fi
\begin{equation*}
	\boundarypos{I} \eqdef \set{\sigma^+ \, \mid \, \exists \sigma^-. \ \sigma^+ \models I, \, \sigma^- \models \neg I, \, \mbox{Hamming-Distance}(\sigma^+,\sigma^-)=1}.
\end{equation*}
\begin{theorem}[\citet{DBLP:journals/pacmpl/FeldmanSSW21}]
\label{lem:itp-fence-condition}
Let $I$ be an inductive invariant for $(\Init,\tr,\Bad)$,
and $\mathcal{R}_s$ the set of states reachable in at most $s$ steps in $(\Init,\tr,\Bad)$.
If
$\boundarypos{I} \subseteq \mathcal{R}_s$,
then a forward bound of $s$ suffices for $\mbox{Dual-Model-Based-ITP}(\Init,\tr,\Bad,s)$ to successfully find an inductive invariant.
\end{theorem}
In the example of~\Cref{fig:skip-counter} (from~\Cref{sec:overview}), this does not hold for the invariant in~\Cref{eq:skip-counter-invariant} unless $s=\Omega(2^n)$ (for example, $\vec{x}=110\ldots0,\vec{y}=0\ldots0,z=0 \in \boundarypos{I}$ but reachable only in $\Omega(2^n)$ steps).
In contrast, we prove that for $\Lambda$-PDR, it is enough that $\boundarypos{I}$ is $s$-reachable in the \emph{abstract} (hyper)system, which interleaves concrete steps and abstraction (see~\Cref{def:abs-hypertr}), and thus can reach $\boundarypos{I}$ in fewer steps, which would result in convergence of $\Lambda$-PDR with a smaller number of frames:
\begin{theorem}
\label{lem:eepdr-fence-condition}
Let $I \in \bkwspan{k}$ be an inductive invariant for $(\Init,\tr,\Bad)$,
and $\absr{\mathcal{R}}_s$ the set of states reachable in at most $s$ steps in $(\absr{\Init},\absr{\tr},\absr{\Bad})$ (\Cref{def:abs-hypertr}).
If $\boundarypos{I} \subseteq \absr{\mathcal{R}}_s$,
then $s+1$ frames suffice for $\eepdr(\Init,\tr,\Bad,k)$ to successfully find an inductive invariant.
\end{theorem}
The two results can be proved similarly, using tools from the monotone theory. In both cases, the argument is that when $s$ is large enough, the monotone hull of the current candidate must include the entire $I$.
In~\Cref{alg:dual-itp-termmin}, the argument is that always $I \subseteq \mhull{\mathcal{R}_s}{\mathcal{C}_i} \subseteq \varphi$, where $\mathcal{C}_i$ is the set of counterexamples $\sigma_b$ that the algorithm has encountered so far.\yotamforlater{opportunity to use this argument to infer CDNF invariants using interpolation?}
In~\Cref{alg:eepdr}, the argument is that $I \subseteq \mhull{\absr{\mathcal{R}}_s}{\bkwrch{k}}$.
Both rely on the following fact about the monotone hull of a boundary of a set:
\begin{lemma}
\label{lem:mhull-boundary}
Let $I,S,B$ be sets of states s.t.\ $\boundarypos{I} \subseteq S$ and $B \cap I = \emptyset$.
Then $I \subseteq \mhull{S}{B}$.
\end{lemma}
\begin{proof}
Let $\sigma \in I$. For every $b \in B$, assume for the sake of contradiction that $\sigma$ is not in $\monox{S}{b}$. By~\Cref{lem:monox-conjunction-clauses}, there is some cube $e$ such that $b \models e$ and $\sigma \models e$ but $S \implies \neg e$. In particular, $\boundarypos{I} \implies \neg e$. Consider some shortest path between $\sigma,b$ in the Hamming cube. Because $\sigma \models I, b \not\models I$, there is some crossing point $\sigma^{+} \in \boundarypos{I}$ on that path. This state $\sigma^{+}$ agrees with $\sigma,b$ on the literals on which they agree, which include all the literals in $e$, since this is a cube (a conjunction of literals). Hence also $e \models \sigma^{+}$, but this is a contradiction to $\boundarypos{I} \implies \neg e$.
\end{proof}
\iflong
We use this to prove the above claims:
\begin{proof}[Proof of~\Cref{lem:itp-fence-condition}]
We show that always $I \subseteq \varphi$, which implies that $\sigma_b \not\models I$ (because otherwise $\tr(\sigma_b) \subseteq I \subseteq \varphi$, in contradiction to the choice of $\sigma_b$), and because $\mathcal{R}_s \subseteq I$ this implies that $\sigma_b \not\in \mathcal{R}_s$ and so no restart is required. Because $\varphi$ strictly decreases in each iteration and the number of non-equivalent formulas is finite, this implies that the algorithm terminates with an inductive invariant.

First note that by \Cref{thm:mhull-conjunctive}, always $\varphi \in \mspan{\mathcal{C}_i}$. By induction on the iterations of the algorithm, always $\mathcal{R}_s \subseteq \varphi$, and hence by~\Cref{lem:mhull-monotonicity} always $\mhull{\mathcal{R}_s}{\mathcal{C}_i} \subseteq \mhull{\varphi}{\mathcal{C}_i} \equiv \varphi$.
We argue that $I \subseteq \varphi$ and $I \cap \mathcal{C}_i = \emptyset$, by induction.
Initially, $I \subseteq \neg \Bad$ because it is an inductive invariant, and $\mathcal{C}_0 = \emptyset$.
For a step, $\sigma_b \not\in I$ because otherwise $\tr(\sigma_b) \subseteq I \subseteq \varphi$.
To show that $I \subseteq \varphi \land \neg d$, by the induction hypothesis $\mathcal{R}_s \subseteq I \subseteq \varphi$ and by construction $\mathcal{R}_s \subseteq \neg d$. Hence, by~\Cref{lem:mhull-boundary}, $I \subseteq \mhull{\mathcal{R}_s}{\mathcal{C}_i \cup \set{\sigma_b}} \subseteq \varphi \land \neg d$, as required.
\end{proof}
\begin{proof}[Proof of~\Cref{lem:eepdr-fence-condition}]
The set of states reachable in $s$ steps in $(\absr{\Init},\absr{\tr},\absr{\Bad})$ is $\Frameai_s$ of the Kleene iterations (\Cref{thm:hyperabsract-reach}).
We can apply~\Cref{lem:mhull-boundary}, because $\postimage{\tr}{\Frame_{s-1}}$ of $\Lambda$-PDR includes all $s$-reachable states (by properties~\ref{it:frames-start}--\ref{it:frames-end} in~\Cref{sec:overview-frame-props}) and $I \cap \bkwrch{k} = \emptyset$ since $I$ is an inductive invariant. We obtain that $\Frameai_s = \mhull{\postimage{\tr}{\Frameai_{s-1}}}{\bkwrch{k}}$ contains $I$. It cannot ``overshoot'' beyond $I$ due to~\Cref{lem:eepdr-lfp}.
Apply~\Cref{lem:eepdr-ai-iterations} for the connection to $\Lambda$-PDR.
\end{proof}
\else
Proofs of~\Cref{lem:itp-fence-condition} and~\Cref{lem:eepdr-fence-condition} are derived from~\Cref{lem:mhull-boundary} in the extended version~\cite{extendedVersion}. %
\fi
In essence, these different criteria for when the forward-exploration of the algorithm is sufficient reflect the difference in how the algorithms generalize: per counterexample, both find a minimal clause that does not exclude states from some form of forward reachability, but in $\Lambda$-PDR this is an abstraction of forward reachability, whereas
\Cref{alg:dual-itp-termmin}
uses exact forward reachability.

This difference also manifests in different outcomes of~\Cref{alg:eepdr} and~\Cref{alg:dual-itp-termmin} on the running example of~\Cref{fig:skip-counter}.
For every $s < 2^n$ there is an execution of
~\Cref{alg:dual-itp-termmin}
that fails (\cref{ln:dual-itp-restart}) because it includes reachable states as counterexamples to exclude (for example, the first counterexample in the execution of~\Cref{alg:dual-itp-termmin} is $\sigma_b=(\vec{x}=10\ldots00,\vec{y} = 11\ldots10,z=1)$, which can be generalized to $c=(x_n=0)$ that inadvertently excludes also reachable states such as $\vec{x} = 10\ldots01,\vec{y} = 00\ldots00,z=0$),
although $s=\bigO(1)$ suffices for $\Lambda$-PDR (\Cref{ex:running-all-frames}).

Finally, we remark that~\Cref{alg:dual-itp-termmin} does use a form of successive overapproximation.
By repeatedly generating counterexamples to induction (\cref{ln:dual-itp-cex}), it in a sense uses reverse frames that overapproximate backward reachability.
While both~\Cref{alg:dual-itp-termmin} and~\Cref{alg:eepdr} learn lemmas by minimizing a term w.r.t.\ a forward-reachability analysis in order to block a counterexample from a backward-reachability analysis,
\Cref{alg:eepdr} employs successive overapproximation is in the former analysis, and \Cref{alg:dual-itp-termmin} in the latter.
As we have seen, this successive overapproximation in counterexample generation is not sufficient for~\Cref{alg:dual-itp-termmin} to successfully infer an invariant for the example of~\Cref{fig:skip-counter}. However, it does alleviate the requirement that $I \in \bkwrch{k}$, which is necessary in~\Cref{lem:eepdr-fence-condition} %
but not
in~\Cref{lem:itp-fence-condition}.\footnote{
	The original, non time-dual version of the algorithm, has frames going forward, such as in PDR, but the roles of backward- and forward-reachability in generalization are reversed. This algorithm ``overshoots'' on the example of~\Cref{fig:skip-counter} unless $s=\Omega(2^n)$\yotamsmall{algorithm initialized with $\vec{x
	}=00\ldots00,\vec{y}=0\ldots0,z=0$, cti $\vec{x
	}=00\ldots01,\vec{y}=0\ldots0,z=0$, can be generalized to $\vec{x
	}=00\ldots01$ unless $s \geq 2^{n-1}-2$, once we add this term to the invariant we've added backward reachable states and we're doomed}, but we focus here on overapproximations that are too tight (rather than too loose), the direction in which $\Lambda$-PDR is informative of PDR.
}
\section{Between $\Lambda$-PDR and PDR: Best Abstraction and Even Better}
\label{sec:vs-pdr}
In each frame, $\Lambda$-PDR includes all possible generalizations, %
which we have shown to amount in $\Frame_{i+1}$ to the the best abstraction of $\postimage{\tr}{\Frame_i}$ in the abstract domain $\madom{\bkwrch{k}}$ (\Cref{lem:best-abstraction}).
\begin{changebar}
Its frames are thus the strongest (contain fewest states) that satisfy all the properties of frames listed in~\Cref{sec:overview-frame-props}---the standard ones as well as the monotone span of backward reachable states:
\begin{lemma}
\label{lem:lambda-frames-minimality}
The frames $\Frame_0,\Frame_1,\ldots$ of $\Lambda$-PDR are the least (w.r.t.\ $\implies$) s.t.\ for every $i$,
\begin{inparaenum}
		\setcounter{enumi}{\getrefnumber{it:frames-init}-1}
		\item $\Init \implies \Frame_0$,
		\setcounter{enumi}{\getrefnumber{it:frames-monotone}-1}
		\item $\Frame_i \implies \Frame_{i+1}$,
		\setcounter{enumi}{\getrefnumber{it:frames-onestep-overapprox}-1}
		\item $\tr({\Frame_i}) \implies \Frame_{i+1}$, and
		\setcounter{enumi}{\getrefnumber{it:frames-mbasis}-1}
		\item $\Frame_i \in \bkwspan{k}$.
\end{inparaenum}
\end{lemma}
\toolong{
\begin{proof}
That the frames of $\Lambda$-PDR satisfy the properties is immediate from the relationship $\Frame_0 = \Init$, $\Frame_{i+1} = \malpha{\bkwrch{k}}(\postimage{\tr}{\Frame_i})$. Minimality is from best abstraction (\Cref{lem:best-abstraction}) and induction on $i=0,1,\ldots$: let $\widetilde{\Frame}_0,\widetilde{\Frame}_1,\ldots$ another sequence that satisfies the properties. By property~\ref{it:frames-init}, $\Init = \Frame_0 \implies \widetilde{\Frame}_0$. For the step, assume that $\Frame_i \implies \widetilde{\Frame}_i$. Then from properties~\ref{it:frames-monotone}~and~\ref{it:frames-onestep-overapprox}, $\postimage{\tr}{\widetilde{\Frame}_i} \implies \widetilde{\Frame}_{i+1}$, and in particular also $\postimage{\tr}{\Frame_i} \implies \widetilde{\Frame}_{i+1}$. From property~\ref{it:frames-mbasis}, $\widetilde{\Frame}_{i+1} \in \bkwspan{k}$. Putting these together, \Cref{lem:best-abstraction} implies that $\Frame_{i+1} = \malpha{\bkwrch{k}}(\postimage{\tr}{\Frame_i}) \implies \widetilde{\Frame}_{i+1}$, as required.
\end{proof}
}
In contrast to $\Lambda$-PDR, standard PDR ``samples'' counterexamples and generalizations, and it does not produce in $\Framepdr_{i+1}$ the least abstraction of $\postimage{\tr}{\Framepdr_i}$. Its frames are nevertheless characterized as abstractions (not necessarily the least abstraction) in the same domain:
\begin{lemma}
\label{lem:pdr-also}
At any point during the execution of $\mbox{PDR}(\Init,\tr,\Bad)$ (\Cref{alg:pdr}) when it has at most $N$ frames, $\Framepdr_i \in \bkwspan{N}$ for every $1 \leq i \leq N$. %
\end{lemma}
\toolong{
\begin{proof}
In a call to $\textsc{block}(\sigma_b, i+1)$, it holds $\sigma_b \in \bkwrch{N-i}$, by induction on the recursive calls:
The first call $\textsc{block}(\sigma_b, N+1)$ in~\cref{ln:pdr-block-bad} has $\sigma_b \in \bkwrch{0}=\Bad$, by~\cref{ln:pdr-sample-bad}.
In each recursive call from $(\sigma_b,i+1)$ to $(\sigma,i+1-1)$, in~\cref{ln:pdr-back-sample-prestate} the new counterexample $\sigma$ reaches the counterexample in the parent call $\sigma_b$ in one step, so $\sigma_b \in \bkwrch{N-i}$ implies $\sigma \in \bkwrch{N-i+1} = \bkwrch{N-(i-1)}$ as required.
Since $\bkwrch{N-i} \subseteq \bkwrch{N}$, this ensures that $\sigma_b \in \bkwrch{N}$ in every call to $\textsc{block}(\sigma_b, i+1)$.

Hence, when the algorithm strengthens $\Framepdr_i$ in~\cref{ln:pdr-strengthen}, it is always with a clause $c$ such that $\sigma_b \not\models c$ where $\sigma_b \in \bkwrch{N}$.
This implies that $\Framepdr_i \in \bkwspan{N}$ (see~\Cref{sec:monotone-basis}), completing the proof.
\end{proof}
}
In particular, this shows that PDR overapproximates the frames that $\Lambda$-PDR generates:
\begin{corollary}
\label{cor:lambda-pdr-underapproximates-pdr}
At any point during the execution of $\mbox{PDR}(\Init,\tr,\Bad)$ (\Cref{alg:pdr}) when it has at most $N$ frames, its $i$'th frame, $\Framepdr_i$, satisfies $\Frame_i \implies \Framepdr_i$, where $\Frame_i$ is the $i$'th frame of $\mbox{$\Lambda$-PDR}(\Init,\tr,\Bad,N)$ (\Cref{alg:eepdr}).
\end{corollary}
\toolong{
\begin{proof}
The frames of~\Cref{alg:pdr} satisfy the properties in the premise of~\Cref{lem:lambda-frames-minimality}---all are standard except for the one shown in~\Cref{lem:pdr-also}.

A more direct argument, outlined in~\Cref{sec:overview-eepdr}, is that every lemma that PDR learns is also a lemma that $\Lambda$-PDR includes in its frames. \yotamsmall{added here}Formally,
to strengthen $\Framepdr_i$, $c$ must block some $\sigma_b \in \bkwrch{n}$ and after $c$ is conjoined to the previous frame $\Framepdr_{i-1}$ we must have $\postimage{\tr}{\Framepdr_{i-1}} \implies c$; \sharon{not clear (do we really need to mention induction?). In particular, I think the indices are buggy. In $\Lambda$-pdr the index it should be $\postimage{\tr}{\Frame_{i-1}} \implies c$, since we're talking about frame $i$ not $i+1$}\yotamsmall{right, fixed the indices after here. I think we need the induction}the latter implies, using an induction hypothesis that $\Frame_{i-1} \implies \Framepdr_{i-1}$, that also $\postimage{\tr}{\Frame_{i-1}} \implies c$. $\Lambda$-PDR conjoins \emph{all} such clauses; thus whenever $c$ is conjoined to $\Framepdr_i$, it is also conjoined to $\Frame_i$.
\end{proof}
}

In other words, PDR's frame also constitute some sort of search in the abstract domain $\madom{\bkwrch{N}}$ (though in a complex manner, refining previous frame etc.), and its frames always generate at least as much overapproximation as $\Lambda$-PDR. Hence, our results that show significant overapproximation in $\Lambda$-PDR translate to PDR as well.

Still, the difference between the algorithms is significant---PDR's frames don't employ the best abstraction in this domain. How does this benefit PDR?
\end{changebar}%
We show two
ways.
First, computing all generalizations may be inefficient. Second, it may not be desirable---it could lead to too precise abstraction and slow convergence.

\para{Inefficient frame size}
Consider a system over $n$ variables $x_1,\ldots,x_n$, with $\Init \ = \ x_1=\ldots=x_n=0$, $\Bad=x_1=\ldots=x_n=1$, and $\tr$ that non-deterministically chooses some $i \neq j$ with $x_i=x_j=0$ and sets $x_i \gets 1$.

We start with the analysis of $\Lambda$-PDR.
In this example, $\bkwrch{k} \ = \ 1\ldots1$ (for every $k$).
We argue that $\Frame_i$ is exactly the set $R_i$ of states reachable in at most $i$ steps, which is the set of states with at most $i$ bits $1$, denoted $\set{\vec{x} \mid \#1(\vec{x}) \leq i}$. %
This can be seen by induction:
initially, this holds for $\Frame_0=\Init=\set{0\ldots0}$.
In each step $\tr(\Frame_i) = R_i \cup \set{\vec{x} \mid \#1(\vec{x})=i+1}$. Then
$\Frame_{i+1}=\mhull{\postimage{\tr}{\Frame_i}}{\bkwrch{k}}=\monox{\postimage{\tr}{\Frame_i}}{1\ldots1} = R_i \cup \monox{\set{\vec{x} \mid \#1(\vec{x})=i+1}}{1\ldots1} = R_{i+1}$, because $\monox{\set{\vec{x} \mid \#1(\vec{x})=i+1}}{1\ldots1}$ adds states that are obtained from a state with $\#1(\vec{x})=i+1$ by flipping $1$'s to $0$'s, resulting in states with smaller values of $\#1(\vec{x})$ that are already included in $R_i$.

Unfortunately, the set $R_{\left\lfloor{n/2}\right\rfloor+1}$ is not expressible in polynomial-size CNF nor DNF.\footnote{
	It is the majority function, which is not in $\mbox{AC}^0$~\cite{DBLP:conf/stoc/Hastad86}, a complexity class that includes poly-size CNF and DNF.
}
This means that some of $\Lambda$-PDR's frames need an exponential number of clauses, and so construct an exponential number of generalizations of the bad state. Even an alternative DNF computation (based on~\Cref{lem:bshouty-mon-mindnf}) would not fare better.

In contrast, $\Framepdr_i$ consists of a single clause blocking the bad state, which is short.

\para{Slow convergence}
\label{sec:slow-convergence}
Consider a counter over $\vec{x}=x_n,x_{n-1},\ldots,x_0$ with $\Init = (\vec{x}=0\ldots0)$, $\Bad=(\vec{x}=1\ldots1)$, and $\tr$ that increments the counter except for when $\vec{x}=1\ldots10$ which skips the bad state and wraps-around to $0\ldots0$.

We start with the analysis of $\Lambda$-PDR. Similar to the previous example, $\bkwrch{k} \ = \ 1\ldots1$ (for every $k$) and $\Frame_i = \mathcal{R}_i$, except that $\mathcal{R}_i$ is now
the set of states $\vec{x} \leq i$, because $\tr(\Frame_i)$ always adds the state $\vec{x}=i+1$, and its $1\ldots1$-monotonization adds only states with smaller values of $\vec{x}$ which are already included in $\mathcal{R}_i$ (the derivation is similar to the previous example).

Therefore, the frames $\Frame_i = \set{\vec{x}: \ \vec{x} \leq i}$ do not converge until $i=2^n-1$, which means that $\Lambda$-PDR converges after an exponential number of frames.\yotamsmall{in this example the frames can be expressed in linear-size CNF}

In contrast, in this example, PDR always converges in a \emph{linear} number of frames.
The proof uses the fact from~\Cref{lem:pdr-also} that the frames of PDR are $(1\ldots1)$-monotone, and that $\Framepdr_i$ is always exactly one clause, because it blocks the single backward reachable state using a single lemma. Since $\Framepdr_i \implies \Framepdr_{i+1}$, and $\Framepdr_i$ is $1...1$-monotone, the clause that is $\Framepdr_{i+1}$ must be a syntactic subset of the clause that is $\Framepdr_i$~\cite{quine1954two}. Until they converge, the difference between two successive frames must be that some literals are omitted from the clause, which can happen at most $n$ times.  
\section{Related Work}
\label{sec:related}
\para{PDR as abstract interpretation}
This work is not the first to study the relation between PDR and abstract interpretation. \citet{DBLP:conf/vmcai/RinetzkyS16} prove that the reachable configurations of PDR are in simulation with the reachable states of a non-standard backward trace semantics. Their work studies standard PDR as non-standard abstract interpretation, whereas we study non-standard PDR as standard abstract interpretation (in a new domain); our domain abstracts the simpler collecting states semantics with standard forward iterations. Our work emphasizes the overapproxmation inherent in the abstraction, where, in particular, the abstraction forces overapproximation in the sequence of frames, whereas %
Rinetzky and Shoham's property-guided Cartesian trace semantics domain is precise enough to express any %
sequence of frames that satisfy properties~\ref{it:frames-start}--\ref{it:frames-end} from~\Cref{sec:overview-frame-props}. In contrast, adding property~\ref{it:frames-mbasis} characterizes $\Lambda$-PDR as Kleene iterations in our $\madom{\bkwrch{k}}$ domain.

\para{Abstract transition systems}
\citet{DBLP:journals/toplas/DamsGG97} construct, from a transition system and a Galois connection, abstract transition systems that preserve safety and other temporal properties. These are defined over a state space of abstract elements (e.g., formulas in the case of a logical domain), forming abstract edges between abstract elements through $\exists\exists$ or $\forall\exists$ relations of original transitions between the concretizations. It is important for our diameter bounds from the DNF representation of the abstract (hyper)transition system that it is defined over the original state space (\Cref{def:abs-tr,def:abs-hypertr}), which is possible due to the special structure of $\madom{B}$ (see~\Cref{thm:absract-reach,thm:hyperabsract-reach}).
In that respect our abstract transition systems are closer to monotonic abstraction in well-structured transition systems by~\citet{DBLP:journals/ijfcs/AbdullaDHR09}, the abstract transition systems for universally-quantified uninterpreted domains by~\citet{DBLP:conf/popl/PadonISKS16}, and the surjective abstraction games of~\citet{DBLP:conf/vmcai/FecherH07}.

\para{Diameter bounds}
Diameter bounds have been studied in the context of completeness thresholds for bounded model checking~\cite{DBLP:conf/tacas/BiereCCZ99,DBLP:conf/vmcai/KroeningS03}.
The recurrence diameter~\cite{DBLP:conf/dac/BiereCCFZ99,DBLP:conf/vmcai/KroeningS03}, the longest loop-free path, was studied as a more easily-computable upper bound on the diameter. In our setting, this measure cannot be reduced by the abstraction, which only adds transitions.
There are also works that encode the completeness threshold assumption as another verification condition~\cite[see][\S IV.D]{DBLP:journals/tcad/DSilvaKW08}.
Another line of work computes diameter bounds by a composition of diameter bounds of subsystems formed by separating dependencies between variables in the system's actions~\cite{DBLP:journals/jar/AbdulazizNG18,DBLP:conf/cav/BaumgartnerKA02,DBLP:conf/ijcai/RintanenG13}. Existing works have considered guarded-update actions, in which variables are either modified to a constant value or remain unchanged; this is not directly applicable to the actions that arise in our abstract transition systems, where monotonization in effect ``havocs'' variables. Havocked variables are different because in a transition, they \emph{can} change, but not \emph{necessarily};
weaker notions of dependence to capture this may be interesting in future work.
We are not aware of a previous application of the $\dnfsize{\tr}$ diameter bound (\Cref{lem:diam-dnf}).
This bound is never worsened by monotonization, as $\dnfsize{\monox{\tr}{\ldots}} \leq \dnfsize{\tr}$~\cite[][and a corollary of~\Cref{lem:bshouty-mon-mindnf}]{DBLP:journals/iandc/Bshouty95}, and can be exponentially smaller, as e.g.\ in~\Cref{sec:overview-diameter-bound}.
The diameter bounds by~\citet{DBLP:conf/popl/KonnovLVW17,DBLP:conf/concur/KonnovVW14} share with our work the motivation of analyzing the diameter of abstractions of the original system. They rely on the special structure of counter abstractions of fault-tolerant distributed systems to apply movers~\cite{DBLP:journals/cacm/Lipton75} and acceleration.

\para{Complexity of invariant inference algorithms}
Houdini~\cite{DBLP:conf/fm/FlanaganL01} infers conjunctive invariants in a linear number of SAT calls, and Dual Houdini likewise for disjunctive invariants~\cite{DBLP:conf/cade/LahiriQ09}.
\citet{DBLP:journals/pacmpl/FeldmanSSW21} analyze the complexity of an interpolation-based inference algorithm based on the fence condition, which we compare with our results in~\Cref{sec:dual-itp-compare}.
The work by~\citet{DBLP:conf/mbmv/SeufertS17} includes a complexity analysis of all executions of PDR on the case of two synchronized $n$-bit counters, where PDR requires an exponential number of SAT calls (this also follows from the fact that the only CNF invariant is exponentially-large) but an enhanced time-dual version of it converges in one frame.
Convergence in one frame is also proved for maximal transitions systems with monotone invariants~\cite{DBLP:journals/pacmpl/FeldmanISS20}.
In~\Cref{sec:slow-convergence}
we go beyond this with an analysis of standard PDR on a simple example where convergence requires multiple frames.
Our analysis of $\Lambda$-PDR centered on the number of frames, not the complexity of constructing them, which is an interesting direction for future work.
Although, in the spirit of~\Cref{sec:abstract-diameter-all}, we can bound $\dnfsize{\Frame_i} \leq \dnfsize{\abs{\tr}} = \dnfsize{\monox{\tr}{\bkcube}}$ (when $\bkcube=\bkwrch{k}$ is a cube), the original $\Lambda$-algorithm's complexity analysis~\cite{DBLP:journals/iandc/Bshouty95} for computing $\monox{\varphi}{\bkcube}$ depends on $\dnfsize{\varphi}$, not $\dnfsize{\monox{\varphi}{\bkcube}}$, which in our setting is the difference between the concrete and the significantly reduced abstract diameter.

\para{The monotone theory in invariant inference}
The monotone theory~\cite{DBLP:journals/iandc/Bshouty95} has been used for other purposes in invariant inference.
\citet{DBLP:journals/mscs/JungKDWY15} use Bshouty's CDNF learning algorithm to infer predicate abstraction invariants, employing over- and under-approximations to resolve membership queries, sometimes relying on random guesses.
\citet{DBLP:journals/pacmpl/FeldmanSSW21} use Bshouty's $\Lambda$-algorithm for provably-efficient inference of an invariant whose $s$-reachable (cf.~\Cref{lem:itp-fence-condition}) and belongs to $\mspan{B}$ when $B$ is known a-priori.
\citet{DBLP:conf/cav/ChenCFTTW10} use the CDNF algorithm for automatic generation of contextual assumptions in assume-guarantee.%

\iflong\else
\pagebreak
\fi
\section{Conclusion}
\label{sec:conclusion}
This work has distilled a previously unknown principle of property-directed reachability.
Through $\Lambda$-PDR and its analysis based on the monotone theory from exact learning, we have shown that PDR overapproximates an abstract interpretation process in a new logical abstract domain.
We have further shown how this abstraction achieves a significantly more effective forward reachability exploration than approaches that use exact post-image computations or bounded unrollings, and how this can partially be explained through the difference between diameter bounds between the original system and its abstraction.

In future work, it will be interesting to understand the mechanisms by which PDR deviates from naive backward reachability, avoiding the pitfall in the other direction, of overapproximating too much.
We hope that this will eventually lead to efficient complexity results for PDR itself. 
\begin{changebar}
It will also be interesting to study variants of PDR that target infinite-state using richer logics beyond propositional logic. Our observation that there is inherent abstraction in PDR due to states it \emph{cannot} exclude from a frame may also be relevant in such settings. This could also involve extensions of the monotone theory to other logics, which to our knowledge have not been attempted.
\end{changebar}

\begin{acks}
\iflong
\else{\small 
\fi
We thank our shepherd and the anonymous reviewers for comments which improved the paper.
We thank Mohammad Abdulaziz, Aman Goel, Alexander Ivrii, Noam Parzanchevski, Hila Peleg, and Noam Rinetzky for insightful discussions and comments.
The research leading to these results has received funding from the
European Research Council under the European Union's Horizon 2020 research and
innovation programme (grant agreement No [759102-SVIS]).
This research was partially supported by the United States-Israel Binational Science Foundation (BSF) grant No.\ 2016260, and the Israeli Science Foundation (ISF) grant No.\ 1810/18.
\iflong
\else
}
\fi
\end{acks} 

\bibliography{refs}

\end{document}